\documentclass{LMCS}

\def\dOi{11(3:5)2015}
\lmcsheading%
{\dOi}
{1--23}
{}
{}
{Oct.~11, 2013}
{Sep.~\phantom01, 2015}
{}

\ACMCCS{[{\bf Theory of computation}]: Logic; Formal languages and
  automata theory---Tree languages} 

\usepackage{latexsym}
\usepackage{amssymb}
\usepackage{amsmath,stmaryrd}
\usepackage{hyperref}
\usepackage{xspace}
\usepackage{tikz}
\usetikzlibrary{trees,arrows,shapes,snakes,fit,shadows,decorations.text,decorations.pathmorphing,calc}
\usepackage[utf8]{inputenc}

\definecolor{my1}{cmyk}{0,.6,0,0}
\definecolor{my2}{cmyk}{.3,.0,.0,.0}

\newcommand\ignore[1]{}

\newcommand{\efgame}{Ehrenfeucht-Fra\"iss\'e\xspace}
\def\hole{\ensuremath{\Box}} 
\newcommand\nat{\ensuremath{\mathbb{N}}\xspace}

\newcommand\purdeux{\usefont{T1}{put}{b}{n}}
\newcommand\ourh{\text{\purdeux h}}
\newcommand\ourv{\text{\purdeux v}}

\newcommand{\A}{\ensuremath{\mathbb{A}}\xspace}

\newcommand{\Bs}{\ensuremath{\mathbb{B}_s}\xspace}
\newcommand{\As}{\ensuremath{\mathbb{A}_s}\xspace}
\newcommand{\Cs}{\ensuremath{\mathbb{C}_s}\xspace}

\renewcommand\next{{\mathbf{X}_{\ourh}}}
\newcommand\previous{\next^{\!\!\!-1}}
\newcommand\forward{{\mathbf{F}_{\ourh}}}
\newcommand\past{\forward^{\!\!\!-1}}
\renewcommand\max{\next,\forward,\previous,\past}
\newcommand\sansx{\forward,\past}
\newcommand\EFX[1]{\ensuremath{\mathbf{EF+F}^{-1}(#1)}\xspace}
\newcommand\EFH{\EFX{\mathbf{S}}}
\newcommand\EFHs{\EFX{\mathbf{S}^{\neq}}}
\newcommand\EFmax{\EFX{\max}}
\newcommand\EFF{\EFX{\sansx}}
\newcommand\EF{\ensuremath{\mathbf{EF+F}^{-1}}\xspace}
\newcommand\EFEX{\ensuremath{\mathbf{EF+EX}^{-1}}\xspace}

\newcommand\Fword{\ensuremath{\text{F+ F}^{-1}}\xspace}

\newcommand\orderv{\ensuremath{<_{\ourv}}\xspace}
\newcommand\orderh{\ensuremath{<_{\ourh}}\xspace}
\newcommand\succv{\ensuremath{Succ_{\ourv}}\xspace}
\newcommand\succh{\ensuremath{Succ_{\ourh}}\xspace}

\newcommand\MSO{\ensuremath{\textup{MSO}(\orderv,\orderh)}\xspace}

\newcommand\FO{\ensuremath{\textup{FO}(\orderv,\orderh)}\xspace}
\newcommand\FOv{\ensuremath{\textup{FO}(\orderv)}\xspace}
\newcommand\FOd{\ensuremath{\textup{FO}^2(\orderv,\orderh)}\xspace}
\newcommand\FOdv{\ensuremath{\textup{FO}^2(\orderv)}\xspace}
\newcommand\FOds{\ensuremath{\textup{FO}^2(s,\orderv)}\xspace}
\newcommand\FOdxs{\ensuremath{\textup{FO}^2(\succh,\orderh,\orderv)}\xspace}
\newcommand\FOdx{\ensuremath{\textup{FO}^2(\succh,\orderh,\succv,\orderv)}\xspace}

\newcommand{\deltad}{\ensuremath{\Delta_{2}(\orderv)}\xspace}
\newcommand{\bsigmau}{\ensuremath{\mathcal{B}\Sigma_{1}(\orderv)}\xspace}

\newcommand\FOdw{\ensuremath{\textup{FO}^2(<)}\xspace}
\newcommand\FOdxw{\ensuremath{\textup{FO}^2(\text{Succ},<)}\xspace}

\newcommand\shal{shallow multicontext\xspace}
\newcommand\shals{shallow multicontexts\xspace}
\newcommand\Shals{Shallow multicontexts\xspace}

\newcommand\mequiv[1]{\ensuremath{\equiv_{#1}}\xspace}
\newcommand\mequivk{\mequiv{k}}

\newcommand\kdef{$k$-definable\xspace}

\newcommand\Xnodes{$X$-nodes\xspace}
\newcommand\Xnode{$X$-node\xspace}
\newcommand\bXnodes{$\overline{X}$-nodes\xspace}
\newcommand\bXnode{$\overline{X}$-node\xspace}

\newcommand\relaxed[1]{$#1$-relaxed\xspace}
\newcommand\relaxedX{\relaxed{X}}

\newcommand\pequiv[1]{\ensuremath{\cong_{#1}}\xspace}
\newcommand\pequivk{\pequiv{k}}

\newcommand\saturated[1]{$#1$-saturated\xspace}

\newcommand\fv{\ensuremath{\mathbf v}\xspace}
\newcommand\fu{\ensuremath{\mathbf u}\xspace}

\newcommand\fW{\ensuremath{\mathbf W}\xspace}

\newcommand\fP{\ensuremath{\mathbf P}\xspace}
\newcommand\fQ{\ensuremath{\mathbf Q}\xspace}
\newcommand\fR{\ensuremath{\mathbf R}\xspace}

\newcommand\fV{\ensuremath{\mathbf V}\xspace}
\newcommand\fU{\ensuremath{\mathbf U}\xspace}

\newcommand\Uc{\ensuremath{\mathcal{U}}\xspace}
\newcommand\Vc{\ensuremath{\mathcal{V}}\xspace}

\newcommand\Gc{\ensuremath{\mathcal{G}}\xspace}

\newcommand\bcC{\ensuremath{\mathfrak C}\xspace}
\newcommand\bcV{\ensuremath{\mathfrak V}\xspace}

\newcommand\bcI{\ensuremath{\mathfrak I}\xspace}

\newcommand\bcT{\ensuremath{\mathfrak T}\xspace}

\newcommand\lesspr{\ensuremath{\sqsubseteq}\xspace}

\DeclareMathOperator{\dclos}{\downarrow}

\newcommand{\uplift}[1]{\ensuremath{\llceil{#1}\rrceil}\xspace}

\newcommand\patts[2]{\ensuremath{(#1,#2)}-patterns\xspace}
\newcommand\patt[2]{\ensuremath{(#1,#2)}-pattern\xspace}

\newcommand\set[1]{\ensuremath{\{#1\}}\xspace}
\theoremstyle{plain}
\newtheorem{theorem}{Theorem}[section]
\newtheorem{corollary}[theorem]{Corollary}
\newtheorem{proposition}[theorem]{Proposition}
\newtheorem{lemma}[theorem]{Lemma}
\newtheorem{claim}[theorem]{Claim} 
\newtheorem{remark}[theorem]{Remark} 
 
\newtheorem{fct}[theorem]{Fact} 

\tikzstyle{lab} = [inner sep=0pt]
\tikzstyle{port}=[draw,rectangle,minimum size=12pt]
\tikzstyle{node}=[draw,circle,minimum size=12pt,inner sep=0pt]
\tikzstyle{dot}=[draw,circle,fill,minimum size=4pt,inner sep=0pt]
\tikzstyle{ars}=[draw,->,line width=2pt]
\tikzstyle{arr} = [line width=1pt, ->]
\tikzstyle{bag}=[inner sep=0pt]
\tikzstyle{bag1}=[draw,circle,inner sep= 0pt]

\title[Deciding definability in \FOd on trees]{Deciding definability in \FOd on trees}
\author[T.~Place]{Thomas Place\rsuper a}
\address{{\lsuper a}Bordeaux University, LaBRI}
\author[L.~Segoufin]{Luc Segoufin\rsuper b}
\address{{\lsuper b}INRIA and ENS Cachan, LSV}

\begin{document}

\begin{abstract}
  We provide a decidable characterization of regular forest languages
  definable in \FOd.  By \FOd we refer to the two variable fragment of first
  order logic built from the descendant relation and the following sibling
  relation. In terms of expressive power it corresponds to a fragment of the
  navigational core of XPath that contains modalities for going up to some
  ancestor, down to some descendant, left to some preceding sibling, and right
  to some following sibling.

  We also show that our techniques can be applied to other two variable
  first-order logics having exactly the same vertical modalities as \FOd but
  having different horizontal modalities.
\end{abstract}

\maketitle

\section{Introduction}

Logics for expressing properties of labeled trees and forests figure
importantly in several different areas of Computer Science. This paper
is about logics on finite unranked trees. The most prominent one is
monadic second-order logic (MSO) as it can be captured by finite tree
automata.  All the logics we consider are less expressive than
monadic second-order logic. Even with these restrictions, this
encompasses a large body of important logics, such as variants of
first-order logic, temporal logics including CTL$^*$ or CTL, as well
as query languages used for XML data. 

This paper is part of a research program devoted to understanding and
comparing the expressive power of such logics.

We say that a logic has a decidable characterization if the following
problem is decidable: given as input a finite tree automaton (or
equivalently a formula of MSO), decide if the recognized language is
definable by the logic in question. Usually a decidable
characterization requires a solid understanding of the expressive
power of the corresponding logic as witnessed by decades of research,
especially for logics for strings. The main open problem in this
research program is to find a decidable characterization of
\FOv, the first-order logic using a binary predicate \orderv for the
ancestor relation.

In this paper we work with unranked ordered trees and by \FO we
refer to the logic that has two binary predicates, one for the
descendant relation, one for the following sibling relation. We
investigate an important fragment of \FO, its two variable restriction
denoted \FOd. This is a robust formalism that, in terms of expressive
power, has an equivalent counterpart in temporal logic. This temporal
counterpart can be seen as the fragment of the navigational core of
XPath that does not use the successor axis~\cite{marx}. More
precisely, it corresponds to the temporal logic \EFF that navigates in
the tree using two ``vertical'' modalities, one for going to some
ancestor node ($\mathbf{F}^{-1}$) and one for going to some descendant
node ($\mathbf{EF}$), and two ``horizontal'' modalities for going to some
following sibling ($\forward$) or some preceding sibling ($\past$).

We provide a characterization of \FOd, or equivalently \EFF, over
unranked ordered trees. We also show that this characterization is
decidable.  Since \FOd can express the fact that a tree has rank $k$
for any fix number $k$, our result also applies to ranked trees. 

Our characterization is stated using closure properties expressed
partly using identities that must be satisfied by the syntactic forest
algebra of the input regular language, and partly via a mechanism that
we call saturation.
 
Here, a forest algebra is essentially a pair of finite semigroups, the
``horizontal'' semigroup for forest types and the ``vertical'' semigroup for
context types, together with an action of contexts over forests. It was
introduced in~\cite{forestalgebra}, using monoids instead of semigroups, and is
a formalism for recognizing forest languages whose expressive power is
equivalent to definability in MSO. Given a formula of MSO, one can compute its
syntactic forest algebra, which recognizes the set of forests satisfying the
formula. Hence any characterization based on a finite set of identities over
the syntactic forest algebra can be tested effectively when given a regular
language as long as each identity can be effectively tested, which will always
be the case in this paper.

The syntactic forest algebra was used successfully for obtaining
decidable characterizations for the classes of tree languages
definable in \EFEX~\cite{EFEX}, \EF~\cite{mikolaj},
\bsigmau~\cite{luclics08} and \deltad~\cite{lucicalp08}. Here \EFEX is
the class of languages definable in a temporal logic that navigates in
trees using two vertical modalities, $\mathbf{EF}$, that we have
already seen before, and $\mathbf{EX}$, which goes to a child of the
current  node. \EF is the class of languages definable in \EFF without
using the horizontal modalities. \bsigmau stands for the class of
languages definable by a Boolean combination of existential formulas
of \FOv and \deltad is the class of languages definable in \FOv by both
a formula of the form $\exists^*\forall^*$ and a formula of the form
$\forall^*\exists^*$.

Over strings, the logics induced by $\Delta_2(<)$, \FOdw and
\Fword, have exactly the same expressive
power~\cite{fodeux-UTL,fodeux}. But over trees this is not the
case. For instance \EF is closed under bisimulation while the other
two are not. While decidable characterizations were obtained for \EF
and $\deltad$~\cite{mikolaj,lucicalp08}, the important case of \FOd
was still missing and is solved in this paper. 

Over strings, a regular language is definable in \FOdw iff its
syntactic semigroup satisfies an identity that can be effectively
tested~\cite{fodeux}. Not surprisingly our first set of identities
requires that the horizontal and vertical semigroups of the syntactic
forest algebra both satisfy this identity. Our extra property is more
complex and mixes at the same time the vertical and horizontal
navigational power of \FOd. We call it \emph{closure under saturation}.

It is immediate from the string case that being definable in \FOd
implies that the vertical and horizontal semigroups of the syntactic
forest algebra satisfy the required identity. That closure under
saturation is also necessary is proved via a classical, but tedious,
Ehrenfeucht-Fraïssé game argument. As usual in this area, the
difficulty is to show that the closure conditions are sufficient. In
order to do so, as it is standard when dealing with \FOd (see
e.g.~\cite{mikolaj,lucicalp08,fodeux}), we introduce Green-like
relations for comparing elements of the syntactic algebra. However, in
our case, we parametrize these relations with a set of forbidden
patterns: the contexts authorized for going from one type to another
type cannot use any of the forbidden pattern. We are then able to
perform an induction using this set of forbidden patterns, thus
refining our comparison relations more and more until they become
trivial.

Our proof has many similarities with the one of Bojańczyk that
provides a decidable characterization for the logic \EF~\cite{mikolaj}
and we reuse several ideas developed in this paper. However it departs
from it in many essential ways. First of all the closure under
bisimulation of \EF was used in~\cite{mikolaj} in an essential way in
order to compute a subalgebra and perform inductions on the size of
the algebra. Moreover, because \EF does not have horizontal
navigation, Bojańczyk was able to isolate certain labels and then also
perform inductions on the size of the alphabet. It is the combination
of the induction on the size of the alphabet and on the size of the
algebra that gave an elegant proof of the correctness of the
identities for \EF given in~\cite{mikolaj}. The logic \FOd is no
longer closed under bisimulation and we were not able to perform an
induction on the algebra. Moreover because our logic has horizontal
navigation, it is no longer possible to isolate the label of a node
from the labels of its siblings, hence it is no longer possible to
perform an induction on the size of the alphabet. In order to overcome
these problems our proof replaces the inductions used in~\cite{mikolaj}
by an induction on the set of forbidden patterns. This makes the two
proofs technically fairly different.

It turns out that our proof technique applies to various horizontal
modalities.  In the final section of the paper we show how to adapt
the characterization obtained for \mbox{\FOd} in order to obtain
characterizations for \EFmax, \EFH and \EFHs, where
$\next$, $\previous$, $\mathbf{S}$ and $\mathbf{S}^{\neq}$ are
horizontal navigational modalities moving respectively to the next
sibling, previous sibling, an arbitrary sibling including the current
node, or an arbitrary different sibling excluding the current node.

\paragraph*{\bf Other related work}

Our characterization is essentially given using forest algebras. There
exist several other formalisms that were used for providing
characterizations of logical fragments of MSO~(see e.g.
\cite{BS09,PS09,Wil96,preclones}). It is not clear however how to use
these formalisms in order to provide a characterization of \FOd.

There exist decidable characterizations of \EF and \deltad over trees
of bounded rank~\cite{place-csl08}. But, as these logics cannot
express the fact that a tree is binary, the unranked and bounded rank
characterizations are different. As mentioned above, we don't
have this problem with \FOd.
 
\paragraph*{\bf Organization of the paper}

We first provide the preliminary definitions in Section~\ref{prelim}.
The main definitions and their basic properties are described in
Section~\ref{sec-reach}.  Our characterization is stated in
Section~\ref{sec:carac}.  That our properties are necessary for
being definable in \FOd is proved in Section~\ref{sec-correc}.  We
give the proof that our characterization for \mbox{\FOd} is sufficient in
Section~\ref{main-proof}. Decidability of closure under saturation is
not immediate and Section~\ref{more} is devoted to this issue. In
Section~\ref{others} we show how to adapt the arguments in order to
characterize several other horizontal navigation modalities. 

Note that Section~\ref{sec-correc}, Section~\ref{main-proof} and
Section~\ref{more} can be read in an arbitrary order.

\medskip

This paper is the journal version of~\cite{PlaceS10}. From the conference
version the statement of the characterization has been slightly changed  and
the proofs have been significantly modified in order to simplify the presentation.


\section{Preliminaries}\label{prelim}

We work with finite unranked ordered trees and forests labeled over a
finite alphabet \A. A finite alphabet is a pair $\A = (A,B)$ where $A$
and $B$ are finite sets of labels. We use $A$ to label leaves and $B$
to label inner nodes. Making the distinction between leaves and inner
nodes labels makes our presentation slightly simpler without harming
the generality of our results. Given a finite alphabet $\A = (A,B)$,
trees and forests are defined inductively as follows: for any $a \in
A$, $a$ is a tree. If $t_1,\cdots,t_k$ is a finite non-empty sequence
of trees then $t_1 + \cdots + t_k$ is a forest. If $s$ is a forest and
$b \in B$, then $b(s)$ is a tree. Notice that we do not consider empty
trees nor empty forests. A set of trees (forests) over a finite
alphabet \A is called a tree language (forest language).

We use standard terminology for trees and forests defining nodes,
ancestors, descendants, following and preceding siblings. We write $x
\orderv y$ to say that $x$ is a strict ancestor of $y$ or,
equivalently, that $y$ is a strict descendant of $x$. We write $x
\orderh y$ to say that $x$ is a strict preceding sibling of $y$ or,
equivalently, that $y$ is a strict following sibling of $x$. 

A context is a forest over $(A\cup\set{\hole},B)$ with a single leaf
of label $\hole$ that cannot be a root and that has no sibling. This
distinguished node is called \emph{the port} of the context (see
Figure~\ref{fig-exem-cont}). This definition is not standard as
usually contexts are defined without the ``no sibling'' restriction
but it is important for this paper to work with this non-standard
definition. If $c$ is a context, the path in $c$ containing all the
ancestors of its port is called the \emph{backbone} of $c$.

A context $c$ can be composed with another context $c'$ or with a
forest $s$ in the obvious way by substituting $c'$ or $s$ in place of
the port of $c$. This composition yields either the context $cc'$ or
the forest $cs$. 

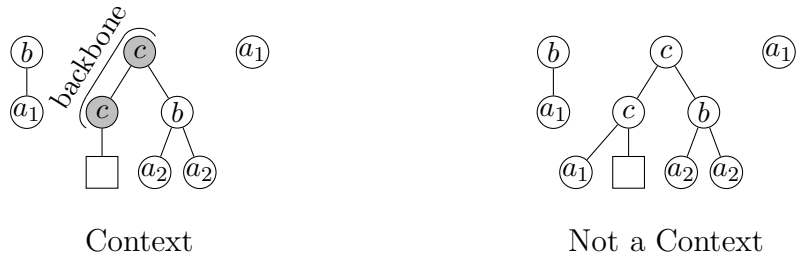
\begin{figure}[h]
\begin{center}
\begin{tikzpicture}
\node[node] (r1) at (-1.5,0) {$b$};
\node[node,fill=gray!50] (r2) at (0,0) {$c$};
\node[node] (r3) at (1.5,0) {$a_1$};

\node[node] (s1) at (-1.5,-0.8) {$a_1$};
\node[node,fill=gray!50] (s2) at (-0.5,-0.8) {$c$};
\node[node] (s3) at (0.5,-0.8) {$b$};

\node[port] (t1) at (-0.5,-1.6) {};
\node[node] (t2) at (0.2,-1.6) {$a_2$};
\node[node] (t3) at (0.8,-1.6) {$a_2$};

\draw (r1) -- (s1);

\draw (r2) -- (s2);
\draw (r2) -- (s3);

\draw (s2) -- (t1);

\draw (s3) -- (t2);
\draw (s3) -- (t3);

\coordinate (b1) at ($(r2)+(-0.21,0.21)$);
\coordinate (b2) at ($(s2)+(-0.21,0.21)$);
\coordinate (b3) at ($(r2)+(0.21,0.21)$);
\coordinate (b4) at ($(s2)+(-0.21,-0.21)$);

\draw (b1) to node[above,sloped] {backbone} (b2);
\draw (b1) to [out=60,in=120] (b3);
\draw (b2) to [out=240,in=150] (b4);

\node[lab] at (0,-2.5) {\large Context};

\begin{scope}[xshift=7cm]
\node[node] (r1) at (-1.5,0) {$b$};
\node[node] (r2) at (0,0) {$c$};
\node[node] (r3) at (1.5,0) {$a_1$};

\node[node] (s1) at (-1.5,-0.8) {$a_1$};
\node[node] (s2) at (-0.5,-0.8) {$c$};
\node[node] (s3) at (0.5,-0.8) {$b$};

\node[port] (t1) at (-0.5,-1.6) {};
\node[node] (tb) at (-1.2,-1.6) {$a_1$};
\node[node] (t2) at (0.2,-1.6) {$a_2$};
\node[node] (t3) at (0.8,-1.6) {$a_2$};

\draw (r1) -- (s1);

\draw (r2) -- (s2);
\draw (r2) -- (s3);

\draw (s2) -- (t1);
\draw (s2) -- (tb);

\draw (s3) -- (t2);
\draw (s3) -- (t3);

\node[lab] at (0,-2.5) {\large Not a Context};
\end{scope}
\end{tikzpicture}
\end{center}
\caption{Illustration of the notion of context. The squared nodes represent
  ports. The right part is not a context because the port has a sibling.}\label{fig-exem-cont}
\end{figure}

If $x$ is a node of a forest then the \emph{subtree at $x$} is the
tree rooted at $x$.  The \emph{subforest of $x$} is the forest
consisting of all the subtrees that are rooted at siblings of $x$
(including $x$). Finally, if $x$ is not a leaf, the \emph{subforest
  below $x$} is the forest consisting of all the subtrees that are
rooted at children of $x$, see Figure~\ref{fig-exem-sub}. Notice
that from the definitions if follows that $s$ is a subforest of a
forest $t$ iff there exists a context $c$ such that $t = cs$. In
particular if we consider the forest $b(a_1 +a _2 + a_3)$,
$a_1+a_2+a_3$ is a subforest but not $a_1+a_2$.

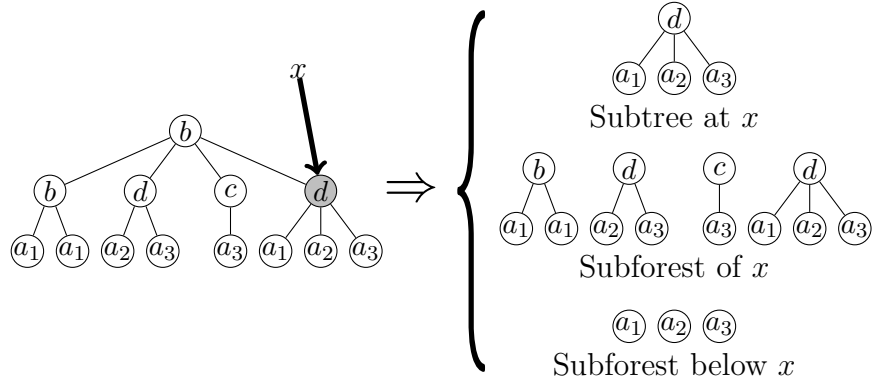
\begin{figure}[h]
\begin{center}
\begin{tikzpicture}

\node[node] (r2) at (1.0,0) {$b$};

\node[node] (s2) at (-0.8,-0.8) {$b$};
\node[node] (s3) at (0.4,-0.8) {$d$};
\node[node] (s4) at (1.6,-0.8) {$c$};
\node[node,fill=gray!50] (s5) at (2.8,-0.8) {$d$};

\node[node] (t1) at (-1.1,-1.6) {$a_1$};
\node[node] (t2) at (-0.5,-1.6) {$a_1$};

\node[node] (t3) at (0.1,-1.6) {$a_2$};
\node[node] (t32) at (0.7,-1.6) {$a_3$};

\node[node] (t33) at (1.6,-1.6) {$a_3$};

\node[node] (t4) at (2.2,-1.6) {$a_1$};
\node[node] (t5) at (2.8,-1.6) {$a_2$};
\node[node] (t6) at (3.4,-1.6) {$a_3$};

\draw (r2) -- (s2);
\draw (r2) -- (s3);
\draw (r2) -- (s4);
\draw (r2) -- (s5);

\draw (s2) -- (t1);
\draw (s2) -- (t2);

\draw (s3) -- (t3);
\draw (s3) -- (t32);

\draw (s4) -- (t33);

\draw (s5) -- (t4);
\draw (s5) -- (t5);
\draw (s5) -- (t6);

\node[lab] (x) at (2.5,0.8) {\large $x$};

\draw[ars] (x) -- (s5);

\node[lab] at (4.0,-0.8) {\huge $\Rightarrow$};

\node[lab,yscale=12.5,xscale=2.7] at (4.8,-0.8) {$\{$};

\begin{scope}[xshift=4.7cm,yshift=2.3cm]
\node[lab] at (2.8,-2.1) {\large Subtree at $x$};
\node[node] (s5) at (2.8,-0.8) {$d$};
\node[node] (t4) at (2.2,-1.6) {$a_1$};
\node[node] (t5) at (2.8,-1.6) {$a_2$};
\node[node] (t6) at (3.4,-1.6) {$a_3$};
\draw (s5) -- (t4);
\draw (s5) -- (t5);
\draw (s5) -- (t6);
\end{scope}

\begin{scope}[xshift=6.5cm,yshift=0.3cm]
\node[lab] at (1.0,-2.1) {\large Subforest of $x$};
\node[node] (s2) at (-0.8,-0.8) {$b$};
\node[node] (s3) at (0.4,-0.8) {$d$};
\node[node] (s4) at (1.6,-0.8) {$c$};
\node[node] (s5) at (2.8,-0.8) {$d$};

\node[node] (t1) at (-1.1,-1.6) {$a_1$};
\node[node] (t2) at (-0.5,-1.6) {$a_1$};

\node[node] (t3) at (0.1,-1.6) {$a_2$};
\node[node] (t32) at (0.7,-1.6) {$a_3$};

\node[node] (t33) at (1.6,-1.6) {$a_3$};

\node[node] (t4) at (2.2,-1.6) {$a_1$};
\node[node] (t5) at (2.8,-1.6) {$a_2$};
\node[node] (t6) at (3.4,-1.6) {$a_3$};

\draw (s2) -- (t1);
\draw (s2) -- (t2);

\draw (s3) -- (t3);
\draw (s3) -- (t32);

\draw (s4) -- (t33);

\draw (s5) -- (t4);
\draw (s5) -- (t5);
\draw (s5) -- (t6);
\end{scope}

\begin{scope}[xshift=6.5cm,yshift=-1.0cm]
\node[lab] at (1.0,-2.1) {\large Subforest below $x$};

\node[node] (t4) at (0.4,-1.6) {$a_1$};
\node[node] (t5) at (1.0,-1.6) {$a_2$};
\node[node] (t6) at (1.6,-1.6) {$a_3$};

\end{scope}

\end{tikzpicture}
\end{center}
\caption{Illustration of the notion of subtrees and subforests}\label{fig-exem-sub}
\end{figure}

\section{The Logic \FOd}\label{sec:logic}

\subsection{Definition}
A forest can be seen as a relational structure. The domain of the
structure is the set of nodes. The signature contains a unary
predicate $P_a$ for each symbol $a\in\A$ plus the binary predicates
\orderv and \orderh. By \MSO we denote the monadic second order logic
over this relational signature.  We use the classical semantics for
\MSO and write $s\models \phi(\bar u)$ if the formula $\phi$ is true
on $s$ when interpreting its free variables with the corresponding
nodes of $\bar u$.  As usual, each sentence $\varphi$ of \MSO 
defines a forest language $L_\varphi=\set{s \mid s\models \varphi}$. A
language defined in \MSO is called a \emph{regular language}. As usual
regular languages form a robust class of languages and there is a
matching notion of unranked ordered forests automata (see for
instance~\cite[chapter 8]{tata}). We will see in Section~\ref{forest}
a corresponding notion of recognizability using forest algebras.

The logic of interest for this paper is \FOd, the two variable
restriction of the first-order fragment of \MSO.

In terms of expressive power, \FOd is equivalent to \EFF, a temporal
logic that we now describe. Essentially, \EFF is the restriction of
the navigational core of XPath without the {\sc child}, {\sc parent},
{\sc next-sibling} and {\sc previous-sibling} predicates. It is
defined using the following grammar:
\begin{equation*}
\varphi :: \A ~|~ \varphi \vee \varphi ~|~ \varphi \wedge \varphi ~|~
\neg\varphi ~|~ \mathbf{EF}\varphi ~|~ \mathbf{F}^{-1}\varphi ~|~ 
\mathbf{F}_{\ourh} \varphi ~|~ \mathbf{F}_{\ourh}^{-1} \varphi
\end{equation*}
We use the classical semantics for this logic which defines when a
formula holds at a node $x$ of a forest $s$. In particular,
$\mathbf{EF}\varphi$ holds at $x$ if there is a strict descendant of
$x$ where $\varphi$ holds, $\mathbf{F}^{-1}\varphi$ holds at $x$ if
there is a strict ancestor of $x$ where $\varphi$ holds,
$\mathbf{F}_{\ourh}\varphi$ holds at $x$ if $\varphi$ holds at some
strict following sibling of $x$, and $\mathbf{F}_{\ourh}^{-1}\varphi$
holds at $x$ if $\varphi$ holds at some strict preceding sibling of
$x$. We then say that a forest $s$ satisfies a formula $\phi$ if
$\phi$ holds at the root of the first tree of $s$. The following
result is immediate from~\cite{marx}:

\begin{theorem}\label{thm-fod-eff}
For any \FOd formula $\phi(x)$, there exists a \EFF formula $\varphi$
holding true on the same set of nodes for every forests. In particular
a forest language is definable in \FOd iff it is definable in \EFF.
\end{theorem}

\bigskip

We aim at providing a decidable characterization of regular forest
languages definable in \FOd. This means finding an algorithm that
decides whether or not a given regular forest language is definable in
\FOd. 

Note that \FOd is expressive enough to test whether a forest is a tree
and, for each $k$ whether it has rank $k$. Hence any result concerning
forest languages definable in \FOd also applies to tree languages
definable in \FOd and covers the ranked and unranked cases.

\medskip

We shall mostly adopt the \EFF point of view as it is useful when
considering other horizontal modalities or when making comparisons
with the decision algorithm obtained for \EF in~\cite{mikolaj}.

\subsection{\efgame Games}\label{games}
As usual definability in \FOd corresponds to winning strategies in a
\efgame game that we briefly describe here. We adopt the \EFF point of
view as the corresponding game is slightly simpler. Its definition is
standard.  
 
There are two players, Duplicator and Spoiler. The board consists in
two forests and the number $k$ of rounds is fixed in advance. At any
time during the game there is one pebble placed on a node of one
forest and one pebble placed on a node of the other forest and both
nodes have the same label. If the initial position is not specified,
the game starts with the two pebbles placed on the root of the
leftmost tree in each forest. Each round starts with Spoiler moving
one of the pebbles inside its forest, either to some ancestor of its
current position, or to some descendant or to some left or right
sibling. Duplicator must respond by moving the other pebble inside the
other forest in the same direction to a node of the same label. If
during a round Duplicator cannot move then Spoiler wins the game. If
Duplicator was able to respond to all the moves of Spoiler then she
wins the game. Winning strategies are defined as usual. If Duplicator
has a winning strategy for the game played on the forests $s,t$ then
we say that $s$ and $t$ are $k$-equivalent.

The following results on games are classical and simple to prove. 

\begin{lemma}[Folklore]\label{lemma-ef-game}$\quad$
\begin{enumerate}
\item For every $k$, $k$-equivalence is an equivalence relation of
  finite index. 
\item For every $k$, each class of the $k$-equivalence relation is
  definable by a sentence of \EFF such that the nesting depth of its
  navigational modalities is bounded by $k$.
\item {\sloppy For every $k$, if $s$ and $t$ are $k$-equivalent then they
  satisfy the same sentences of\break \EFF such that the nesting depth of
  their navigational modalities is bounded by $k$.}
\end{enumerate}
\end{lemma}

\noindent When played on words instead of forests, the game is the same except
that now Spoiler can move either to a previous or to a following
position. The results are identical after replacing \FOd with \FOdw,
the two variable first-order logic on strings, using the predicate $<$
for the following position relation.

\subsection{Antichain Composition Principle} As mentioned in the introduction,
we use induction to prove that if $L$ satisfies the characterization then we
can construct a \FOd formula for $L$. At each step in
this construction, we prove that $L$ can be defined as the composition of
simpler languages such that a formula for $L$ can be constructed from formulas
defining the simpler languages. This is what we do with the following simple
composition lemma, essentially adapted from~\cite{mikolaj} and using the same
terminology. 

A formula of \FOd with one free variable is called \emph{antichain} if in every
forest, the set of nodes where it holds forms an antichain, i.e. a set (not
necessarily maximal) of nodes pairwise incomparable with respect to the
descendant relation. This is a semantic property that may not be apparent just
by looking at the syntax of the formula. A typical antichain formula selects in
a forest the set of nodes of label $b \in B$ that have no ancestor of label
$b$.

Given (i) an antichain formula $\varphi$, (ii) disjoint forest
languages $L_1,\cdots,L_n$ and (iii) labels $a_1, \cdots,a_n
\in A$ and (iv) a forest $s$, we define the forest $s' =
s[(L_1,\varphi)  \rightarrow a_1, \cdots, (L_n, \varphi) \rightarrow
a_n]$ as follows. For each node $x$ of $s$ such that $s \models
\varphi(x)$, we determine the unique $i$ such that the forest language
$L_i$ contains the \emph{subforest below $x$}. If such an $i$ exists,
we remove the whole subforest below $x$, and replace it by a leaf of
label $a_i$. Since $\varphi$ is antichain, this can be done
simultaneously for all $x$. Note that the formula $\varphi$ may also
depend on ancestors of $x$, while the languages $L_i$ only talk about
the subforest below $x$.

The composition method that we will use is summarized in the the following
lemma:

\begin{lemma}\label{lemma-ACL}[Antichain Composition Lemma]
Let $\varphi$ , $L_1,\dots,L_n$ and $a_1,\dots,a_n$ be as above. 
If $L_1, \dots, L_n$ and $K$ are languages definable in
\FOd, then so is
\begin{equation*}
L = \set{t ~\mid~ t[(L_1,\varphi) \rightarrow a_1, 
\cdots,(L_n,\varphi) \rightarrow a_n] \in K}.
\end{equation*}
\end{lemma}

This lemma follows from a simple \efgame game argument. Using the \EFF point of
view we can also construct a formula defining $L$.
The formula for $L$ is obtained from the one for $K$ by restricting all
navigation to nodes that are not descendants of nodes satisfying $\varphi$ and
by replacing each test that a label is $a_i$ by the formula for $L_i$ where all
navigations are now restricted to nodes that are descendants of a node that
satisfies $\varphi$. The fact that $\varphi$ is antichain makes this
construction sound. The details are simple and are omitted here as they
paraphrase those given in~\cite{mikolaj} for \EF.

The inductive step of our proof consists in exhibiting $L_1, \dots, L_n$ and
$K$, together with an antichain formula $\varphi$ such that $L = \{t \mid
t[(L_1,\varphi) \rightarrow a_1, \cdots,(L_n,\varphi) \rightarrow a_n] \in K\}$
and $K,L_1, \dots, L_n$ have smaller inductive parameters than
$L$. In~\cite{mikolaj} the antichain formula is of the form: ``select the set
of nodes of label $b \in B$ that have no ancestor of label $b$.'' Observe that
such a formula allows us to use the size of $B$ as an induction parameter as $K$
does not contain the label $b$. In our case, we replace $B$ by sibling patterns
that we will define in Section~\ref{sec-reach}.

\section{Forest algebras}\label{forest}

\newcommand\one{\textup{1}}
\newcommand\mydelta{\vartriangle}

A key ingredient in our characterization is based on syntactic forest
algebras. \emph{Forest algebras} were introduced by Boja{\'n}czyk and
Walukiewicz as an algebraic formalism for studying regular forest
languages~\cite{forestalgebra}. We work with the semigroup variant of
forest algebras. Moreover we require that the port of each context has
no sibling. These restrictions are necessary as, without them, the
languages definable in \FOd would not form a variety, i.e.  would not
be characterizable by its syntactic forest algebra only.

We give a brief summary of the definition of forest algebras and of
their important properties.  More details can be found
in~\cite{forestalgebra}.  A forest algebra consists of a pair $(H,V)$
of finite semigroups, subject to some additional requirements, which
we describe below. We write the operation in $V$ multiplicatively and
the operation in $H$ additively, although $H$ is not assumed to be
commutative.
 
We require that $V$ acts on the left of $H$.  That is, there is a map 
$$(h,v)\in H\times V\mapsto vh\in H$$ such that $$w(vh)=(wv)h$$ for
all $h\in H$ and $v,w\in V$. We further require that for every $g\in
H$ and $v \in V$, $V$ contains elements $(v+g)$ and $(g+v)$ such
that $$(v+g)h ~~~~=~~~~ vh+g, ~\text{\hspace{2cm}}~ (g+v)h ~~~~=~~~~ g+vh$$ for all $h\in H.$ 

Let $\A = (A,B)$ be a finite alphabet. The \emph{free forest algebra}
on \A, denoted by $\A^\mydelta$, is the pair of semigroups $(H_\A,V_\A)$
where $H_\A$ is the set of forests over \A equipped with the $+$
operation and $V_\A$ the set of contexts equipped with the composition
operation, together with the natural action. One can verify that this
action turns $\A^\mydelta$ into a forest algebra.
 
A morphism $\alpha:(H_1,V_1)\to (H_2,V_2)$ of forest algebras is a
pair $(\gamma,\delta)$ of semigroup morphisms $\gamma: H_1 \rightarrow
H_2$, $\delta: V_1 \rightarrow V_2$ such that
$\gamma(vh)=\delta(v)\gamma(h)$ for all $h\in H,$ $v\in V.$ However,
we will abuse notation slightly and denote both component maps by
$\alpha$. 

We say that a forest algebra $(H,V)$ {\it recognizes} a forest
language $L$ if there is a morphism $\alpha:\A^{\mydelta}\to (H,V)$ and
a subset $X$ of $H$ such that $L=\alpha^{-1}(X).$ We also say that the
morphism $\alpha$ recognizes $L$. It is easy to show that a forest
language is regular if and only if it is recognized by a finite
forest algebra. 
 
Consider some forest language $L$ over an alphabet \A. We define an
equivalence relation $\sim_L$ over contexts and over forests. Given
two forests $t_1,t_2$, we say that $t_1 \sim_L t_2$ iff for any two
forests $s,s'$ and any context $c$, $c(s + t_1 +s') \in L$ iff $c(s +
t_2 +s') \in L$.  Given two contexts $c_1,c_2$ we say that $c_1 \sim_L
c_2$ iff for any forest $s$, $c_1s \sim_L c_2s$.  This equivalence is
a congruence of forest algebras that is of finite index iff $L$ is
regular. The quotient of $\A^\mydelta$ by this congruence yields a
forest algebra recognizing $L$ which we call the \emph{syntactic
  forest algebra} of $L$. The mapping sending a forest or a context to
its equivalence class in the syntactic forest algebra, denoted
$\alpha_L$, is a morphism called the \emph{syntactic morphism} of
$L$.

It is also important to know that given an \MSO sentence $\phi$, the 
syntactic forest algebra of $L_\phi$ and the syntactic morphism
$\alpha_\phi$ can be computed from $\phi$. 

\paragraph*{\bf Idempotents} It follows from standard arguments of semigroup
theory that given any finite semigroup $S$, there exists a
number $\omega(S)$ (denoted by $\omega$ when $S$ is understood from
the context) such that for each element $x$ of $S$, $x^\omega$ is an
idempotent: $x^\omega=x^\omega x^\omega$. Given a forest algebra
$(H,V)$ we will denote by $\omega(H,V)$ (or just $\omega$ when $(H,V)$  
is understood from the context) the product of $\omega(H)$ and
$\omega(V)$ and for any element $u \in V$ and $g \in H$ we will write
$u^\omega$ and $\omega g$ for the corresponding idempotents. 

\medskip
Finally, given a semigroup $S$ we will denote by $S^{\one}$ the monoid
formed from $S$ by adding a neutral element.

\paragraph*{\bf Leaf Surjective Morphisms} Let $\A = (A,B)$ be a
finite alphabet and let $\alpha: \A^\mydelta \rightarrow (H,V)$ be a
morphism into a finite forest algebra $(H,V)$. We say that $\alpha$ is
\emph{leaf surjective} iff for any $h \in H$, there exists $a \in A$
such that $\alpha(a) = h$.

Observe that given any morphism $\alpha: (A,B)^\mydelta \rightarrow
(H,V)$, one can construct a leaf surjective one $\beta:
(A \cup H,B)^\mydelta \rightarrow (H,V)$ by extending $\alpha$ in the
obvious way. We call $\beta$ the \emph{leaf completion} of
$\alpha$. 

\begin{lemma} \label{lem:leafsurj}
Let $\alpha: \A^\mydelta \rightarrow (H,V)$ be a morphism into a finite
forest algebra, $\beta$ be its leaf completion and $h \in H$ be such
that $\beta^{-1}(h)$ is \FOd definable. Then $\alpha^{-1}(h)$ is \FOd
definable.
\end{lemma}

\begin{proof}
A forest in $\alpha^{-1}(h)$ is a forest in $\beta^{-1}(h)$ that
contains no leaf with label in $H$. Therefore, one can construct an \FOd
formula for $\alpha^{-1}(h)$ from a formula for $\beta^{-1}(h)$.
\end{proof}

Working with leaf surjective morphisms will be convenient for us. Typically
when applying the Antichain Composition Lemma we will construct $K$ from $L$ by
replacing some subforests of $L$ by leaf nodes with the same forest type. It is
therefore important that such nodes exist.

\paragraph*{\bf The string case}
A reason for using syntactic forest algebras is that the same problem for
strings was solved using syntactic semigroups. In the string case there is only
one linear order and the corresponding logic is denoted by \FOdw. Recall that
the syntactic semigroup (monoid) of a regular string language is the transition
semigroup (monoid) of its minimal deterministic automata. It is therefore
computable from any reasonable presentation of the regular string language. The
following characterization was actually stated using syntactic monoids, but it
is equivalent to this statement\footnote{We are actually not using the
  identity of~\cite{fodeux}. Ours can easily seen to be equivalent
  to it. This is a folklore result. A proof can be found in~\cite{kufDiss06}.}\label{foot}.

\begin{theorem}[\cite{fodeux}]\label{thm-word-fod}
 A regular string language is definable in \FOdw iff its syntactic semigroup
 $S$ satisfies for all $u,v \in S$:
\begin{equation*} (uv)^\omega v (uv)^\omega =(uv)^\omega
  \end{equation*} 
\end{theorem}

Unfortunately in the case of forest languages we were not able to
state our characterization using only the syntactic forest algebra of
the input regular language. We will need an extra ingredient that we
call \emph{saturation}.

\section{\Shals} \label{sec-reach}

In this section, we define \shals which represent sequences of siblings. We
will often manipulate \shals modulo \FOd definability. This is captured by an
equivalence relation on \shals that we also define in this section.

This notion of \shal is central for this paper as we will use it not only as a
parameter in the inductive argument but also to define the notion of
\emph{saturation} that we use in our characterization of \FOd.

Set $\A = (A,B)$ as a finite alphabet. All definitions are parametrized by a
morphism $\alpha: \A^\mydelta \rightarrow (H,V)$. Note that while the definitions
make sense for any morphism $\alpha$, they are designed to be used with a leaf
surjective one.  Given a forest $s$ (a context $p$) we refer to its image under
$\alpha$ as the \emph{forest type} of $s$ (the \emph{context type} of $p$).

\subsection{\Shals} 
A multicontext is defined in the same way as a context but with several ports, possibly
none. The \emph{arity} of a multicontext is the number of its ports, possibly
0. Note that as our forests are ordered, each multicontext implicitly defines a
linear order on its ports. A multicontext is said to be \emph{shallow} if each
of its trees is either a single node $a\in A$, a single inner node with a port
below, $b(\hole)$, or a tree of the form $b(a)$ where $b\in B$ and $a\in A$ (see
Figure~\ref{fig-exem-shal}). 

For technical reasons we do not consider forests with a single tree of the form
$a\in A$ as a \shal. Observe that in our definition of \shal we include trees
of the form $b(a)$. This is because, as mentioned earlier, we will often
replace a subforest by a node having the same type and therefore it is
convenient to immediately have access to this type by looking at the label
of that node.

\begin{figure}[h]
\begin{center}
\begin{tikzpicture}

\node[node] (s1) at (0.0,-0.8) {$b$};
\node[node] (s2) at (0.9,-0.8) {$b$};
\node[node] (s3) at (1.8,-0.8) {$d$};
\node[node] (s4) at (2.7,-0.8) {$a_2$};
\node[node] (s5) at (3.6,-0.8) {$d$};
\node[node] (s6) at (4.5,-0.8) {$c$};
\node[node] (s7) at (5.4,-0.8) {$a_1$};

\node[lab] at (0.45,-0.8) {$+$};
\node[lab] at (1.35,-0.8) {$+$};
\node[lab] at (2.25,-0.8) {$+$};
\node[lab] at (3.15,-0.8) {$+$};
\node[lab] at (4.05,-0.8) {$+$};
\node[lab] at (4.95,-0.8) {$+$};

\node[node] (t1) at (0.0,-1.6) {$a_1$};
\node[port] (t2) at (0.9,-1.6) {};
\node[node] (t3) at (1.8,-1.6) {$a_2$};
\node[port] (t5) at (3.6,-1.6) {};
\node[port] (t6) at (4.5,-1.6) {};

\draw (s1) -- (t1);
\draw (s2) -- (t2);
\draw (s3) -- (t3);
\draw (s5) -- (t5);
\draw (s6) -- (t6);

\end{tikzpicture}
\end{center}
\caption{Illustration of a shallow multicontext of arity 3}\label{fig-exem-shal}
\end{figure}
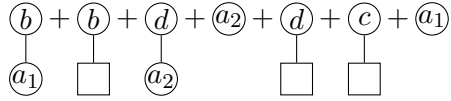

Let $x$ be a node of a forest $s$. Let $t = t_1 + \cdots + t_\ell$ be the
subforest of $x$, composed of $\ell$ trees. The \emph{\shal of $x$ in $s$} is
the sequence $p_1+\cdots+p_\ell$ such that $p_i := a$ if $t_i=a \in A$, $p_i: =
b(a)$ if $t_i=b(a)$, and $p_i := b(\hole)$ if $t_i=b(s')$, where $a\in A,b \in
B$ and $s' \notin A$. A \shal $p$ \emph{occurs} in a forest $s$ iff there
exists a node $x$ of $s$ such that $p$ is the \shal of $x$ in $s$. In the rest
of the paper, a node $x$ will almost always be considered together with
the \shal $p$ occurring at $x$. For this reason we will write
``let $(p,x)$ be a node of a tree $t$'' when $x$ is a node of $t$ and $p$ is
the \shal at $x$. Similarly, if $P$ is a set of \shals, we will write ``let
$(p,x) \in P$'' when $p \in P$ and $x$ is a node of $p$.

Given a \shal $p$ of arity $n$ and a sequence $T$ of $n$ forests,
$p[T]$ denotes the forest obtained after placing the $i^{th}$ forest
of $T$ at the $i^{th}$ port of $p$. Moreover, given a node $x$ of $p$
whose unique child is a port (i.e. $x$ is the root of a tree of the form
$b(\hole)$ within $p$) and a sequence $T$ of $n-1$ forests, $p[T,x]$ denotes
the context obtained as above but leaving the subtree at $x$ unchanged.

\paragraph*{\bf $P$-Valid Forests and Contexts.} Let $P$ be a set of
\shals. Later on $P$ will be a key parameter for the induction. We say that a
forest $t$ is \emph{$P$-valid} iff it has more than one node and all \shals
occurring in $t$ are in $P$. Similarly we define the notion of a $P$-valid
context. Note that we distinguish forests with one node in the definition. This
is a technical restriction that will be convenient without harming the
generality of the argument as the omitted forests are definable in \FOd. We extend
the notion of $P$-validity to elements of $H$ and $V$. We say $h \in H$ is
$P$-valid iff there exists a $P$-valid forest $t$ such that $h =
\alpha(t)$. Similarly for $v \in V$.

\paragraph*{\bf $P$-Reachability.} The logic \FOd can be seen as a two-way
logic navigating up and down within forests. Over strings, this two-way
behavior is reflected by two partial orders over the syntactic semigroup
capturing respectively the current knowledge when reading the string from left
to right and from right to left, and correspond to the Green's
relations $L$ and $R$.

Over forests it turns out that the relevant bottom-up order is a partial order
on forest types while the relevant top-down order is a partial order on context
types~\cite{mikolaj}.  In our case those are even parametrized by a set $P$ of
\shals and are called \emph{$P$-reachability}. The index within these orders
will be another parameter in our induction. 

Let $h,h'$ be two $P$-valid forest types, $h$ is said to be
$P$-reachable from the forest type $h'$ if there exists a $P$-valid
context type $v$ such that $h=vh'$. Two forest types are
$P$-equivalent if they are mutually $P$-reachable. 

The definition is similar for context types and is defined for any
$v,v' \in V$, not just $P$-valid ones. Given two contexts $v,v' \in V$
we say that $v$ is \emph{$P$-reachable} from $v'$ whenever 
there is a $P$-valid context type $u$ such that $v=v'u$. 

Notice that for both partial orders, if $P \subseteq P'$ then
$P$-reachability implies $P'$-reachability.  

\medskip

We will reduce the case when all \shals of $P$ have arity 1 or less to 
the string case. On the other hand, when $P$ contains at least one
\shal of arity at least 2 we will make use of the following property:

\begin{claim}\label{claim-maximal}
If $P$ contains a \shal of arity at least 2 then among $P$-equivalence
classes of $P$-valid forest types there exists a unique maximal one
with respect to $P$-reachability.
\end{claim}
\begin{proof}
Take $p\in P$ of arity $n \geq 2$. Given $h,h' \in H$ that are
$P$-valid, consider $t$ and $t'$ two $P$-valid forests such that 
$\alpha(t)=h$ and $\alpha(t')=h'$. Consider the sequence $T$
of $n$ $P$-valid forests containing copies of $t$ and $t'$, with
at least one copy of $t$ and one copy of $t'$. Now $\alpha(p[T])$
is $P$-reachable from both $h$ and $h'$. The result follows.
\end{proof}
In the cases when Claim~\ref{claim-maximal} applies we say that $P$ is
\emph{branching} and we denote by $H_P$ the maximal class given by
Claim~\ref{claim-maximal}. Finally, we say that a branching set of \shals $P$
is \emph{reduced} if all $P$-valid forest types are mutually
reachable, i.e. $H_p$ is the whole set of $P$-valid forest types.

\subsection{\FOd-Equivalence for \Shals.} It will often be necessary to
manipulate \shals modulo definability in \FOd. Typically, when applying the
Antichain Composition Lemma with a formula of the form ``select all nodes whose
\shal is in $P$ but have no ancestor with that property'', it will be necessary
that $P$ is definable in \FOd.

\paragraph*{\bf Definable set of \shals} Intuitively, \FOd
treats a \shal as a string whose letters are $a$, $b(a)$, or
$b(\hole)$. More formally, we define $\As$ as the alphabet containing
the letters $a$, $b(\hole)$ and $b(a)$ for all $a \in A$ and $b \in B$. We
say that $b$ is the \emph{inner-label} of $b(\hole)$ or $b(a)$. We see
a \shal $p$ as a string over the alphabet~$\As$. 
 
For each positive integer $k$ and any two \shals $p$ and $p'$, we
write $p \mequivk p'$ for the fact that Spoiler has a winning strategy
in the game played on $p$ and $p'$, seen as strings over $\As$. In
particular, we say that a set $P$ of \shals is \kdef iff it is a union
of classes of under $\mequivk$. As the name suggests a set $P$ of
\kdef \shals is definable in \FOd. In particular we get the following
claim which is an immediate consequence of
Lemma~\ref{lemma-ef-game}. 

\begin{claim}\label{lemma-p-valid}
For any $k$ and any \kdef set $P$ of \shals, the language of $P$-valid 
forests is definable in \FOd.
\end{claim}

\paragraph*{\bf Definable nodes in \shals} 
It will also sometimes be necessary to refer to an explicit node within a \shal.
Typically, when applying the Antichain Composition Lemma with a formula of the
form ``select all nodes $(p,x)$ such that \ldots and having no ancestor with that
property''. We will of course need to treat $(p,x)$ modulo definability.  It
would be tempting to use a notion of definability similar to the one used for
\shals in the previous paragraph. However the notion of \emph{saturation}
building from this would give a necessary but not sufficient characterization
for languages definable in \FOd. In particular, the induction in our
suffiency proof would not terminate (see Lemma~\ref{lemma-LPtau}). Our
notion of definability is based on an \efgame game relaxing the rules
defined in Section~\ref{games} in order to grant more power to
Duplicator.  Moreover it will be useful to parametrize the game by a
set $X \subseteq H$.  In the sequel $X$ will be another parameter of
the induction denoting those forest types for which we are still
looking for a formula defining the set of forests for that type.  When
applying the Antichain Composition Lemma, any forest of type $h
\not\in X$ can be safely replaced by a node with the appropriate label
as the corresponding language is definable in \FOd.

Given $X$, we distinguish between three kinds of nodes within \shals:
\emph{port-nodes} are nodes that are roots of trees of the form $b(\hole)$ with
$b \in B$, \emph{\Xnodes } are nodes that are roots of trees of the form $b(a)$
with $a \in A,b \in B$ and $\alpha(a) \in X$, finally \emph{\bXnodes } are the
remaining nodes, i.e. nodes with label $a\in A$ or roots of trees of the form
$b(a)$ with $a \in A, b \in B$ and $\alpha(a) \not\in X$.

Let $p$ and $p'$ be two \shals, seen as strings over \As. The $k$-round
\emph{$X$-relaxed game} between $p$ and $p'$ is defined as in
Section~\ref{games} but tests on labels in this new game are relaxed, making
the game easier for Duplicator. Since this alone makes the equivalence too
permissive (this yields a non-necessary characterization), we compensate by
giving Spoiler a third ``safety'' move in addition to the usual left and right
sibling moves. Spoiler can use this third move only under specific conditions
depending on the labels of the nodes the pebbles are currently on. These two
modifications achieve the right balance of expressive power.

At the start of each round Spoiler can move one of the two pebbles inside its
\shal to some left or right sibling $x$. Duplicator must respond by moving the
other pebble inside the other \shal in the same direction to a node $y$, if $x$
is an \bXnode then $y$ must have the same label as $x$. Otherwise, if $x$ is a
port- or \Xnode, $y$ must be a port- or \Xnode with the same
inner-label. Hence, at any point in the game the pebbles may lie on nodes with
labels $b(c),b(c')$ with $c \neq c'$ and $c,c' \in A \cup \{\hole\}$. In that
particular case Spoiler can use a 'safety' move: he selects one of the two
pebbles but does not move it, by hypothesis this pebble is on a node of label
$b(\hole)$ or $b(a)$ with $\alpha(a) \in X$. Duplicator must then place the other
pebble on a node of label $b(\hole)$.

Given two nodes $(p,x)$ and $(p',x')$, we write $(p,x) \pequiv{k}^X
(p,x')$ if $i)$ $p$ and $p'$ contain the same set of labels in \As,
$ii)$ $x,x'$ have the same label and $iii)$ Duplicator wins the
$k$-rounds $X$-relaxed game between $p$ and $p'$ starting at positions
$(p,x)$ and $(p',x')$. It will also be convenient to define
$\pequiv{k}^X$ on \shals only. We write $p \pequiv{k}^X p'$ iff $(p,x)
\pequiv{k}^X (p',x')$ with $x,x'$ the leftmost positions in $p,p'$. The
following claim is immediate from the definitions:

\begin{claim} \label{clm:global}
Assume $p \pequiv{k+2}^X p'$, then for any port-node $x \in p$
(resp. $x' \in p'$) there exists a port-node $x' \in p'$ (resp $x \in
p$) such that $(p,x) \pequiv{k}^X (p',x')$.
\end{claim}

\begin{proof}
Set $y,y'$ as the leftmost positions in $p,p'$. In the $(k+2)$-rounds
$X$-relaxed game between $(p,y)$ and $(p',y')$, Spoiler can use the
two initial rounds to move the pebble to $x'$ in $p'$ (resp. to
$x$ in $p$) and then use (if necessary) a safety move. Duplicator's
strategy yields the desired port-node $x$ in $p$ (resp. $x' \in p'$).
\end{proof}

It is not immediate from the definitions that $\pequiv{k}^X$ is an
equivalence relation (transitivity is not obvious). We prove it in the
next lemma.

\begin{lemma} \label{lem-eqiseq}
For all $X\subseteq H$ and all $k$, $\pequiv{k}^X$ is an equivalence
relation.
\end{lemma}

\begin{proof}
It is clear from the definitions that the relation is reflexive and
symmetric. We now prove transitivity. Assume $(p,x) \pequiv{k}^X (p',x')$
and $(p',x') \pequiv{k}^X (p'',x'')$. We want to show that $(p,x)
\pequiv{k}^X (p'',x'')$. It is clear that $p$ and $p''$ contain the
same set of labels and that $x$ and $x'$ have the same label. We need
to prove that Duplicator has a winning strategy in the $k$-rounds
$X$-relaxed game between $(p,x)$ and $(p'',x'')$.

By hypothesis, Duplicator has winning strategies in the $k$-rounds
$X$-relaxed games 
between $(p,x)$ and $(p',x')$ and between $(p',x')$ and
$(p'',x'')$. We combine these strategies in the obvious way. Assume
that Spoiler makes a non safety move on $p$, then Duplicator obtains
an answer in $p'$ from her strategy in the game on $p$ and $p'$, plays
that answer as a move for Spoiler in the game on $p'$ and $p''$ which
yields an answer in $p''$ from her strategy on that game. This is her
answer to Spoiler's move, and it is immediate to check that this is a
correct answer. By symmetry she can answer a similar move of
Spoiler on $p''$.

Assume now that Spoiler makes a safety move on $p$. Observe that this
cannot happen in the first-round since $x$ and $x''$ have the same
label. Therefore after this round there will be at most $k-2$ rounds
left to play. Let $b$ be the inner label of the pebble on $p$. Since a
safety move was used, the pebbles in $p$ and $p''$ must have different
labels. Hence at least one of these labels is different from the label
of the pebble in $p'$. We distinguish two cases depending on which one
it is.

If the pebbles on $p'$ and $p''$ have different labels, then
Duplicator can use a Safety move in the game between $p'$ and $p''$ to 
get $y''$ in $p''$ with label $b(\hole)$ from where she can continue
to play the game. If the pebbles on $p'$ and $p$ have different
labels, then Duplicator can use a Safety move in the game between $p'$
and $p$ to get $y'$ in $p'$ with label $b(\hole)$. We can then use
Claim~\ref{clm:global} to get $y''$ with label $b(\hole)$ and such
that $(p'',y'') \pequiv{k-2}^X (p',y')$, this is Duplicator's
answer. It is correct since the number of rounds left to play is less
than $k-2$.
\end{proof}

As for Lemma~\ref{lemma-ef-game} is it easy to show that for all
$k$, $\pequiv{k}^X$ has finite index. Finally, the following claim is
a simple variant of Claim~\ref{lemma-p-valid}:

\begin{claim}\label{lemma-types-fod}
Let $X \subseteq H$ and $(p,x)$ be a node. There is a
\FOd formula $\psi_{p,x}$ having one free variable and such that for
any forest $s$, $\psi_{p,x}$ holds exactly at all nodes $(p',x')$ such
that $(p,x) \pequiv{k}^X (p',x')$.
\end{claim}

\begin{proof}
It is immediate that $(p,x) \mequiv{k} (p',x') \Rightarrow (p,x)
\pequiv{k}^X (p',x')$ (following the strategy provided by $(p,x)
\mequiv{k} (p',x')$ prevents Spoiler from using any safety move
in the $X$-relaxed game). Hence any $\pequiv{k}^X$-class is a union of
$\mequiv{k}$-classes and the result follows.
\end{proof}

\section{Decidable Characterization of \FOd } \label{sec:carac}

In this section we present our decidable characterization for \FOd.  It
involves three properties of the syntactic morphism of the language. As we
explained, the first two are simple identities on $H$ and $V$, the third one,
\emph{saturation} is a new notion that is parametrized by a set $P$ of \shals
and the associated equivalences.  We first define saturation and then state the
characterization.

\subsection{Saturation}\label{sec-saturation}

Let $\alpha: \A^\mydelta \rightarrow (H,V)$ be a morphism into a finite
forest algebra $(H,V)$. Note that as for \shals while saturation
makes sense for any morphism, it is designed to be used with leaf
surjective ones. In particular, in the characterization, we state
saturation on the leaf completion of the syntactic morphism of the
language. Consider some \emph{branching} and \emph{reduced} set $P$ of
\shals (note that we do \emph{not} ask $P$ to be definable). Recall that
$P$ will be the set of allowed patterns. Let $H_P$ be the unique
maximal class given by Claim~\ref{claim-maximal}. Note that since $P$
is reduced, $H_P$ is also the set of $P$-valid forest types. 

Set $k \in \nat$. We say that a context $\Delta$ is
\emph{\saturated{(P,k)}} if (i) it is $P$-valid and, (ii) for each
port-node $(p,x) \in P$ there exists a port-node $(p',x')$ 
on the backbone of $\Delta$ such that $(p,x) \pequiv{k}^{H_P} (p',x')$.
We say that $\alpha$ is \emph{closed under $k$-saturation} if for
all branching and reduced sets $P$ of \shals, for all contexts
$\Delta$ that are \saturated{(P,k)} and all $h_1,h_2 \in H_P$, we
have:

\begin{equation}\label{eqt}
\alpha(\Delta)^\omega h_1 = \alpha(\Delta)^\omega h_2
\end{equation}

We say that $\alpha$ is \emph{closed under saturation} if it is closed
under $k$-saturation for some $k$. We will need the following simple
observation.

\begin{lemma} \label{lem-sat-k}
Let $\alpha: \A^\mydelta \rightarrow (H,V)$ be a morphism into a finite
forest algebra $(H,V)$ and $k,k'$ two integers such that $k' > k$. If
$\alpha$ is closed under $k$-saturation then $\alpha$ is closed under
$k'$-saturation.
\end{lemma}

\begin{proof}
This is immediate since by definition, any $\Delta$ that is
\saturated{(P,k')} is \saturated {(P,k)} as well.
\end{proof}

\subsection{Characterization of  \FOd}\label{carac}

We are now ready to state the main result of this paper.

\begin{theorem}\label{th-efmax}
A regular forest language $L$ is definable in \FOd iff its syntactic
morphism $\alpha: \A^{\mydelta} \rightarrow (H,V)$  satisfies the
following properties:

\begin{enumerate}[label=\alph*)]
\item $H$ satisfies the equation
{\small\begin{equation}\label{eqh}\omega(h+g)+ g+ \omega(h+g)=\omega(h+g)
  \end{equation}
}
\item $V$ satisfies the equation
{\small\begin{equation}\label{eqv}
(uv)^{\omega}v(uv)^{\omega}=(uv)^\omega
\end{equation}
}
\item the leaf completion of $\alpha$ is closed under saturation.
\end{enumerate}
\end{theorem}

\noindent
Notice that~\eqref{eqh} and~\eqref{eqv} above are exactly the identities
characterizing, over strings, definability in \FOdw (recall
Theorem~\ref{thm-word-fod}) and are therefore necessary for being definable in
\FOd. It is easy to see that they are not sufficient to characterize \FOd. To
see this, consider the language of trees that corresponds to Boolean
expressions, with {\sc and} and {\sc or} inner nodes and $0$ or $1$ leaves,
that evaluates to $1$. One can verify that the syntactic forest algebra of this
language satisfies~\eqref{eqh} and~\eqref{eqv}. However it is not definable in
\FOd, actually not even in $\textup{FO}(\orderv,\orderh)$~\cite{Potthoff94}.

Recall that \FOd can express the fact that a forest is a tree and, for
each $k$, that a tree has rank $k$, hence Theorem~\ref{th-efmax} also
apply for regular tree languages and regular ranked tree languages.

In Section~\ref{sec-correc} we will prove that the
properties listed in the statement of Theorem~\ref{th-efmax} are
necessary for having definability in \FOd using a simple but tedious
\efgame argument. In Section~\ref{main-proof} we prove the difficult
direction of Theorem~\ref{th-efmax}, i.e. that the properties imply
definability in \FOd.

Finally in Section~\ref{more} we show that the properties listed in 
Theorem~\ref{th-efmax} can be effectively tested. This is simple
for~\eqref{eqh} and~\eqref{eqv} but will require an intricate pumping
argument for saturation. Altogether Theorem~\ref{th-efmax} achieves
our goal and provides a decidable characterization of regular forest
languages definable in \FOd.

\section{Correctness of the properties}\label{sec-correc}

\newcommand\inv[1]{\ensuremath{\mathcal{P}(#1)}}

\tikzstyle{lab} = [inner sep=4pt]
\tikzstyle{scal} = [inner sep=0pt]
\tikzstyle{port}=[draw,rectangle,minimum size=8pt,inner sep=0pt]
\tikzstyle{node}=[draw,circle,minimum size=12pt,inner sep=0pt]
\tikzstyle{nodeg}=[draw,circle,,fill=gray!50,minimum size=12pt,inner sep=0pt]
\tikzstyle{dot}=[draw,circle,fill,minimum size=4pt,inner sep=0pt]
\tikzstyle{tree}=[anchor=north,draw,shape=isosceles triangle,minimum size=28pt,shape border rotate=90,scale=0.8]
\tikzstyle{trees}=[anchor=north,draw,shape=isosceles triangle,minimum size=20pt,shape border rotate=90,scale=0.8]
\tikzstyle{minidot}=[draw,circle,fill,minimum size=1pt,inner sep=0pt]
\tikzstyle{ars}=[draw,->,line width=2pt]

\newcommand\psport{ex-port-node\xspace}
\newcommand\psrela{ex-\Xnode}
\newcommand\psxnod{ex-\bXnode\xspace}
\newcommand\psports{ex-port-nodes\xspace}
\newcommand\psrelas{ex-\Xnodes}
\newcommand\psxnods{ex-\bXnodes\xspace}

In this section we prove that the properties stated in
Theorem~\ref{th-efmax} are necessary for being definable in \FOd. We
prove that any language $L$ definable in \FOd is closed under
saturation and its syntactic forest algebra satisfies the
identities~\eqref{eqh} and~\eqref{eqv}.

If $L$ is definable in \FOd then it is simple to see that its
syntactic forest algebra must satisfy the identities~\eqref{eqh}
and~\eqref{eqv}. This is because Identity~\eqref{eqh} is essentially
concerned by sequences of forests with the $+$ operation. Therefore
each such sequence can be treated as a string over \orderh and
Theorem~\ref{thm-word-fod} can be applied to show that the identity is
necessary. Similarly Identity~\eqref{eqv} concerns only sequences of
contexts that can also be treated as strings over \orderv.

The necessary part of Theorem~\ref{th-efmax} then follows from the
following proposition.

\begin{proposition}\label{prop-saturation}
Let $L$ be a forest language that is definable in \FOd. Then the leaf 
completion of the syntactic morphism of $L$ is closed under
saturation.
\end{proposition}

The rest of this section is devoted to the proof of
Proposition~\ref{prop-saturation}. It is a classical but tedious
Ehrenfeucht-Fraïssé argument.

Assume $L$ is definable in \FOd. It follows from
Theorem~\ref{thm-fod-eff} that $L$ is definable in \EFF. Let $\alpha:
\A^\mydelta \rightarrow (H,V)$ be the leaf completion of its syntactic
morphism. Let $k$ be the nesting depth of the navigational modalities
used in the formula recognizing $L$, we prove that $\alpha$ is closed
under $k$-saturation.

Let $P$ be a branching and reduced set of \shals.
Let $X=H_P$ be the associated class of $P$-valid forest types.
Let $\Delta$ be a \saturated{(P,k)} context.
Let $u=\alpha(\Delta)$ and $h_1,h_2$ be two forest types in $H_P$.
We need to show that $u^\omega h_1=u^\omega h_2$.

For this we exhibit two forests $S_1$ and $S_2$ over \A such that
$\alpha(S_1)=u^\omega h_1$ and $\alpha(S_2)=u^\omega h_2$ and such
that Duplicator has a winning strategy for the $k$-move game described
in Section~\ref{games} when playing on $S_1$ and $S_2$.  Therefore it
follows from Lemma~\ref{lemma-ef-game} that no formula of \EFF whose
nesting depth of its navigational modalities is less than~$k$ can
distinguish between the two forests. This implies $u^\omega h_1 =
u^\omega h_2$ as desired.

Our agenda is now as follows. In Section~\ref{section-game-definition}
we define the two forests $S_1$ and $S_2$ on which we will play. Then in
Section~\ref{section-game-winning} we give the winning strategy for
Duplicator in the $k$-move game on $S_1$ and $S_2$.

We start with some definitions that will play the key role in the
proof. The \emph{root} of a forest is the root of the leftmost tree of
that forest. Recall that the \emph{backbone} of a context is the path
containing all the ancestors of the port of that context. The
\emph{skeleton} of a context is the set of nodes composed of the
backbone together with their siblings. Both notions are illustrated in
Figure~\ref{fig-ske}.

\begin{figure}[h]
\begin{center}
\begin{tikzpicture}
\node[node] (r1) at (-1.9,0) {$b$};
\node[node] (r2) at (-0.95,0) {$a$};
\node[nodeg] (r3) at (0,0) {$b$};
\node[node] (r4) at (0.95,0) {$a$};
\node[node] (r5) at (1.9,0) {$b$};

\coordinate (k0) at (r1.south);
\draw (k0) to ($(k0)-(0.3,0.2)$) to ($(k0)-(0.4,0.7)$) to
($(k0)-(-0.4,0.7)$) to ($(k0)-(-0.3,0.2)$) to ($(k0)-(0.3,0.2)$); 
\draw (k0) to ($(k0)-(-0.3,0.2)$);

\coordinate (k0) at (r4.south);
\draw (k0) to ($(k0)-(0.3,0.2)$) to ($(k0)-(0.4,0.7)$) to
($(k0)-(-0.4,0.7)$) to ($(k0)-(-0.3,0.2)$) to ($(k0)-(0.3,0.2)$); 
\draw (k0) to ($(k0)-(-0.3,0.2)$);

\coordinate (k0) at (r5.south);
\draw (k0) to ($(k0)-(0.3,0.2)$) to ($(k0)-(0.4,0.7)$) to
($(k0)-(-0.4,0.7)$) to ($(k0)-(-0.3,0.2)$) to ($(k0)-(0.3,0.2)$); 
\draw (k0) to ($(k0)-(-0.3,0.2)$);

\node[node] (s1) at (-0.95,-1.2) {$a$};
\node[nodeg] (s2) at (0.0,-1.2) {$c$};
\draw (r3) to [out=-105,in=90] (s1);
\draw (r3) to (s2);

\coordinate (k0) at (s1.south);
\draw (k0) to ($(k0)-(0.3,0.2)$) to ($(k0)-(0.4,0.7)$) to
($(k0)-(-0.4,0.7)$) to ($(k0)-(-0.3,0.2)$) to ($(k0)-(0.3,0.2)$); 
\draw (k0) to ($(k0)-(-0.3,0.2)$);


\node[nodeg] (t2) at (0.0,-2.4) {$d$};
\node[node] (t3) at (0.95,-2.4) {$d$};
\node[node] (t4) at (1.9,-2.4) {$a$};
\node[node] (t5) at (2.85,-2.4) {$b$};

\draw (s2) to (t2);
\draw (s2) to [out=-75,in=90] (t3);
\draw (s2) to [out=-60,in=90] (t4);
\draw (s2) to [out=-45,in=90] (t5);

\coordinate (k0) at (t4.south);
\draw (k0) to ($(k0)-(0.3,0.2)$) to ($(k0)-(0.4,0.7)$) to
($(k0)-(-0.4,0.7)$) to ($(k0)-(-0.3,0.2)$) to ($(k0)-(0.3,0.2)$); 
\draw (k0) to ($(k0)-(-0.3,0.2)$);



\node[port] (p3) at (0.0,-3.0) {};
\draw (t2) -- (p3);

\node[scal] at (0.0,-4.0) {Backbone};

\begin{scope}[xshift=7cm]
\node[nodeg] (r1) at (-1.9,0) {$b$};
\node[nodeg] (r2) at (-0.95,0) {$a$};
\node[nodeg] (r3) at (0,0) {$b$};
\node[nodeg] (r4) at (0.95,0) {$a$};
\node[nodeg] (r5) at (1.9,0) {$b$};

\coordinate (k0) at (r1.south);
\draw (k0) to ($(k0)-(0.3,0.2)$) to ($(k0)-(0.4,0.7)$) to
($(k0)-(-0.4,0.7)$) to ($(k0)-(-0.3,0.2)$) to ($(k0)-(0.3,0.2)$); 
\draw (k0) to ($(k0)-(-0.3,0.2)$);

\coordinate (k0) at (r4.south);
\draw (k0) to ($(k0)-(0.3,0.2)$) to ($(k0)-(0.4,0.7)$) to
($(k0)-(-0.4,0.7)$) to ($(k0)-(-0.3,0.2)$) to ($(k0)-(0.3,0.2)$); 
\draw (k0) to ($(k0)-(-0.3,0.2)$);

\coordinate (k0) at (r5.south);
\draw (k0) to ($(k0)-(0.3,0.2)$) to ($(k0)-(0.4,0.7)$) to
($(k0)-(-0.4,0.7)$) to ($(k0)-(-0.3,0.2)$) to ($(k0)-(0.3,0.2)$); 
\draw (k0) to ($(k0)-(-0.3,0.2)$);

\node[nodeg] (s1) at (-0.95,-1.2) {$a$};
\node[nodeg] (s2) at (0.0,-1.2) {$c$};
\draw (r3) to [out=-105,in=90] (s1);
\draw (r3) to (s2);

\coordinate (k0) at (s1.south);
\draw (k0) to ($(k0)-(0.3,0.2)$) to ($(k0)-(0.4,0.7)$) to
($(k0)-(-0.4,0.7)$) to ($(k0)-(-0.3,0.2)$) to ($(k0)-(0.3,0.2)$); 
\draw (k0) to ($(k0)-(-0.3,0.2)$);


\node[nodeg] (t2) at (0.0,-2.4) {$d$};
\node[nodeg] (t3) at (0.95,-2.4) {$d$};
\node[nodeg] (t4) at (1.9,-2.4) {$a$};
\node[nodeg] (t5) at (2.85,-2.4) {$b$};

\draw (s2) to (t2);
\draw (s2) to [out=-75,in=90] (t3);
\draw (s2) to [out=-60,in=90] (t4);
\draw (s2) to [out=-45,in=90] (t5);

\coordinate (k0) at (t4.south);
\draw (k0) to ($(k0)-(0.3,0.2)$) to ($(k0)-(0.4,0.7)$) to
($(k0)-(-0.4,0.7)$) to ($(k0)-(-0.3,0.2)$) to ($(k0)-(0.3,0.2)$); 
\draw (k0) to ($(k0)-(-0.3,0.2)$);



\node[port] (p3) at (0.0,-3.0) {};
\draw (t2) -- (p3);

\node[scal] at (0.0,-4.0) {Skeleton};

\end{scope}

\end{tikzpicture}
\end{center}
\caption{Illustration of the notions of backbone and skeleton.}\label{fig-ske}
\end{figure}

\subsection{Definition of the forests $S_1$ and
  $S_2$.}\label{section-game-definition}


Let $X=H_P=\set{h_1,\cdots,h_\ell}$. Recall that since $P$ is reduced,
$X$ is also the set of $P$-valid types. In particular all subforests
of a $P$-valid forest or context that are not leaves have a type in $X$.
We fix a set $\set{s_1,\dots,s_\ell}$ of $P$-valid forests such that
for all $i$, $s_i \not\in A$ and $\alpha(s_i)=h_i$. This is without
loss of generality as for each $i$, $h_i$ is in $H_P$ and therefore
reachable from any type and therefore there is a forest of arbitrary
depth of that type.


Given a $P$-valid context $C$ and $\ell$ forests $t_1,\cdots,t_\ell$,
we say that \emph{$C'$ is the context obtained from $C$ by replacing
  all subforests of type $h_j$ by $t_j$} if $C'$ is constructed by
considering all the nodes $x$ that are on the skeleton of $C$ but not
on the backbone and, if the subforest below $x$ is $s$ where
$\alpha(s)=h_j$ with $j\leq \ell$, we replace it with $t_j$. By
construction, $C$ and $C'$ have the same skeleton. Notice that since
$P$ is reduced and $C$ is assumed to be $P$-valid, the construction
replaces the subforests below all ports and \Xnodes on the skeleton of
$C$ that are not on the backbone and leaves the \bXnodes
unchanged. Since we assumed that all $s_i$ are not in $A$, this means
that all \Xnodes on the skeleton of $C$ become port nodes within
$C'$. Therefore, $C'$ may contain on its backbone \shals that are
not in $P$ and saturation may not be preserved. We will show how to
deal with this fact later.


\bigskip

Since $P$ is branching, there exists a \shal $q_0\in P$ of arity
greater than~$1$. For $i\leq \ell$, we denote by $V_i$ the context
obtained from $q_0$ by placing the forest $s_i$ into all the ports of
$q_0$ except for the rightmost one.

By maximality of $H_P$ relative to $P$-reachability, for all $i \leq
\ell$ there exists a $P$-valid context $U'_i$ such that $h_i =
\alpha(U'_i)\alpha(V_\ell \cdots V_1) u^\omega h_1$. For $i \leq
\ell$, we write $U_i$ for the context $U'_i V_\ell \cdots V_0$. For
all $i \leq \ell$ we write $u_i = \alpha(U_i)$. By definition, for $i
\leq \ell$, the contexts $U_i$ have the following properties:

\begin{itemize}
\item $U_i$ is $P$-valid,
\item $u_i u^\omega h_1 = h_i$,
\item the context $U_i$ contains for all $j \leq \ell$ a
  subforest of type $h_j$ such that all nodes on the path
  to this copy are port-nodes (namely $s_j$ within $V_j$).
\end{itemize}

We now construct by induction on $j$ contexts $\Delta_{j}$ and
$U_{i,j}$, and forests $T_{i,j}$ such that for all $i \leq \ell$, we
have $\alpha(\Delta_j)=u$, $\alpha(U_{i,j})=u_i$ and $\alpha(T_{i,j})
= h_i$. We initialize the process by setting for all $i\leq \ell$:

\begin{itemize}
\item $U_{i,0}$ is formed from $U_{i}$ by replacing all subforests of
  type $h_{j}$ by $s_{j}$,
\item $\Delta_0$ is obtained from $\Delta$ by replacing all subforests of type
  $h_j$ with $s_j$,
\item $T_{i,0}:= U_{i,0} \cdot (\Delta_{0})^{\omega} \cdot s_1$. 
\end{itemize}

By construction we have $\alpha(\Delta_0)=u$, $\alpha(U_{i,0})=u_i$ and
$\alpha(T_{i,0})=u_i u^{\omega} h_1 = h_i$ as desired. When $j > 0$,
the inductive step of the construction is done as follows for all $i
\leq \ell$:

\begin{itemize}
\item $U_{i,j}$ is formed from $U_{i}$ by replacing each
 subforest of type $h_{i'}$ by $T_{i',j-1}$ (see Figure~\ref{fig-construc-uij}),
\item $\Delta_j$ is formed from $\Delta$ by replacing each
 subforest of type $h_{i'}$ by $T_{i',j-1}$,
\item $T_{i,j} = U_{i,j} \cdot (\Delta_{j})^{\omega} \cdots (\Delta_{0})^{\omega} \cdot s_1$.
\end{itemize}

\begin{figure}[h]
\begin{center}
\begin{tikzpicture}
\node[node] (r1) at (-1.9,0) {$b$};
\node[node] (r2) at (-0.95,0) {$a$};
\node[node] (r3) at (0,0) {$b$};
\node[node] (r4) at (0.95,0) {$a$};
\node[node] (r5) at (1.9,0) {$b$};

\coordinate (k0) at (r1.south);
\draw (k0) to ($(k0)-(0.3,0.2)$) to ($(k0)-(0.4,0.7)$) to
($(k0)-(-0.4,0.7)$) to ($(k0)-(-0.3,0.2)$) to ($(k0)-(0.3,0.2)$); 
\draw (k0) to ($(k0)-(-0.3,0.2)$);
\node[bag] at ($(k0)-(0.0,0.5)$) {\small $h_1$};

\coordinate (k0) at (r4.south);
\draw (k0) to ($(k0)-(0.3,0.2)$) to ($(k0)-(0.4,0.7)$) to
($(k0)-(-0.4,0.7)$) to ($(k0)-(-0.3,0.2)$) to ($(k0)-(0.3,0.2)$); 
\draw (k0) to ($(k0)-(-0.3,0.2)$);
\node[bag] at ($(k0)-(0.0,0.5)$) {\small $h_2$};

\coordinate (k0) at (r5.south);
\draw (k0) to ($(k0)-(0.3,0.2)$) to ($(k0)-(0.4,0.7)$) to
($(k0)-(-0.4,0.7)$) to ($(k0)-(-0.3,0.2)$) to ($(k0)-(0.3,0.2)$); 
\draw (k0) to ($(k0)-(-0.3,0.2)$);
\node[bag] at ($(k0)-(0.0,0.5)$) {\small $h_1$};

\node[node] (s1) at (-0.95,-1.2) {$a$};
\node[node] (s2) at (0.0,-1.2) {$c$};
\draw (r3) to [out=-105,in=90] (s1);
\draw (r3) to (s2);

\coordinate (k0) at (s1.south);
\draw (k0) to ($(k0)-(0.3,0.2)$) to ($(k0)-(0.4,0.7)$) to
($(k0)-(-0.4,0.7)$) to ($(k0)-(-0.3,0.2)$) to ($(k0)-(0.3,0.2)$); 
\draw (k0) to ($(k0)-(-0.3,0.2)$);
\node[bag] at ($(k0)-(0.0,0.5)$) {\small $h_3$};

\node[node] (t2) at (0.0,-2.4) {$d$};
\node[node] (t3) at (0.85,-2.4) {$d$};
\node[node] (t4) at (1.7,-2.4) {$a$};
\node[node] (t5) at (2.55,-2.4) {$b$};

\coordinate (k0) at (t4.south);
\draw (k0) to ($(k0)-(0.3,0.2)$) to ($(k0)-(0.4,0.7)$) to
($(k0)-(-0.4,0.7)$) to ($(k0)-(-0.3,0.2)$) to ($(k0)-(0.3,0.2)$); 
\draw (k0) to ($(k0)-(-0.3,0.2)$);
\node[bag] at ($(k0)-(0.0,0.5)$) {\small $h_1$};


\node[port] (p3) at (0.0,-3.1) {};

\draw (s2) to (t2);
\draw (s2) to [out=-75,in=90] (t3);
\draw (s2) to [out=-60,in=90] (t4);
\draw (s2) to [out=-45,in=90] (t5);

\draw (t2) -- (p3);

\node[lab] at (0.0,-4.0) {\large $U_i$};

\draw[ars,double] (3.7,-2.0) to (4.7,-2.0);


\begin{scope}[xshift=8cm]
\node[node] (r1) at (-1.9,0) {$b$};
\node[node] (r2) at (-0.95,0) {$a$};
\node[node] (r3) at (0,0) {$b$};
\node[node] (r4) at (0.95,0) {$a$};
\node[node] (r5) at (1.9,0) {$b$};

\coordinate (k0) at (r1.south);
\draw (k0) to ($(k0)-(0.3,0.2)$) to ($(k0)-(0.4,0.7)$) to
($(k0)-(-0.4,0.7)$) to ($(k0)-(-0.3,0.2)$) to ($(k0)-(0.3,0.2)$); 
\draw (k0) to ($(k0)-(-0.3,0.2)$);
\node[bag] at ($(k0)-(0.0,0.5)$) {\small $T_{1,j}$};

\coordinate (k0) at (r4.south);
\draw (k0) to ($(k0)-(0.3,0.2)$) to ($(k0)-(0.4,0.7)$) to
($(k0)-(-0.4,0.7)$) to ($(k0)-(-0.3,0.2)$) to ($(k0)-(0.3,0.2)$); 
\draw (k0) to ($(k0)-(-0.3,0.2)$);
\node[bag] at ($(k0)-(0.0,0.5)$) {\small $T_{2,j}$};

\coordinate (k0) at (r5.south);
\draw (k0) to ($(k0)-(0.3,0.2)$) to ($(k0)-(0.4,0.7)$) to
($(k0)-(-0.4,0.7)$) to ($(k0)-(-0.3,0.2)$) to ($(k0)-(0.3,0.2)$); 
\draw (k0) to ($(k0)-(-0.3,0.2)$);
\node[bag] at ($(k0)-(0.0,0.5)$) {\small $T_{1,j}$};

\node[node] (s1) at (-0.95,-1.2) {$a$};
\node[node] (s2) at (0.0,-1.2) {$c$};
\draw (r3) to [out=-105,in=90] (s1);
\draw (r3) to (s2);

\coordinate (k0) at (s1.south);
\draw (k0) to ($(k0)-(0.3,0.2)$) to ($(k0)-(0.4,0.7)$) to
($(k0)-(-0.4,0.7)$) to ($(k0)-(-0.3,0.2)$) to ($(k0)-(0.3,0.2)$); 
\draw (k0) to ($(k0)-(-0.3,0.2)$);
\node[bag] at ($(k0)-(0.0,0.5)$) {\small $T_{3,j}$};

\node[node] (t2) at (0.0,-2.4) {$d$};
\node[node] (t3) at (0.85,-2.4) {$d$};
\node[node] (t4) at (1.7,-2.4) {$a$};
\node[node] (t5) at (2.55,-2.4) {$b$};

\coordinate (k0) at (t4.south);
\draw (k0) to ($(k0)-(0.3,0.2)$) to ($(k0)-(0.4,0.7)$) to
($(k0)-(-0.4,0.7)$) to ($(k0)-(-0.3,0.2)$) to ($(k0)-(0.3,0.2)$); 
\draw (k0) to ($(k0)-(-0.3,0.2)$);
\node[bag] at ($(k0)-(0.0,0.5)$) {\small $T_{1,j}$};


\node[port] (p3) at (0.0,-3.1) {};

\draw (s2) to (t2);
\draw (s2) to [out=-75,in=90] (t3);
\draw (s2) to [out=-60,in=90] (t4);
\draw (s2) to [out=-45,in=90] (t5);

\draw (t2) -- (p3);

\node[lab] at (0.0,-4.0) {\large $U_{i,j+1}$};
\end{scope}
\end{tikzpicture}
\end{center}
\caption{Illustration of the construction of $U_{i,j}$ from $U_i$: each
  subforest of type $h_{i'}$ in $U_i$ is replaced by $T_{i',j-1}$.}\label{fig-construc-uij}
\end{figure}

By induction we have for all $j \leq \ell$, $\alpha(U_{i,j}) =
\alpha(U_i) = u_i$, $\alpha(\Delta_{j}) = \alpha(\Delta) = u$ and
$\alpha(T_{i,j}) = u_i u^{\omega} h_1 = h_i$ as required. Notice that
each $U_{i,j}$ contains a copy of $T_{i',j-1}$ for all $i' \leq
\ell$. Let $m=2^{k}$ and let:
\begin{equation}
\begin{split}
 S_1 & := (\Delta_{m})^{(m+1)\omega} \cdot T_{1,m}\\
 S_2 & := (\Delta_{m})^{(m+1)\omega} \cdot T_{2,m}
\end{split}
\end{equation}

Note that by definition $\alpha(S_1) = u^\omega h_1$ and
$\alpha(S_2)=u^\omega h_2$. Therefore the following lemma concludes
the proof of Proposition~\ref{prop-saturation}. 

\begin{lemma}\label{lemma-winning}
  Duplicator has a winning strategy for the $k$-move game between $S_1$ and
  $S_2$.
 \end{lemma}

\subsection{Basic Properties}
Before proving Lemma~\ref{lemma-winning} we state some basic
properties of the construction.

Recall that the operation constructing $C'$ from a $P$-valid $C$ by
replacing all subforests of type $h_j$ by $t_j$ does not preserve
saturation. The issue is that all \Xnodes become port-nodes and the
resulting \shal may no longer be equivalent to one occurring in the
backbone of $C$. In order to cope with this problem, we remember what
the initial situation was. If $x$ is a port-node on the skeleton of
$C'$ (recall that $C$ and $C'$ have the same skeleton) we say that:

\begin{itemize}
\item $x$ is an \emph{\psport} if $x$, as a node of $C$, was a
  port-node. 
\item $x$ is an \emph{\psrela} if $x$, as a node of $C$, was a
  \Xnode.
\end{itemize}

Observe that, by construction, for any port-node $x$ of $C'$, the
subtree at $x$ in $C'$ is $b(t_j)$ for some $j$ and some $b \in
B$. For \psrelas, this is by construction. For \psports, this is
because $P$ is reduced and $C$ was $P$-valid, hence all subforets of
$C$ but leaves had type in $H_P$. Also notice that all remaining
nodes of $C'$ are \bXnodes and are unchanged with respect to $C$. 

\newcommand\psim{\sim}

Let $x_1'$ (resp. $x'_2$) be a node on the skeleton of a context
$C'_1$ (resp. $C'_2$) constructed from a $P$-valid context $C_1$
(resp. $C_2$) by replacing all subforests of type $h_i$ by
$T_{i,j_1}$ for some fixed $j_1$ (resp. $T_{i,j_2}$ for some fixed
$j_2$). Let $x_1,x_2$ be the corresponding nodes on the skeletons of
$C_1,C_2$. We write $x_1' \psim_n^X x_2'$ when $x_1 \pequiv{n}^X
x_2$. Because $\pequiv{n}^X$ is an equivalence, it is 
straightforward to verify that $\psim_n^X$ is also an equivalence
relation.

There is a game definition of $\psim_n^X$, called pseudo-\relaxedX
game. In this game Spoiler can now use the safety move when one pebble 
is on a \psport and the other on a \psrela or when both pebbles are on
\psrelas with subtrees $b(T_{i_1,j_1}),b(T_{i_2,j_2})$ such that $i_1
\neq i_2$. In this case Duplicator must answer by placing the other
pebble on a \psport.

By construction the following property is an immediate consequence of
the $(P,k)$-saturation of $\Delta$: Assume $C$ is $P$-valid and that
$C'$ is constructed from $C$ by replacing all subforests of type $h_i$
by $T_{i,j}$ for some fixed $j$. Then for any $n$ and any \psport $x$
of $C'$ there exists a $y$ in the {\bf backbone} of $\Delta_n$ such
that $x \psim_k^X y$. We call this property, the
\emph{pseudo-saturation} of $\Delta_n$.

\medskip

Notice that when using \emph{pseudo-saturation} in her strategy,
Duplicator may end up in a situation where the pebbles are above
subtrees $T_{i,j_1}$ and $T_{i,j_2}$ with $j_1 \neq j_2$. Note that
$T_{i,j_1}$ and $T_{i,j_2}$ are essentially the same tree, only with
different nesting. In this situation, Duplicator may have to play a
subgame within the trees $T_{i,j_1}$ and $T_{i,j_2}$. The following
lemma states that this is possible as soon as $j_1,j_2$ are large
enough. 

\begin{lemma} \label{lem-aper}
Given integers $i,n,j_1,j_2$ such that $2(n-1)+1 \leq j_1,j_2$, 
Duplicator has a winning strategy for the $n$-move game between 
$T_{i,j_1}$ and $T_{i,j_2}$.
\end{lemma}

\begin{proof}
This can be proved by a simple induction on $n$. Essentially this is a
straightforward generalization of a classical argument used to prove
that the words $w^{j_1}$ and $w^{j_2}$ cannot be distinguished by a
first-order formula of fixed quantifier rank provided that $j_1,j_2$
are large enough (see~\cite{bookstraub} for example).
\end{proof}

\subsection{The winning strategy: {\bf Proof of Lemma~\ref{lemma-winning}}}\label{section-game-winning}

We give a winning strategy for Duplicator in the $k$-move game between
$S_1$ and $S_2$. In order to be able to formulate this strategy, we
first define the useful parameters and their key properties that we
will later use.

The \emph{backbone} of $S_1$ ($S_2$) is the path going through all the ports of
the copies of $\Delta_m$ within $S_1$ ($S_2$) and the \emph{skeleton} of $S_1$ is
the set of nodes within the backbone of $S_1$ together with their siblings.

The \emph{nesting level} of a node $x$ of $S_1$ or $S_2$ is the
minimal number $j$ such that $x$ belongs to a context $\Delta_{j}$ or
$U_{i,j}$. We set the nesting level of the nodes that are in any copy
of a forest $s_1,\dots,s_\ell$ to~$0$. The notion of nesting level is
illustrated in Figure~\ref{fig-nlevel}. Note that this number is
equal in siblings and may only increase when going up in the
tree. When this number is low, we are near the leaves of the forest
and we need to make sure that the current positions of the game point
to isomorphic subtrees. Recall that because of the construction of the
context $U_{i,j}$, a node of nesting level $j$ always has, for all $i'
\leq \ell$, a descendant that is at the root of a copy of $T_{i',j-1}$
and, for all $j'<j$ a descendant that is at the root of a copy of
$\Delta_{j'}$.

\begin{figure}[h!]
\begin{center}
\begin{tikzpicture}

\draw (0,3.5) -- (-1.2,2.5) -- (1.2,2.5) -- (0,3.5);
\draw (0,1.5) -- (-1.2,0.5) -- (1.2,0.5) -- (0,1.5);
\draw (0,0.5) -- (-1.2,-0.5) -- (1.2,-0.5) -- (0,0.5);



\node[lab] at (-0.4,2.7) {$\Delta_m$};
\node[lab] at (-0.4,0.7) {$U_{1,m}$};
\node[lab] at (-0.4,-0.3) {$\Delta_{m}$};

\draw[decorate,decoration={zigzag}] (0,0.5) -- (0,1.5);
\draw[decorate,decoration={zigzag},dotted,thick] (0,1.5) -- (0,2.5);
\draw[decorate,decoration={zigzag}] (0,2.5) -- (0,3.5);
\draw[decorate,decoration={zigzag}] (0,-0.5) -- (0,0.5);

\draw[decorate,decoration={brace},very thick] (-1.3,-0.5) to
coordinate[sloped,above] (nm) (-1.3,3.5);
\node[align=center,anchor=east] (lj) at ($(nm) - (1.2,0.0)$) {skeleton
  nodes:\\nesting level $m$};
\draw[very thick,->] (lj) to ($(nm)-(0.2,0.0)$);

\draw[decorate,decoration={zigzag},dotted,thick] (0,-0.5) -- (0,-1.3);
\draw (0,-1.3) -- (-1.2,-2.3) -- (1.2,-2.3) -- (0,-1.3);
\draw[decorate,decoration={zigzag}] (0,-1.3) -- (0,-2.3);

\node[lab] at (-0.4,-2.1) {$\Delta_j$};

\draw[decorate,decoration={brace},very thick] (-1.3,-2.3) to
coordinate[sloped,above] (nj) (-1.3,-1.3);
\node[align=center,anchor=east] (lj) at ($(nj) - (1.2,0.0)$) {skeleton
  nodes:\\nesting level $j$};
\draw[very thick,->] (lj) to ($(nj)-(0.2,0.0)$);

\draw (0,-3.1) -- (-1.2,-4.1) -- (1.2,-4.1) -- (0,-3.1);
\draw (0,-4.1) -- (-1.2,-5.1) -- (1.2,-5.1) -- (0,-4.1);

\node[lab] at (-0.4,-3.9) {$\Delta_0$};
\node[lab] at (0.0,-4.9) {$s_1$};

\draw[decorate,decoration={zigzag},dotted,thick] (0,-2.3) -- (0,-3.1);
\draw[decorate,decoration={zigzag}] (0,-3.1) -- (0,-4.1);

\draw[decorate,decoration={brace},very thick] (-1.3,-5.1) to
coordinate[sloped,above] (n0) (-1.3,-3.1);
\node[align=center,anchor=east] (l0) at ($(n0) - (1.2,0.0)$) {skeleton
  nodes:\\nesting level $0$};
\draw[very thick,->] (l0) to ($(n0)-(0.2,0.0)$);


\node[lab] at (0.0,-5.6) {\huge $S_1$};









\draw[dotted,thick] (0.3,2.9) to [out=-60,in=90] (3.2,2.0);

\begin{scope}[xshift=3.2cm,yshift=1.0cm]
\draw (0,1.0) -- (-1.2,-0.0) -- (1.2,-0.0) -- (0,1.0);
\draw (0,-0.0) -- (-1.2,-1.0) -- (1.2,-1.0) -- (0,-0.0);
\draw (0,-1.8) -- (-1.2,-2.8) -- (1.2,-2.8) -- (0,-1.8);
\draw (0,-3.6) -- (-1.2,-4.6) -- (1.2,-4.6) -- (0,-3.6);
\draw (0,-4.6) -- (-1.2,-5.6) -- (1.2,-5.6) -- (0,-4.6);


\node[lab] at (-0.4,0.2) {$U_{i,n}$};
\node[lab] at (-0.4,-0.8) {$\Delta_n$};
\node[lab] at (-0.4,-2.6) {$\Delta_{j'}$};
\node[lab] at (-0.4,-4.4) {$\Delta_0$};
\node[lab] at (0.0,-5.4) {$s_1$};

\draw[decorate,decoration={zigzag}] (0,1.0) -- (0,0.0);
\draw[decorate,decoration={zigzag}] (0,-0.0) -- (0,-1.0);
\draw[decorate,decoration={zigzag},dotted,thick] (0,-1.0) -- (0,-1.8);
\draw[decorate,decoration={zigzag}] (0,-1.8) -- (0,-2.8);
\draw[decorate,decoration={zigzag},dotted,thick] (0,-2.8) -- (0,-3.6);
\draw[decorate,decoration={zigzag}] (0,-3.6) -- (0,-4.6);

\draw[decorate,decoration={brace},very thick] (1.3,1.0) to
coordinate[sloped,above] (nn) (1.3,-1.0);
\node[align=center,anchor=west] (ln) at ($(nn) + (1.2,0.0)$) {skeleton
  nodes:\\nesting level $n < m$};
\draw[very thick,->] (ln) to ($(nn)+(0.2,0.0)$);

\draw[decorate,decoration={brace},very thick] (1.3,-1.8) to
coordinate[sloped,above] (nj) (1.3,-2.8);
\node[align=center,anchor=west] (lj) at ($(nj) + (1.2,0.0)$) {skeleton
  nodes:\\nesting level $j'$};
\draw[very thick,->] (lj) to ($(nj)+(0.2,0.0)$);

\draw[decorate,decoration={brace},very thick] (1.3,-3.6) to
coordinate[sloped,above] (n0) (1.3,-5.6);
\node[align=center,anchor=west] (l0) at ($(n0) + (1.2,0.0)$) {skeleton
  nodes:\\nesting level $0$};
\draw[very thick,->] (l0) to ($(n0)+(0.2,0.0)$);







\end{scope}

\end{tikzpicture}
\end{center}
\caption{Illustration of the notion of nesting level in the proof of Lemma~\ref{lemma-winning}.}\label{fig-nlevel}
\end{figure}
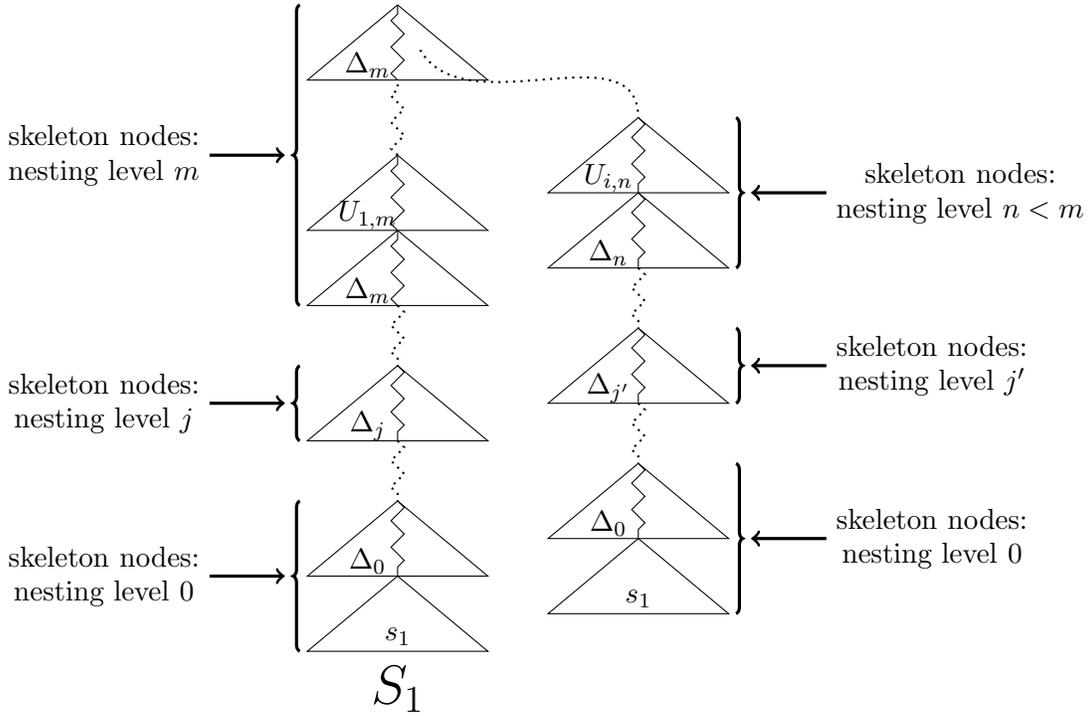

The \emph{upward number} of a node $x \in S_1$ (or $x \in S_2$) is the
number of occurrences of $\Delta_m$ in the path from $x$ to the root
of $S_1$ (see Figure~\ref{fig-unumber}). When this number is low, we
are near the roots and we need to make sure the current positions are
identical. Fortunately the two forests $S_1$ and $S_2$ are identical
up to a certain depth. This number is equal in siblings. When
moving up in the tree this number may only decrease and, by
construction, it can only decrease when traversing a copy of
$\Delta_m$ and therefore the resulting node must be on the backbone of
$S_1$ ($S_2$).

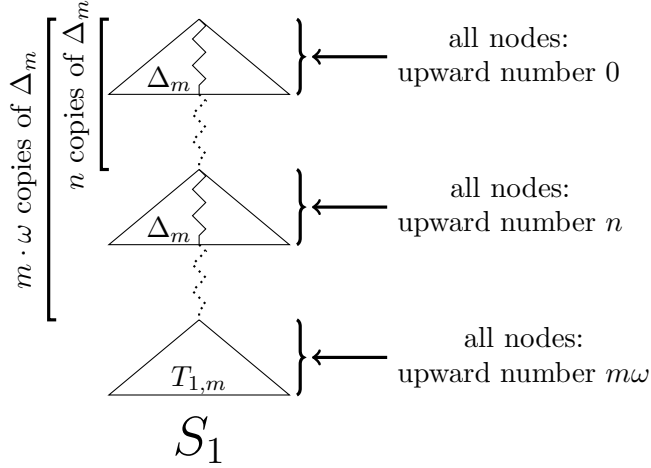
\begin{figure}[h]
\begin{center}
\begin{tikzpicture}

\draw (0,5.0) -- (-1.2,4.0) -- (1.2,4.0) -- (0,5.0);
\draw (0,3.0) -- (-1.2,2.0) -- (1.2,2.0) -- (0,3.0);
\draw (0,1.0) -- (-1.2,0.0) -- (1.2,0.0) -- (0,1.0);

\node[lab] at (-0.4,4.2) {$\Delta_m$};
\node[lab] at (-0.4,2.2) {$\Delta_m$};
\node[lab] at (-0.0,0.2) {$T_{1,m}$};

\draw[decorate,decoration={zigzag}] (0,5.0) -- (0,4.0);
\draw[decorate,decoration={zigzag},dotted,thick] (0,4.0) -- (0,3.0);
\draw[decorate,decoration={zigzag}] (0,3.0) -- (0,2.0);
\draw[decorate,decoration={zigzag},dotted,thick] (0,2.0) -- (0,1.0);

\draw[very thick] (-1.9,1.0) to (-2.0,1.0) to node[lab,sloped,above] {$m \cdot \omega$ copies of
  $\Delta_m$} (-2.0,5.0) to (-1.9,5.0);

\draw[very thick] (-1.2,3.0) to (-1.3,3.0) to node[lab,sloped,above]
{$n$ copies of $\Delta_m$} (-1.3,5.0) to (-1.2,5.0);

\draw[decorate,decoration={brace},very thick] (1.3,5.0) to
coordinate[sloped,above] (n0) (1.3,4.0);
\node[align=center,anchor=west] (l0) at ($(n0) + (1.2,0.0)$) {all
  nodes:\\upward number $0$};
\draw[very thick,->] (l0) to ($(n0)+(0.2,0.0)$);

\draw[decorate,decoration={brace},very thick] (1.3,3.0) to
coordinate[sloped,above] (nn) (1.3,2.0);
\node[align=center,anchor=west] (ln) at ($(nn) + (1.2,0.0)$) {all
  nodes:\\upward number $n$};
\draw[very thick,->] (ln) to ($(nn)+(0.2,0.0)$);

\draw[decorate,decoration={brace},very thick] (1.3,1.0) to
coordinate[sloped,above] (no) (1.3,0.0);
\node[align=center,anchor=west] (lo) at ($(no) + (1.2,0.0)$) {all
  nodes:\\upward number $m\omega$};
\draw[very thick,->] (lo) to ($(no)+(0.2,0.0)$);





\node[lab] at (0.0,-0.6) {\huge $S_1$};





\end{tikzpicture}
\end{center}
\caption{Illustration of the notion of upward number in the proof of Lemma~\ref{lemma-winning}.}\label{fig-unumber}
\end{figure}

Given a node $x \in S_1$ (or $x \in S_2$), the \emph{horizontal
  number} of $x$ is the maximal number $n \leq k$ such that for all
strict ancestors $y$ of $x$, there exists a node $z$ on the backbone
of $\Delta_m$ such that $y \psim_n^X z$. Note that this number is
equal in siblings and can only increase when going up in the
tree. Recall also that if $y$ is an \psport by pseudo-saturation $y
\psim_k^X z$ for some $z$ on the backbone of $\Delta_m$. Hence, if all
the strict ancestors of $x$ are \psports, then its horizontal number
is $k$. In particular all nodes $x$ in the skeleton of $S_1$ (or
$S_2$) have horizontal number $k$.

We now state a property \inv{n} that depends on an integer $n$ and two
nodes $x \in S_1$ and $y \in S_2$. We then show that when \inv{n+1}
holds in a game starting at $x,y$, then Duplicator can play one move
while enforcing \inv{n}. As it is easy to see that \inv{k} holds for
the roots of $S_1$ and $S_2$, this will conclude the proof of
Lemma~\ref{lemma-winning}. The inductive property \inv{n} is a
disjunction of three cases:

\begin{enumerate}
\item There exist ancestors $\hat x,\hat y$ of $x,y$ such that $\hat
  x$ and $\hat y$ have nesting level $\geq 2n+1$, upward number $\geq
  n$ and horizontal number $\geq n$. Furthermore, either 
  \begin{enumerate}
  \item $\hat x \psim^{X}_{n} \hat y$ and Duplicator has a winning
    strategy in the $n$-move game played on the \emph{subtrees at}
    $\hat x$ and $\hat y$, starting at positions $x$ and $y$, or, 
  \item Duplicator has a winning strategy in the $n$-move game played
    on the \emph{subforests of} $\hat x$ and $\hat y$, starting at
    positions $x$ and $y$. 
   \end{enumerate}

\item The nodes $x$ and $y$ are at the same relative position within
  the copy of the context $(\Delta_m)^{m\omega}$ in their respective
  forest.

\item The upward numbers of $x$ and $y$ are $\geq n$, their nesting
  levels are $\geq 2n+1$ and their horizontal number are $\geq
  n$. Moreover, we have $x \psim^X_n y$. 
\end{enumerate}

\noindent Observe that there is a factor $2$ involved in the conditions on the
nesting levels of the nodes. We need this factor in order to be able
to use Lemma~\ref{lem-aper}.

Assume we are in a situation where \inv{n+1} holds. We show how
Duplicator can play to enforce \inv{n}. The strategy depends on why
\inv{n+1} holds. In all cases we assume that Spoiler moves the pebble
from $x$ in $S_1$. The case when  Spoiler moves the pebble from $y$ in
$S_2$ is symmetrical. Recall that $n < k$, and $m=2^k > 2n$.

\subsubsection{\bf Case 1} \inv{n+1} holds because of Item~(1). 

In this case we have two nodes $\hat x$, and $\hat y$ satisfying the
properties of Item~(1).

\noindent$\bullet$ \emph{Spoiler moves from $x$ to a node that is still in the
subtree at $\hat x$.} In that case, Duplicator simply responds in the
subtree at $\hat y$ using the strategy provided by Item~(1) of
\inv{n+1} and \inv{n} holds because Item~(1) remains true.

\noindent$\bullet$ \emph{Spoiler moves to a sibling $x'$ of $\hat x$.} This can
only occur if $x = \hat x$ and $y = \hat y$. If $a)$ holds, 
by hypothesis we have $\hat x \psim^{X}_{n+1} \hat y$, therefore,
Duplicator can answer with a sibling $y'$ of $\hat y$ such that
$x'\psim^{X}_{n} y'$. Since $x'$ ($y'$) is a sibling of $\hat x$
($\hat y$), it has the same upward number, nesting level and
horizontal number as $\hat x$ ($\hat y$). Hence by hypothesis, all
those numbers satisfy Item~(3) of \inv{n} and we are done. If
$b)$ holds Duplicator simply responds in the subforest of $\hat y$
using the strategy provided by $b)$ and \inv{n} holds because Item~(1.b)
remains true.

\noindent$\bullet$ \emph{Spoiler moves to an ancestor $x'$ of $\hat x$.}

Assume first that the upward number of $x'$ is $< n$. Recall that
by hypothesis the upward number of $\hat x$ is $\geq n+1$. Hence $x'$
is on the backbone of $S_1$. As, by hypothesis, $\hat y$ has also an
upward number $\geq (n+1)$, the copy $y'$ of $x'$ in the other forest
is also an ancestor of $\hat y$. Duplicator then selects $y'$ as her
answer, satisfying Item~(2) of \inv{n}.

Assume now that the upward number of $x'$ is $\geq n$. Since the
horizontal number of $\hat x$ is $\geq n+1$, there exists a node $z$
on the skeleton of $\Delta_m$ such that $x' \psim^{X}_{n+1} z$. By
hypothesis the upward number of $\hat y$ is $\geq (n+1)$.  Hence we
can find above $\hat y$ an occurrence of $\Delta_m$ of upward number
$\geq n$. Duplicator answers by the copy of $z$ in this occurrence of
$\Delta_m$. By construction, $x',y'$ have upward numbers $\geq
n$. Moreover $x'$ ($y'$) is an ancestor of $\hat x$ ($\hat y$) and
therefore has a bigger nesting level. As by hypothesis the latter was
$\geq 2(n+1)+1$, $x'$ and $y'$ have nesting level $\geq 2n+1$. For the 
same reason the horizontal number of $x'$ is larger than the one of
$\hat x$ and is therefore $\geq n$. It follows that Item~(3) of
\inv{n} is satisfied.

\subsubsection{\bf Case 2} \inv{n+1} holds because of Item~(2). 

In this case $x$ and $y$ are at the same relative position within the
copy of the context $(\Delta_m)^{m\omega}$ in their respective forest.

\noindent$\bullet$ \emph{Spoiler moves to a node $x'$ that remains within the
context $(\Delta_m)^{m\omega}$}. Then Duplicator copies this move in
the other forest and \inv{n} is satisfied because of Item~(2).

\noindent$\bullet$ \emph{Spoiler moves to a node $x'$ that is within the
subforest $T_{1,m}$ of $S_1$.} In particular this means that $x,y$ are
on the backbones of $S_1,S_2$. By construction the subforest $T_{2,m}$
of $S_2$ contains at least one copy of the forest $T_{1,m-1}$ that can
be chosen such that all nodes occurring on the path to this copy are
\psports. It follows from Lemma~\ref{lem-aper} that there exists a
node $y'$ in $T_{1,m-1}$ such that Duplicator has a winning strategy
in the $n$-move game played on $T_{1,m}$ and $T_{1,m-1}$, starting at
positions $x'$ and $y'$. This is Duplicator's answer.

Set $\hat x$ and $\hat y$ as the roots of the copies $T_{1,m}$ and
$T_{1,m-1}$ in $S_1$ and $S_2$. Observe that by construction, $\hat x$
and $\hat y$ have nesting level $\geq m-1 \geq 2n+1$, upward number
$m\omega \geq n$ and horizontal number $k \geq n$. Moreover, by choice
of $y'$, Duplicator has a winning strategy in the $n$-move game played
on the subforest of $\hat x$ and $\hat y$, starting at positions $x'$
and $y'$. We conclude that \inv{n} holds because of Item~(1.b).

\subsubsection{\bf Case 3} \inv{n+1} holds because of Item~(3).

In this case the upward numbers of $x$ and $y$ are $\geq n+1$, their
nesting levels are $\geq 2(n+1)+1$ and their horizontal number are
$\geq n+1$. Moreover we have $x \psim^{X}_{n+1} y$.

\noindent$\bullet$ \emph{Spoiler moves horizontally.} Then Duplicator moves
according to the winning strategy provided by $x \psim_{n+1}^{X} y$
and Item~(3) of \inv{n} holds.

\noindent$\bullet$ \emph{Spoiler moves to an ancestor $x'$ of $x$.}

If the upward number of $x'$ is $< n$, as the upward number of $x\geq
n+1$, $x'$ must be on the backbone of $S_1$. Duplicator answers by the
copy $y'$ of $x'$ in the other forest, satisfying Item~(2) of
\inv{n}. Note that the upward number of $y$ is $\geq n+1$. Therefore
$y'$, having an upward number $< n$ is indeed an ancestor of $y$.

If the upward number of $x'$ is $\geq n$. By hypothesis, the
horizontal number of $x$ is $\geq n+1$, therefore, there exists a node
$z$ in the skeleton of $\Delta_m$ such that $x' \psim_{n+1}^X z$. By
hypothesis the upward number of $y$ is $\geq n+1$. Hence we can find
above $y$ an occurrence of $\Delta_m$ of upward number $n$. Duplicator
answers by the copy $y'$ of $z$ in this occurrence of $\Delta_m$. By
construction we have $x' \psim_n^X y'$. By hypothesis, for $x'$, and
by construction, for $y'$, both have upward number $\geq n$. As this
is a move up, the nesting level increases and therefore remains $\geq
2n+1$.  Hence Item~(3) of \inv{n} is satisfied.

\noindent$\bullet$ \emph{Spoiler moves down to some node $x'$.} Note that this
means that $x$ is a port-node and therefore either an \psrela or an
\psport. Moreover, since  $x \psim_{n+1}^{X} y$, this is also true of
$y$. Set $T_{i_x,j_x}$ and $T_{i_y,j_y}$ as the subforests below $x$
and $y$. Observe that by hypothesis $j_x,j_y \geq 2(n+1)$. We
distinguish two cases:

Assume first that $x$ and $y$ are both \psrelas and $i_x=i_y$. Set
$\hat x$ as $x$ and $\hat y$ as $y$. Using Lemma~\ref{lem-aper}, we
get that Duplicator wins the $(n+1)$-moves game played on the subtrees
at $\hat x$ and $\hat y$. This gives Duplicator's answer $y'$ to $x'$
and Item~(1.a) of \inv{n} holds.

Otherwise, we use pseudo-saturation to prove that there exists a
node $z$ on the backbone of $\Delta_m$ such that $z \psim_n^{X} x
\psim_n^{X} y$ and provide an answer satisfying Item~(1.b) of \inv{n}
for Duplicator.

When either $x$ or $y$ is an \psport node, the existence of $z$ is
immediate from pseudo-saturation and transitivity of $\psim_n^{X}$.
In the only remaining case, $x$ and $y$ are both \psrelas and $i_x
\neq i_y$. Therefore, Spoiler is allowed to use the safety move in the
pseudo-\relaxedX game $x \psim_{n+1}^{X} y$, and we get a \psport $z'$
in the \shal of $y$ such that $z' \psim_{n}^{X} x$. By
pseudo-saturation we then obtain $z$ on the backbone of $\Delta_m$
such that $z \psim_{k}^{X} z'$. By transitivity, we get 
that $z \psim_n^{X} x \psim_n^{X} y$.

We can now describe Duplicator's answer. By hypothesis, $x'$ belongs
to $T_{i_x,j_x}$ and by definition $T_{i_y,j_y}$ contains at least one
copy of the forest $T_{i_x,j_y-1}$ that can be chosen such that all
nodes occurring on the path to this copy are \psports. Since $j_x,j_y
\geq 2(n+1)$, it follows from Lemma~\ref{lem-aper} that there exists a
node $y'$ in $T_{i_x,j_y-1}$ such that Duplicator has a winning
strategy in the $n$-move game played on $T_{i_x,j_x}$ and
$T_{i_x,j_y-1}$, starting at positions $x'$ and $y'$. This is
Duplicator's answer. 

Set $\hat x$ and $\hat y$ as the roots of the copies $T_{i_x,j_x}$ and
$T_{i_x,j_x-1}$ in $S_1$ and $S_2$. Observe that by definition, $\hat x$
and $\hat y$ have nesting level $j_x,j_y-1 \geq 2n+1$ and upward
numbers $\geq n+1 > n$ (the same as $x,y$). Moreover, all ancestors of
$\hat x, \hat y$ are either ancestors of $x,y$, \psports or $x,y$
themselves. Since we proved that there exists $z$ on the backbone of 
$\Delta_m$ such that $z \psim_n^{X} x \psim_n^{X} y$, it follows that
$\hat x, \hat y$ have horizontal number $\geq n$. Finally, by choice
of $y'$, Duplicator has a winning strategy in the $n$-move game played 
on the subforests of $\hat x$ and $\hat y$, starting at positions $x'$
and $y'$. We conclude that \inv{n} holds because of Item~(1.b).

This concludes the proof of Lemma~\ref{lemma-winning} and therefore
the proof of Proposition~\ref{prop-saturation}.

\section{Sufficiency of the properties}\label{main-proof}

For this section we fix a regular forest language $L$ recognized by a
morphism $\alpha: \A^\mydelta \rightarrow (H,V)$ into a 
finite forest algebra $(H,V)$. Assume that $H$ and $V$ satisfy
Identities~\eqref{eqh} and~\eqref{eqv} and that the leaf completion of
$\alpha$ is closed under saturation. We prove that any language
recognized by $\alpha$, including $L$, is definable in \FOd,
concluding the proof of Theorem~\ref{th-efmax}.

Recall that given a forest $s$ (a context $p$) we refer to its image
by $\alpha$ as the \emph{forest type} of $s$ (the \emph{context type}
of $p$). In view of Lemma~\ref{lem:leafsurj}, we assume without loss
of generality that $\alpha$ itself is leaf surjective and closed under
saturation. By definition, this implies in particular that for each $h
\in H$ there exists a tree consisting of a single node whose forest
type is $h$.

As mentioned earlier, we will often manipulate \shals modulo \mequivk
for some fixed integer $k$. We start by defining a suitable $k$. Given
a \shal $q$ and a forest $s$ we denote by $q[\bar s]$ the forest
constructed from $q$ by placing $s$ at each port of $q$.

\begin{lemma} \label{lem-choicek}
There exists a number $k'$ such that for all  $k\geq k'$, for all
\shals $p \mequivk p'$ and for all forests $s$, $p[\bar s]$ and
$p'[\bar s]$ have the same forest type.
\end{lemma}

\begin{proof}
This is a consequence of Theorem~\ref{thm-word-fod} and the fact that
$H$ satisfies Identity~\eqref{eqh}. Consider strings over $H$ as
alphabet and the natural morphism $\beta : H^{+} \rightarrow H$. Since
$H$ satisfies Identity~\eqref{eqh}, it follows from
Theorem~\ref{thm-word-fod} that for every $h \in H$, $\beta^{-1}(h)$
is definable using a formula of $\varphi_h$ of \FOdw. We choose $k'$
as the maximal rank of all these formulas.

Let $k \geq k'$ and take $p \mequivk p'$ and $s$ some forest. Let
$t_1,\dots,t_n$ be the sequence of trees occurring in $p[\bar s]$ and
$t'_1,\dots,t'_{n'}$ be the sequence of trees 
occurring in $p'[\bar s]$.  For all $i$ let $h_i = \alpha(t_i)$ and
$h'_i=\alpha(t'_i)$. As $p\mequiv{k} p'$ the strings $h_1\dots h_n$
and $h'_1\dots h'_{n'}$ satisfy the same formulas of \FOdw of rank
$k'$ over the alphabet $H$.  Let $h = \beta(h_1\dots h_n)$, by our
choice of $k'$ it follows that $h'_1\dots h'_{n'} \models
\varphi_h$. Hence $\beta(h_1\dots h_n) = h = \beta(h'_1\dots
h'_{n'})$. Therefore $\alpha(p[\bar s])=\alpha(p'[\bar s])$.
\end{proof} 

As $\alpha$ is closed under saturation, there is an integer $k''$ such
that $\alpha$ is closed under $k''$-saturation. We set $k$ as the
maximum of $k'$ as given by Lemma~\ref{lem-choicek} and
$k''$. By Lemma~\ref{lem-sat-k}, $\alpha$ remains closed
under $k$-saturation. Recall that a set $P$ of \shals is $k$-definable if
it is a union of equivalence classes of $\mequivk$.

Recall that $V^{\one}$ is the monoid obtained from $V$ by adding a
neutral element $\one_V$.
For each $h \in H$, $v \in V^{\one}$ and each 
set $P$ of \shals let
\[
L^P_{v,h}=\{t ~|~ v\alpha(t)=h \text{ and $t$ is $P$-valid}\}
\]
Our goal in this section is to show that:

\begin{proposition}\label{main-prop}
For all $h \in H$, all $v\in V^{\one}$ and all 
sets $P$ of \shals, there exists a language definable in \FOd that
agrees with $L^P_{v,h}$ on $P$-valid forests.
\end{proposition}

Theorem~\ref{th-efmax} is a direct consequence of
Proposition~\ref{main-prop}. Let $L'$ be the union of all definable languages
resulting from applying Proposition~\ref{main-prop} to all $L^P_{v,h}$
where $h \in \alpha(L)$, $v = 1_V$ and $P$ is the set of all \shals. By
definition $L'$ is definable in \FOd and agrees with $L$ on all $P$-valid
forests. Hence $L = L' \cup \{a \in A \mid a \in L\}$ which is
definable in \FOd.

The remainder of this section is devoted to the proof of
Proposition~\ref{main-prop}. Assume that $v,h$ and $P$ are fixed as in the
statement of the proposition, we prove that there exists a definable language
that agrees on $L_{v,h}^P$ on $P$-valid forests. We begin by considering the
special case when $P$ is not branching (i.e. contains only \shals of arity $0$
or $1$). In that case we conclude directly by
applying Theorem~\ref{thm-word-fod}.

\subsection{Special Case: $P$ is not branching} In this case we treat
our forests as strings and use the known results on strings. Since all
\shals in $P$ have arity $0$ or $1$, any $P$-valid forest $t$ is of
the form:  
\[
c_{1} \cdots c_{k}s 
\]
where $k$ is possibly $0$ and the $c_{1},\cdots,c_{k}$ are $P$-valid
\shals of arity $1$ and $s$ a $P$-valid \shal of arity $0$. For each
$u\in V^{\one}$ and $g\in H$, consider the languages:
\begin{equation*}
\begin{split}
M_{u,g} = \set{t ~~~|~~~  t=c_{1} \cdots c_{k} s \text{ is } P\text{-valid},
\alpha(c_{1}\cdots c_{k}) = u,  \text{ and } \alpha(s)=g}
\end{split}
\end{equation*}
Notice that $L^P_{v,h}$ is the union of those languages where
$vug=h$. We show that for any $u$ and $g$, there exists a language definable in
\FOd that agrees with $M_{u,g}$ on $P$-valid forests. This will conclude this case.

By definition \shals of arity $1$ are contexts. Let $\set{v_{1},...,v_{n}}$ be
the context types that are images of \shals of arity $1$ in $P$. 

Let $P'$ be the set of \shals from $P$ of arity 1. Let $p,p' \in P'$,
by Lemma~\ref{lem-choicek} if $p \mequivk p'$ for all
forests $s$, $p[s]$ and $p'[s]$ have the same forest type. Hence, $p$
and $p'$ have the same context type. This means that for all $v_i$ the
set of \shals of context type $v_i$ is $k$-definable. Therefore, by
Claim~\ref{lemma-p-valid}, there is a formula $\theta_{v_i}(x)$ of
\FOd testing whether the \shal of $x$ has $v_i$ as forest type.

Let $\Gamma = \set{d_{1},...,d_{n}}$ be an alphabet and define a morphism
$\beta : \Gamma^{*} \rightarrow V$ by $\beta(d_{i}) = v_{i}$. Since $V$
satisfies Identity~(\ref{eqv}), for each $u \in V$ there is a \FOdw formula
$\varphi_{u}$ such that the strings of $\Gamma^{*}$ satisfying $\varphi_{u}$
are the strings of type $u$ under $\beta$. From $\varphi_{u}$ we construct a
formula $\Psi_{u}$ of \FOd defining all $P'$-valid contexts having $u$ as
context type. This is done by replacing in $\varphi_u$ all atomic formulas
$P_{d_{i}}(x)$ with $\theta_{v_i}(x)$. We can also easily define in \FOd the
set of \shals of arity $0$ such that $\alpha(s)=g$. After combining this last
formula with $\Psi_u$ we get the desired language definable in \FOd and
agreeing with $M_{u,g}$ on $P$-valid forests.

\medskip

In the remainder of the proof we assume that $P$ is branching, i.e. it
contains one \shal of arity at least~2. Recall, that by
Claim~\ref{claim-maximal}, it follows that there exists a unique
maximal $P$-reachable class $H_P$. The rest of the proof is by
induction on three parameters that we now define.

\subsection{Induction Parameters} The first and most important of our
induction parameters is the size of the set of $P$-valid forest
types. We denote this set by $X$. Observe that by definition $H_P
\subseteq X$.

\smallskip

Our second parameter is an index defined on sets $P$ of \shals. During
the proof we will construct from $P$ new sets $P'$ by replacing some
of their port-nodes with \Xnodes. Our definition ensures that the
index of $P'$ will be smaller than the index of $P$, hence guarantees
termination of the induction. It is based on following preorder on
\shals called simulation modulo~$X$.

Given two \shals $p$ and $p'$, we say that \emph{$p$ simulates $p'$
  modulo $X$} if $p'$ is obtained from $p$ by replacing some of its
port-nodes $b(\hole)$ by an \Xnode $b(a)$ with the same inner
label. Observe that simulation modulo $X$ is a partial order.

For each \shal $p$ its \emph{$X$-number} is the number of non
$\pequiv{k+2}^X$-equivalent \shals $q$ (not necessarily in $P$) that
can be simulated modulo $X$ by $p$. For each set $P$ of \shals the
\emph{$n$-index of $P$} is the number of non $\pequiv{k+2}^X$-equivalent
\shals $p \in P$ of $X$-number $n$. Our second induction parameter,
called the \emph{index of $P$}, is the sequence of its $n$-indexes
ordered by decreasing $n$.

\smallskip

The third parameter is based on $v$. Consider the preorder on context
types defined by the quotient of the $P$-reachability relation by the
$P$-equivalence relation. The \emph{$P$-depth} of a context type $v$
is the maximal length of a path in this preorder from the empty context~to~$v$.

\smallskip

We prove Proposition~\ref{main-prop} by induction on the following
three parameters, given below in their order of importance:

\begin{center}
\begin{enumerate}[label=(\roman*)]
 \item $|X|$
 \item the index of $P$
 \item the $P$-depth of $v$
 \end{enumerate}
\end{center}

We distinguish three cases: a base case and two cases in which we will 
use antichain composition and induction. We say that a context type
$u$ \emph{$P$-preserves} $v$ if $v$ is $P$-reachable from $vu$. A
context $c$ \emph{$P$-preserves} $v$ if its context-type $P$-preserves
$v$.

\subsection{Base Case: $P$ is reduced and $v$ is $P$-preserved by a
  \saturated{(P,k)} context $\Delta$}  We use saturation to prove that
$v$ is constant over the set $X = H_P$ of $P$-valid forest types,
i.e. for any $P$-valid $h_1,h_2 \in H$, $vh_1 = vh_2$. Since all forests in $L^P_{v,h}$
are $P$-valid, it follows that $L^P_{v,h}$ is either empty or the
language of all $P$-valid forests. The desired definable language is
therefore either the empty language of the language of all forests.

Since $\Delta$ $P$-preserves $v$, there exists a $P$-valid context $c$
such that $v\alpha(\Delta c) = v$. It follows that $v = v
\alpha(\Delta c)^{\omega}$. Moreover, observe that since $\Delta$ is
$P$-saturated and $c$ is $P$-valid, $\Delta c$ is $P$-saturated as
well. It then follows from saturation that

\[
vh_1 = v \alpha(\Delta c)^{\omega}h_1 = v \alpha(\Delta c)^{\omega}h_2
= vh_2
\]

This terminates the proof of the base case. We now consider two cases
in which we conclude by induction.

\subsection{Case 1: $P$ is not reduced, Bottom-Up Induction.}

By definition, since $P$ is not reduced there exists a $P$-valid
forest type $g \in X \setminus H_P$. We choose $g$ to be minimal with
respect to $P$-reachability, i.e., any $P$-valid forest type $g'$ is
either $P$-equivalent to $g$ or $g$ is not $P$-reachable from
$g'$. Let $G$ be the set of $P$-valid forest types that are
$P$-equivalent $g$. Observe that by minimality of $g$, in any
$P$-valid forest $s$ whose type is in $G$, all subforests of $s$ that are not
single leaves have a forest type in $G$ (recall that a subforest
consists of \emph{all} the children of some node). Moreover, by choice of $g$,
$G \cap H_P = \emptyset$ and all $g' \in X$ such that $G$ is
$P$-reachable from $g'$ are in $G$. We obtain the desired definable language
for $L^P_{v,h}$ via the Antichain Composition Lemma using languages that we prove to be
definable in \FOd by induction on $|X|$. Correctness of the construction relies on both
Equation~\eqref{eqh} and Equation~\eqref{eqv}. 
 
\medskip
\noindent {\bf Outline.} Our agenda is as follows. We construct from $P$ a
\kdef set $P'$ and prove that a $P$-valid forest has a type in $G$ iff it
contains no \shal of $P'$. Since \kdef sets of \shals can be expressed in \FOd,
we use $P'$ to define an antichain formula $\varphi$ which selects all
positions whose subforest contains a \shal in $P'$ (i.e. has a forest type
outside $G$) but have no descendant with that property (i.e. all descendants of
the position have a subforest of type in $G$). This formula splits $P$-valid
forests into two parts: a lower part and upper part. In the lower part all
subforests have type in $G \subsetneq X$ and in the upper part the set of valid
forest types is included in $X \setminus G$. In both cases we get definable
languages by induction on $|X|$ we glue them back together using the Antichain
Composition Lemma. This situation is illustrated in
Figure~\ref{figure-composition}.

\begin{figure}[h]
\begin{center}
\begin{tikzpicture}

\draw (0,0) -- (5,0) -- (2.5,3) -- (0,0);

\draw (0.7,0) -- (1.7,0) -- (1.2,1.3) -- (0.7,0);

\draw (2,0) -- (3.5,0) -- (2.75,1.7) -- (2,0);

\draw (3.55,0) -- (4.3,0) -- (3.925,1) -- (3.55,0);

\node[dot] (d1) at (1.2,1.3) {};

\node[dot] (d2) at (2.75,1.7) {};

\node[dot] (d3) at (3.925,1) {};

\node[bag] (labd) at (2,-1) {$(P \setminus P')$-valid, has type
in $G$ by Lemma~\ref{lemma-LG}};

\node[bag] (labu) at (4.5,3.8) {All forest types in $X \setminus G$};

\node[bag] (labf) at (5.5,2.0) {$\varphi$ holds};

\draw[arr] (labu) -> (2.5,2.5);

\draw[arr] (labd) -> (1.2,0.65);

\draw[arr] (labd) -> (2.75,0.85);

\draw[arr] (labd) -> (3.925,0.5);

\draw[arr] (labf) to [bend right] (1.2,1.3);

\draw[arr] (labf) to (2.75,1.7);

\draw[arr] (labf) to [bend left](3.925,1.0);

\end{tikzpicture}
\end{center}
\caption{Illustration of the Antichain Composition Lemma for
  Case~1. The marked are the lowest nodes whose subforest contains a
  \shal in $P'$}\label{figure-composition}
\end{figure}
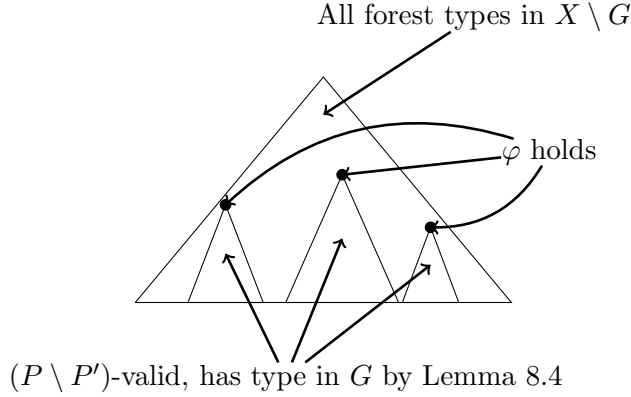

\medskip
\noindent
{\bf Definition of $P'$.} Let $s$ be some arbitrarily chosen
$P$-valid forest such that $\alpha(s) \in G$. We set
\[
P' = \set{p ~\mid~ \alpha(p[\bar{s}]) \not\in G}  
\]
We prove in the next lemma that $P'$ is well-defined, i.e. that its
definition does not depend on the choice of $s$.

\begin{lemma}\label{lem-H}
Let $p$ be a \shal of arity $n$ and $T$ and $T'$ be two sequences of
$n$ $P$-valid forests of forest type in $G$. We have:
\begin{equation*}
\alpha(p[T]) \in G \ \Leftrightarrow \  \alpha(p[T']) \in G
\end{equation*}
\end{lemma}
\begin{proof}
We use Identity~\eqref{eqv} to prove this lemma. Let $T = 
(t_1,...,t_n)$ and $T' = (t'_1,...,t'_n)$. For $i \in [1,n]$ we write
$c_i$ the context obtained from $p[T']$ by replacing $t'_i$ by a port
and $t'_j$ by $t_j$ for $j>i$. Notice that by hypothesis on $p$, $T$
and $T'$, $c_{i}$ is $P$-valid for all $i\leq n$. For all $i\leq n$,
we write $u_{i} = \alpha(c_{i})$, $h_i=\alpha(t_i)$ and
$h'_i=\alpha(t'_i)$. We first show that:
\begin{equation}\label{eq-inG}
\forall i\leq n, ~~~~u_ih_{i} \in G \ \Leftrightarrow \  u_ih'_{i} \in G
\end{equation}
Assuming that $u_i h_i \in G$, we show that $u_i h'_i \in G$. By
symmetry this will prove~\eqref{eq-inG}.  As $G$ is closed under
mutual $P$-reachability, it is enough to show that $u_ih'_i$ is
mutually $P$-reachable from $h'_i$.  By definition $u_ih'_i$ is
$P$-reachable from $h'_i$, therefore it remains to show that $h'_i$ is
$P$-reachable from $u_{i}h'_i$.  From $u_i h_i \in G$ we get that
$h'_i$ is $P$-reachable from $u_ih_i$ and therefore there is a
$P$-valid context $u$ such that $h'_i=uu_ih_i$.  By hypothesis $h_i$
is $P$-reachable from $h'_i$ and therefore there exists a $P$-valid
context $u'$ such that $h_i = u' h_i'$. A little bit of algebra and
Identity~\eqref{eqv} yields:
\begin{align*}
h'_i &=  uu_iu'h'_i& \\
 &= (uu_iu')^{\omega+1} h'_i&\\
 &= uu_i(u'uu_i)^{\omega}u'h'_i&\\
 &= uu_i(u'uu_i)^{\omega}uu_i(u'uu_i)^{\omega}u' h'_i &\textrm{\small using Identity~\eqref{eqv}} \\ 
 &= (uu_iu')^{\omega}uu_i(uu_iu')^{\omega+1} h'_i &\\
 &= (uu_iu')^{\omega}u~~~u_i h'_i &
\end{align*}
as $ (uu_iu')^{\omega}u$ is $P$-valid, $h'_i$ is $P$-reachable from
$u_ih'_i$ and~\eqref{eq-inG} is proved. 

For concluding the proof of the lemma, notice that by construction 
$\alpha(p[T]) = u_1h_1$, $\alpha(p[T'])=u_nh'_{n}$ and $u_ih'_i
=u_{i+1}h_{i+1}$. As from~\eqref{eq-inG} we get $u_ih_i \in G$ iff
$u_ih'_i \in G$, this implies by induction on $i$ that for all $i\leq
n$, $u_1h_1 \in G$ iff $u_ih_i\in G$ iff $u_ih'_i \in G$. The case
$i=n$ proves the lemma.
\end{proof}

We now prove that $P'$ can be used to test whether a $P$-valid forest has
a type in~$G$. 

\begin{lemma}\label{lemma-LG}
A $P$-valid forest has type in $G$ iff it is $(P \setminus P')$-valid.
\end{lemma}

\begin{proof}
This is a consequence of Lemma~\ref{lem-H}. Let $t$ be a $P$-valid
forest such that $\alpha(t) \not\in G$. We prove that $t$ contains a
\shal of $P'$. Consider a minimal subforest $t'$ of $t$ whose type is
not in $G$. Then we have $t' = p[T]$ where $p$ is a \shal and $T$ a
sequence of forests of forest type in $G$ (possibly empty if $p$ is of
arity 0). Let $T'$ be the sequence $\bar s$ for some $s$ with $\alpha(s)\in
G$. By Lemma~\ref{lem-H}~ $\alpha(p[T']) \not\in G$ and therefore $p \in P'$.

Conversely, if $\alpha(t) \in G$, by minimality of $G$, all
subforests of $t$ are in $G$. It is then immediate by definition of
$P'$ and Lemma~\ref{lem-H} that $t$ cannot contain a \shal of $P'$.
\end{proof}

\medskip
\noindent {\bf Setting up the Composition.} Let $\varphi$ be the antichain
formula which holds at port-nodes $(p,x)$ such that $p\in P'$ and $x$ has no
descendant with that property. It follows from the next lemma that $\varphi$ is expressible in \FOd.

\begin{lemma} \label{lem:kdef}
$P'$ is \kdef.
\end{lemma}

\begin{proof}
This is a consequence of Lemma~\ref{lem-choicek} (which is itself a
consequence of~\eqref{eqh}). Set $p \in P'$ and $p' \mequivk p$. We
prove that $p' \in P'$. By definition of $P'$, $\alpha(p[\bar{s}])
\not\in G$. As $p \mequivk p'$, by choice of $k$ and
Lemma~\ref{lem-choicek} we get $\alpha(p'[\bar{s}]) = 
\alpha(p[\bar{s}])$. Hence $\alpha(p'[\bar{s}]) \not\in G$ and $p' \in
P'$.
\end{proof}

We now define the languages that we will use to apply the Antichain Composition Lemma.

\begin{lemma}\label{lemma-Lg}
For any $g\in G$,
there exists a language definable in \FOd that
agrees with $L^P_{1_V,g}$ on $P$-valid forests.
\end{lemma}
\begin{proof}
  Notice that the set of all $(P \setminus P')$-valid forest types is $G
  \subsetneq X$. Hence by induction on the first parameter in
  Proposition~\ref{main-prop} there exists a language $L_1$ definable in \FOd
  that agrees with $L^{(P\setminus P')}_{1_V,g}$ on $(P\setminus P')$-valid
  forests. By Lemma~\ref{lemma-LG}, a $P$-valid forest has type $g \in
  G$ iff it is $(P \setminus P')$-valid. Hence $L_1$ agrees with
  $L^{P}_{1_V,g}$ on  $P$-valid forests.
\end{proof}

Assume $G= \set{g_1,\cdots,g_n}$. For all $i
\leq n$, let $L_i$ be a language definable in \FOd that agrees with
$L^P_{\one_V,g_i}$ on $P$-valid forests given by Lemma~\ref{lemma-Lg}.

Let $Q$ be the set of \shals $q$ that can be obtained from some $p \in P$ by replacing
some port-nodes (possibly none) with $G$-nodes of the same inner label
and such that either: 
\begin{itemize}
\item $q$ has arity greater than $1$ (i.e. one port-node of $p$ was
  left unchanged)
\item or $q$ has arity $0$ and $p \in P'$ (hence $\alpha(q) \not\in
  G$).
\end{itemize}

We have:

\begin{lemma}\label{lemma-XG}
There is a language $K$ definable in \FOd that agrees with $L^Q_{v,h}$ on $Q$-valid forests.
\end{lemma}

\begin{proof}
Let $Y$ be the set of $Q$-valid forest types. We prove that $Y
\subseteq X$ and $Y \cap G = \emptyset$. It will follow that $|Y| <
|X|$. Hence $K$ is obtained by applying Proposition~\ref{main-prop}
by induction on the first parameter.

Let $h \in Y$, by definition, there exists a $Q$-valid forest $s$ such
that $\alpha(s) = h$. All \shals $q \in Q$ occurring in $s$ are
constructed from $p \in P$ by replacing some port-nodes of $p$ with
$G$-nodes. As $G$ contains only $P$-valid forest types, for any $g \in
G$ there exists a $P$-valid forest whose type is $g$. By replacing the
newly introduced $G$-nodes in $s$ by the correspond $P$-valid forest
with the same type we get a $P$-valid forest $s'$ whose type remains
$h$. Hence $h \in X$. Moreover, for any \shal of arity $0$ occurring in
$s$, the corresponding \shal occurring in $s'$ must belong to
$P'$. It follows that $s'$ contains at least one \shal in $P'$ and by
Lemma~\ref{lemma-LG} that $h\not\in G$.
\end{proof}

\medskip
\noindent
{\bf Applying Antichain Composition.} We now apply the Antichain Composition
Lemma to the languages $K$, and $L_1 \cdots L_n$ defined above. The situation is depicted in
Figure~\ref{figure-composition}.

Recall that $G= \set{g_1,\cdots,g_n}$. For any $i \leq n$, let $a_i\in A$ be
such that $\alpha(a_i)=g_i$. Set $L = \{t \mid t[(L_1,\varphi) \rightarrow
a_1, \cdots, (L_n, \varphi) \rightarrow a_n] \in K\}$. Since $K,L_1,\dots,L_n$
are definable in \FOd, it follows from Lemma~\ref{lemma-ACL} that $L$ is
definable in \FOd. We terminate the proof by proving that $L$ agrees with
$L^{P}_{v,h}$ on $P$-valid forests.

\begin{lemma} \label{lem:thisistheend}
Let $t$ be a $P$-valid forest, then $\alpha(t) =
\alpha(t[(L_1,\varphi) \rightarrow a_1, \cdots, (L_n, \varphi)
\rightarrow a_n])$. 
\end{lemma}

\begin{proof}
This is immediate by definition of $Q$, $K$ and $L_1\cdots L_n$.
\end{proof}

\subsection{Case 2: $P$ is reduced but there exists no
  \saturated{(P,k)} context $\Delta$ that $P$-preserves $v$, Top-Down
  Induction}

In this case we use again the Antichain Composition Lemma using
languages that we prove to be definable by induction on  
the index of $P$ and the $P$-depth of $v$. Recall that since $P$ is
reduced, $X = H_P$. Correctness relies on Identity~\eqref{eqv}.

\medskip
\noindent
{\bf Outline.} We proceed as follows. First we use our hypothesis to
define a port-node $(p,x) \in P$ with the following properties. For
any $P$-valid forest $t$ and port-node $(p',x')$ of $t$ such that $(p,x)
\pequiv{k}^X (p',x')$, the context $c$ obtained from $t$ by
replacing the subforest below $x'$ by a port does not $P$-preserve
$v$. Since by Claim~\ref{lemma-types-fod} all such nodes $(p',x')$
can be defined in \FOd, this gives an antichain formula $\varphi$ which
selects such nodes having no ancestor with that property. Such
a formula splits a forest in two parts: an upper part and a lower
part. For the upper part, we will prove that the set of occurring
\shals has smaller index and use induction on that
parameter. Moreover, observe that by choice of $(p,x)$, each subforest
in the lower part is below a context that has larger $P$-depth than
$v$, we will use induction on this parameter. This situation is
depicted in Figure~\ref{fig:comptd}.

\begin{figure}[h]
\begin{center}
\begin{tikzpicture}

\node[bag] at (2.5,3.2) {$v$};

\draw (0,0) -- (5,0) -- (2.5,3) -- (0,0);

\draw (0.7,0) -- (1.7,0) -- (1.2,1.3) -- (0.7,0);

\draw (2,0) -- (3.5,0) -- (2.75,1.7) -- (2,0);

\draw (3.55,0) -- (4.3,0) -- (3.925,1) -- (3.55,0);

\node[dot] at (1.2,1.3) {};

\node[dot] at (2.75,1.7) {};

\node[dot] at (3.925,1) {};

\node[bag,align=center] (labd) at (2,-1) {\small $v^+$-equivalence classes definable in \FOd by Lemma~\ref{lemma-Lv+}};

\node[bag,align=center] (labu) at (5.5,3) {\small Set of \shals\\
  \small $P'$ with smaller index};


\draw[arr] (labu) -> (2.5,2);

\draw[arr] (labd) -> (1.2,0.65);

\draw[arr] (labd) -> (2.75,0.85);

\draw[arr] (labd) -> (3.925,0.5);




\end{tikzpicture}
\end{center}
\caption{Illustration of the Antichain Composition Lemma for Case~2. The
  marked nodes are the topmost nodes equivalent to $(p,x)$.}\label{fig:comptd}
\end{figure}
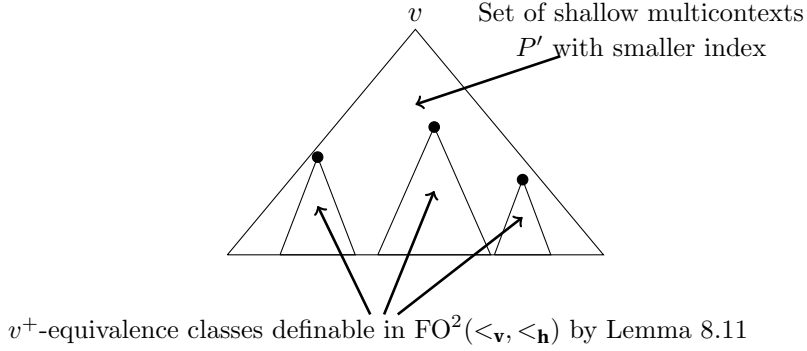

\medskip
\noindent
{\bf Definition of $(p,x)$.} Let $(p,x)$ be a port-node in $P$. We say
that $(p,x)$ is \emph{$P$-bad for $v$} iff there exists no $P$-valid
context $c$ satisfying the two following properties: 
\begin{enumerate}
\item $c$ $P$-preserves $v$.
\item the port-node $(p',x')$ above the port of $c$ verifies $(p,x)
  \pequiv{k}^X (p',x')$.
\end{enumerate}

\begin{lemma} \label{lem:delchoice}
There exists a port-node $(p,x) \in P$ that is $P$-bad for $v$.
\end{lemma}

\begin{proof}
This is where Identity~\eqref{eqv} is used. We proceed by
contradiction and assume that no port-node $(p,x) \in P$ is $P$-bad
for $v$.

By definition, for all port-nodes $(p,x) \in P$ we get a $P$-valid
context $c_{p,x}$ that $P$-preserves $v$ and such that the port-node
$(p',x')$ above the port of $c_{p,x}$ verifies $(p,x) \pequiv{k}^X
(p',x')$. Note that since $\pequiv{k}^X$ is of finite index, we may
assume that there are finitely many different contexts $c_{p,x}$ for
all $(p,x) \in P$. Let $\Delta$ be the context obtained by
concatenating all these finitely many contexts $c_{p,x}$ for all
$(p,x)$. By definition $\Delta$ is $(P,k)$-saturated. We use
Identity~\eqref{eqv} to prove that $\Delta$ $P$-preserves $v$ 
which contradicts the hypothesis of this case. This is an immediate 
consequence of the next claim.

\begin{claim} \label{clm:fo2classic}
Let $u,u' \in V$ such that both $u$ and $u'$ $P$-preserve $v$. Then
$uu'$ $P$-preserves $v$.
\end{claim}

We finish the proof of Lemma~\ref{lem:delchoice} by proving
Claim~\ref{clm:fo2classic}. By hypothesis, we have $w,w' \in V$ that
are $P$-valid and such that $vuw =v$ and $vu'w'=v$. Set $e =
(wu'w'u)^\omega$, a little algebra yields $vue = vu$. Applying
Identity~\eqref{eqv}, we get that
\[
vu=vue=vueu'w'ue=vuu'w'ue
\]
Hence $v = vuu'w'uew$ and since $w'uew$ is $P$-valid, this terminates
the proof.
\end{proof}

\medskip
\noindent
{\bf Setting-up the Composition.} For the remainder of the proof
we set $(p,x) \in P$ as a port-node which is $P$-bad for $v$ as
given by Lemma~\ref{lem:delchoice}. We define our antichain formula
$\varphi$ as the formula holding exactly at all port-nodes $(p',x')$
such that $(p,x) \pequiv{k}^X (p,x')$ and having no ancestor with that
property. By definition, $\varphi$ is antichain and by
Claim~\ref{lemma-types-fod}, $\varphi$ is expressible in \FOd. We now
define the languages $L_1,\dots,L_n$ and $K$ necessary for applying the
Antichain Composition Lemma.

Given two elements $g$ and $g'$ of $H$, we say that $g$ is
$v^+$-equivalent to $g'$ if for all context types $u$ which do not  
$P$-preserve $v$ (hence the $P$-depth of $vu$ is strictly higher that
the $P$-depth of $v$) we have $vug=vug'$. Set $\{\gamma_1,\dots,
\gamma_n\}$ as the set of all $v^+$-equivalence classes. For all $i$,
we define $M_i = \set{s \mid \alpha(s) \in \gamma_i \text{ and $s$
    $P$-valid}}$.

\begin{lemma}\label{lemma-Lv+}
For all $i$, there is a language $L_i$ definable in \FOd that agrees with $M_i$
on $P$-valid forests.
\end{lemma}
\begin{proof}
Fix a $v^+$-equivalence class $\gamma_i$ and let $g \in \gamma_i$. For
any $u$ such that $vu$ is not $P$-reachable from $v$, by induction in
Proposition~\ref{main-prop}, the third parameter has increased and the
other two are unchanged, there is a language $K_u$ definable in \FOd
that agrees with $L_{vu,vug}^P$ on $P$-valid forests. The lemma then
follows by taking for $L_i$ the intersection of all languages $K_u$
for $u$ such that $vu$ is not $P$-reachable from $v$.
\end{proof}

Let $P' = \set{p' \mid p \pequiv{k+2}^X p'}$. Observe that by
Claim~\ref{clm:global}, any \shal $p'$ in $P'$ contain at least one
position $x'$ such that $(p,x) \pequiv{k}^X (p',x')$. For $p' \in P'$,
let $x_1,\cdots,x_\ell$ be all the port-nodes of $p'$ such that $(p,x)
\pequiv{k}^X (p',x_i)$. Let $b(\hole)$ be the label of all the $x_i$
in $p'$. Let $\Delta_{p'}$ be the set of all the \shals that are
constructed from $p'$ by replacing at all the positions $x_i$,
$b(\hole)$ by a label $b(a)$ (possibly different for each position),
for $a$ such that $\alpha(a) \in X$. Let $\hat P$ be the union of all
$\Delta_{p'}$ for $p' \in P'$.  Finally, let $Q = (P \setminus P')
\cup \hat P$.

\begin{lemma}\label{lemma-LPtau}
There is a language $K$ definable in \FOd that agrees with $L^Q_{v,h}$ on
$Q$-definable forests.
\end{lemma}

\begin{proof}
Let $Y$ be the set of all $Q$-valid forest types. We first observe
that $Y \subseteq X$. The argument is similar to the one in the proof
of Lemma~\ref{lemma-XG} as the newly introduced \shals can be
represented by $P$-definable forests. If $Y \subsetneq X$, then the
lemma follows by induction on the first parameter in
Proposition~\ref{main-prop}. Otherwise $X = Y$ and we prove that $Q$
has smaller index than $P$. The result is then immediate by induction
on the second parameter in Proposition~\ref{main-prop}. Note that
this is where we use the fact that our notion of equivalence between
positions is weaker than $\mequivk$. With a stronger notion, it would
not be possible to prove that the index has decreased.

Set $n \in \nat$ as the largest integer such that there exists a \shal 
in $P'$ with $X$-number $n$. We prove that all $\hat{p} \in \hat{P}$
have a $X$-number that is strictly smaller than $n$. It will then be
immediate from the definitions that $Q$ has smaller index than $P$.

Let $\hat{p} \in \hat{P}$. By definition, there exists $p' \in P'$ and 
$x_1,\cdots,x_\ell \in p'$ such that for all $i$, $(p,x) \pequiv{k}^X
(p',x_i)$ and replacing the labels $b(\hole)$ at all positions $x_i$
in $p'$ by $b(a)$ for $\alpha(a) \in X$ yields $\hat{p}$. In
particular this means that $p'$ simulates $\hat{p}$ modulo $X$ and
that the $X$-number of $\hat{p}$ is smaller or equal to that of $p'$
and hence smaller or equal to $n$. We prove that $\hat{p}$ does not
simulate any $p'' \pequiv{k+2}^X p'$ modulo $X$. It will follow that
the inequality is strict which terminates the proof.

We proceed by contradiction, assume that there exists $p''
\pequiv{k+2}^X p'$ such that $\hat{p}$ simulates $p''$ modulo $X$. By
definition, $p'' \pequiv{k+2}^X p' \pequiv{k+2}^X p$, hence, by
Claim~\ref{clm:global}, $p''$ contains a port-node $x''$ such that
$(p'',x'') \pequivk^X (p,x)$. By definition of simulation, $x''$
corresponds to a port-node $\hat{x}$ in $\hat{p}$ and a port-node $x'$
in $p'$. Moreover, since $\hat{x}$ is a port-node, this means that $x'
\not\in \{x_1,\cdots,x_\ell\}$, i.e. $(p',x')$ is not
$\pequivk^X$-equivalent to $(p,x)$. This contradicts the following
claim. 

\begin{claim} \label{clm:finish}
$(p'',x'') \pequiv{k}^X (p',x') \pequiv{k}^X (p,x)$.
\end{claim}

It remains to prove Claim~\ref{clm:finish}. We prove that $(p'',x'')
\pequiv{k}^X (p',x')$. By definition, $p'$ and $p''$ use the same set
of labels in \As and $x',x''$ have the same label $b(\hole)$. We give
a winning strategy for Duplicator in the $X$-relaxed game between
$(p'',x'')$ and $(p',x')$. By definition $p''$ is obtained from $p'$
by replacing some port nodes with \Xnodes with the same inner-node
label. Therefore as long as Spoiler does not use a safety move,
Duplicator can answer by playing the isomorphism.  Assume now that
Spoiler does a safety move. Then the pebbles are on positions $z' \in
p'$ and $z'' \in p''$ with labels $b(\hole)$ and $b(a)$ as
Duplicator's strategy disallow any other possibility such as
$b(a),b(a')$ where $a \neq a'$. Is Spoiler selects $z''$, then
Duplicator continues to play the isomorphism by leaving the other
pebble on $z'$. If Spoiler selects $z'$, observe that since 
$p' \pequiv{k+2}^X p''$, we can use Claim~\ref{clm:global} and get a
node $y'' \in p''$ such that $(p',x') \pequiv{k}^X (p'',y'')$, this is
Duplicator's answer. Duplicator can then continue to play by using the
strategy given by  $(p',x') \pequiv{k}^X (p'',y'')$.
\end{proof}

\medskip
\noindent {\bf Applying Antichain Composition.} Let $K$ and $L_1 \cdots L_n$ be
languages definable in \FOd as given by Lemma~\ref{lemma-LPtau} and
Lemma~\ref{lemma-Lv+}. For all $i$ let $a_i\in A$ be such that $\alpha(a_i)\in \gamma_i$.
Set $L = \{t \mid t[(L_1,\varphi) \rightarrow a_1,
\cdots, (L_k, \varphi) \rightarrow a_n] \in K\}$. It follows from 
Lemma~\ref{lemma-ACL} that $L$ is definable in \FOd.  We terminate the proof by
proving that $L$ agrees with $L^{P}_{v,h}$ on $P$-valid forests.

\begin{lemma}\label{lemma-v+replace}
For any $P$-valid forest $t$, $v\alpha(t)=v\alpha(t[(L_1,\varphi)
\rightarrow a_1, \cdots, (L_k, \varphi) \rightarrow a_k])$.
\end{lemma}
\begin{proof}
  This is because $(p,x)$ is $P$-bad for $v$. The proof goes by induction on
  the number of occurrences in $t$ of port-node $(q,y)$ such that $(q,y)
  \pequiv{k}^X (p,x)$. If there is no occurrence, this is immediate as the
  substitution does nothing.

Consider a node $y$ of a \shal $q$ such that $(q,y) \pequiv{k}^X
(p,x)$ and no node above $y$ satisfies that property. Let $s$ be
the subforest below $y$ in $t$ and let $i$ be such that $\alpha(s) \in
\gamma_i$. Let $c$ be the context formed from $t$ by replacing $s$ by
a port and let $u_c$ be its type. Since $(p,x)$ is $P$-bad for $v$,
$u_c$ does not $P$-preserve $v$. Hence, $v\alpha(t) = vu_c\alpha(s) =
vu_c \alpha(a_i)$ by definition of $v^+$-equivalence. 

We write $t' = ca_i$, we already know that $v\alpha(t') = v\alpha(t)$. Observe
that by construction $t'[(L_1,\varphi) \rightarrow a_1, \cdots, (L_k, \varphi)
\rightarrow a_k]$ is $t[(L_1,\varphi) \rightarrow a_1, \cdots, (L_k, \varphi)
\rightarrow a_k]$. By induction we have that $v \alpha(t') = v
\alpha(t'[(L_1,\varphi) \rightarrow a_1, \cdots, (L_k, \varphi) \rightarrow
a_k])$ which terminates the proof.
\end{proof}

\section{Decidability}\label{more}

In this section we prove that the characterization of \FOd given in
Theorem~\ref{th-efmax} is decidable.

\begin{theorem}\label{cor-decid}
Let $L$ be a regular language of forests. It is decidable 
whether $L$ is definable in \FOd.
\end{theorem}

In view of Theorem~\ref{th-efmax} the decision procedure works as follows. From
$L$ we first compute its syntactic morphism $\alpha: \A^\mydelta \rightarrow
(H,V)$. Then we check that~\eqref{eqh} holds in $H$, that~\eqref{eqv} holds in
$V$ and that $\alpha$ is closed under saturation. This is straightforward
for~\eqref{eqh} and~\eqref{eqv} as $H$ and $V$ contain only finitely many
elements. However it is not obvious from the definitions that closure under
saturation can be decided. The main result of this section is an algorithm
which, given as input a morphism $\alpha$, decides whether $\alpha$ is closed
under saturation.

Recall the definition of saturation. It requires the existence of a number $k$
such that for all branching and reduced sets $P$ of \shals and all
\saturated{(P,k)} contexts a property holds. The main problem is that all these
quantifications range over infinite sets. In the first part of this section we
introduce an ``abstract'' version of these sets with finitely many objects
together with an associated ``abstract'' notion of saturation and show that
closure under saturation corresponds to closure under the abstract notion of
saturation.

Then, in the remaining part of the section we present an algorithm
that computes the needed abstract sets.

\subsection{Abstraction}

Let $\A = (A,B)$ be a finite alphabet and $\alpha: \A^\mydelta
\rightarrow (H,V)$ be a morphism into a finite  forest algebra
$(H,V)$. Recall that we see \shals as strings over $\As$, i.e. as
elements of $\As^+$. In order to stay
consistent with our notation on \shals, we will denote by $+$ the
concatenation within $\As^+$. Recall that if $Q \subseteq \As^+$ is a set of
\shals, then we write $(p,x) \in Q$ instead of $x$ is a node of some \shal $p\in Q$.

We start with some terminology. Let $p$ be a \shal of arity $n$ and
let $G \subseteq H$. We denote by $p(G)$ the set of forest types $h
\in H$ such that there exists a sequence $T$ of $n$ forests which all
have a type in $G$ and such that $\alpha(p[T])=h$. For a port-node $x$
of $p$, we denote by $p(G,x)$ the set of context types $v \in V$ such
that there exists a sequence $T$ of $n-1$ forests which have all a
type in $G$ such that $\alpha(p[T,x]) = v$. If $x$ is not a port-node
of $p$ then we set $p(G,x) = \emptyset$ for all $G$. The following
fact is immediate. 

\begin{fct} \label{fct:gentypes}
Let $(p,x)$ and $(q,y)$ be nodes and $r = p + q$. Then for
any $G \subseteq H$, $r(G) = p(G) + q(G)$, $r(G,x) = p(G,x) + q(G)$
and $r(G,y) = p(G) + q(G,y)$.
\end{fct}

\medskip
\noindent
{\bf Abstracting \shals: Profiles.} We now define an abstract version
of positions in \shals that we call \emph{profiles}.

Consider a pair $(p,x)$ where $p$ is a \shal and $x$ a position in
$p$. The profile of $(p,x)$, denoted $\beta(p,x)$ is the quadruple
$\fv = (i,\Bs,f_H,f_V)$ where
\begin{enumerate}
\item $i \in \{0,1,2\}$ is the \emph{arity} of $p$ counted up to
  threshold $2$,
\item $\Bs \subseteq \As$ is the \emph{alphabet} $p$, i.e. the set of
  labels used in $p$, 
\item $f_H: 2^H \rightarrow 2^H$ is the \emph{forest mapping} of $p$,
  defined as the mapping $G \mapsto p(G)$,
\item $f_V : 2^H \rightarrow 2^V$ is the \emph{context mapping} of
  $(p,x)$, defined as the mapping $G \mapsto p(G,x)$.
\end{enumerate}
Observe that if $p$ has arity $0$ (i.e. $p$ is a forest) then $f_H$ is
the mapping $G \mapsto \{\alpha(p)\}$. Moreover, whenever $x$ is not a
port-node, $f_V$ is the mapping $G \mapsto \emptyset$. We let $\fP$ be
the set of profiles of all \shals. Observe that \fP is finite:
\[
\fP \subseteq \{0,1,2\} \times 2^{\As} \times (2^H)^{2^H} \times (2^V)^{2^H}
\]
In the rest of this section we shall denote by $\fu,\fv,\dots$ profiles
(elements of \fP), by $\fU,\fV, \dots$ sets of profiles (subsets of \fP)
and by $\Uc,\Vc,\dots$ sets of sets of profiles (subsets of
$2^{\fP}$).

Let us first present two semigroup operations for \fP. Both operations
are adapted from the concatenation operation between \shals. If
$(p,x)$ and $(p',x')$ are pairs where $p,p'$ are \shals and $x,x'$ are
positions of $p,p'$, then one can use concatenation to construct two
new pairs: $(p+p',x)$ in which we keep the position $x$ of $p$ and
$(p+p',x')$ in which we keep the position $x'$ of $p'$.

Abstracted on profiles, this yields the two following operations. Let
$\fv,\fv' \in \fP$ be two profiles and set $(i,\Bs,f_H,f_V) = \fv$ and
$(i',\Bs',f'_H,f'_V) = \fv'$. We define, two new profiles $\fv +_\ell
\fv' \in \fP$ and $\fv +_r \fv' \in \fP$ as follows,
\begin{center}
\begin{minipage}{0.46\linewidth}
 $\fv +_\ell \fv' = (j,\Cs,g_H,g_V)$ with
\begin{itemize}
\item $j = min(i+i',2)$
\item $\Cs = \Bs \cup \Bs'$.
\item $g_H: G \mapsto f_H(G) + f'_H(G)$.
\item $g_V: G \mapsto f_V(G) + f'_H(G)$. 
\end{itemize}
\end{minipage}
\begin{minipage}{0.46\linewidth}
 $\fv +_r \fv' = (j,\Cs,g_H,g'_V)$ with
\begin{itemize}
\item $j = min(i+i',2)$
\item $\Cs = \Bs \cup \Bs'$.
\item $g_H: G \mapsto f_H(G) + f'_H(G)$.
\item $g'_V: G \mapsto f_H(G) + f'_V(G)$. 
\end{itemize}
\end{minipage}
\end{center}

On the \shal level, the definition exactly means that for any
$(p,x)$, $(p',x') \in \As^+$ such that $\fv = \beta(p,x)$ and $\fv' =
\beta(p',x')$, we have $\fv +_\ell \fv' = \beta(p+p',x)$ and $\fv +_r
\fv' =  \beta(p+p',x')$. One can verify that 
$+_r$ and $+_\ell$ are both semigroup operations. Moreover, the
following fact is immediate from the definitions and states that one
can use the operations $+_\ell$ and $+_r$ to compute the whole set \fP
from the profiles of one-letter \shals.

\begin{fact} \label{fct:computeP}
\fP is the smallest subset of $\{0,1,2\} \times 2^{\As} \times
(2^H)^{2^H} \times (2^V)^{2^H}$ such that:
\begin{itemize}
\item \fP contains the profiles of one-letter \shals: for all $c \in
  \As$, $\beta(c,x) \in \fP$ (where $x$ is the unique position in $c$)
\item \fP is closed under $+_\ell$.
\item \fP is closed under $+_r$.
\end{itemize}

\end{fact}

\medskip
\noindent
{\bf Abstracting sets of \shals: Configurations.} Recall the
definition of saturated contexts: let $P$ be a set of \shals, a
context is $(P,k)$-saturated iff it is $P$-valid and for all $(p,x)
\in P$, there exists a $\pequiv{k}^X$-equivalent position on the
backbone of the context. This means that we need to define an
abstraction of sets of \shals $P$ that contains two informations:
\begin{itemize}
\item the set of $P$-valid types.
\item the set of images under $\alpha$ of $(P,k)$-saturated contexts.
\end{itemize}
For this we introduce the notion of \emph{configurations}. Notice that this abstraction
needs to be parametrized by the equivalence $\pequiv{k}^X$. In order
to do this, we will abstract this equivalence on profiles which are
our abstraction of the objects compared by $\pequiv{k}^X$.

There is an issue however. Intuitively, we want two profiles \fu and
\fv to be ``equivalent'' if one can find $(p,x)$ and $(q,y)$ such that
$(p,x) \pequiv{k}^X (q,y)$, $\beta(p,x) = \fv$ and $\beta(q,y) =
\fu$. Unfortunately, this is not the right definition as the relation
we obtain is not transitive in general and hence not an equivalence
anymore. This is a problem since the definition of a saturated context
requires to pick one position $(p,x)$ among a \emph{set} of equivalent
ones. We solve this problem by abstracting sets of equivalent
positions directly by sets of profiles.

Moreover, if $P \subseteq \As^+$, the configuration that abstracts $P$
needs to have exhaustive information about \emph{all} sets of
equivalent positions that can be found in $P$. Therefore we define a
configuration as a sets of sets of profiles, i.e. an element of the 
set:
\[
\bcC = 2^{2^{\fP}}
\]
Of course, we are only interested in elements of \bcC that correspond
to actual sets of \shals. Let $k \in \nat$ and $X \subseteq H$, we let
$\bcI_k[\alpha,X]$ be the set of \emph{$(X,k)$-relevant}
configurations: 
\begin{align*}
  \bcI_k[\alpha,X] &= \{ \Vc \in \bcC ~|~  \exists Q \subseteq \As^+ \text{
    such that }\\
  & 1.~~  \forall q,q' \in Q, \exists x,x' \in q,q' \text{ s.t. }  (q,x) \pequiv{k}^X (q',x')\\
   & 2.~~ \forall (q,x) \in Q ~ \text{ there exist } (q_1,x_1) \pequiv{k}^X \cdots \pequiv{k}^X
    (q_n,x_n) \pequiv{k}^X (q,x) \in Q \text{ such that }\\
   & \text{\ ~~~~~\ ~~~} \set{\beta(q_1,x_1),\dots,\beta(q_n,x_n)} \in \Vc\\
    & 3.~~ \forall \fV \in \Vc \text{ there exist } (q_1,x_1) \pequiv{k}^X \cdots \pequiv{k}^X
    (q_n,x_n) \in Q \text{ such that }\\
& ~~~~~\fV = \set{\beta(q_1,x_1),\dots,\beta(q_n,x_n)} \\\}
\end{align*}
Note that condition~(1) restricts the definition to sets of \shals
that are $\pequiv{k}^X$-equivalent. This will later be necessary
when computing the sets of relevant configurations. However,
when considering saturation, we will actually work with unions of
relevant configurations.

This definition takes care of the quantification over the infinite set
of sets of \shals in the definition of saturation. One quantification
still needs to be dealt with: quantification over $k \in \nat$. We
achieve this by defining the set of \emph{$X$-relevant}
configurations as the intersection of the previous sets for all $k$:
\begin{equation*}
\bcI[\alpha,X] = \bigcap_k \bcI_k[\alpha,X]
\end{equation*}
We will present an algorithm for computing $\bcI[\alpha,X]$ in the
second part of this section. The following fact is immediate
from the definitions: 

\begin{fct} \label{fct:pro3}
For any $k,k' \in \nat$ such that $k \leq k'$ and $X \subseteq H$,
$\bcI[\alpha,X] \subseteq \bcI_{k'}[\alpha,X] \subseteq
\bcI_{k}[\alpha,X]$. In particular, there exists $\ell \in \nat$ such
that for all $k \geq \ell$ and all $X \subseteq H$, $\bcI[\alpha,X] =
\bcI_k[\alpha,X]$.
\end{fct}

Note that while proving the existence of $\ell$ in Fact~\ref{fct:pro3}
is simple, computing an actual bound on $\ell$ is more difficult and
will be a consequence of our algorithm computing $\bcI[\alpha,X]$.

Finally we equip $\bcC$ with a semigroup operation by generalizing the
operations $+_\ell$ and $+_r$ defined on \fP. Observe that $+_\ell$
and $+_r$ can be generalized on sets of profiles: we define the sum of
two sets as the set of all possible sums of elements of the two
sets. If $\Uc,\Vc \in \bcC$, we can now define $\Uc + \Vc$ as the set
\[
\{\fU +_\ell \bigcup_{\fV \in
  \Vc} \fV \mid \fU \in \Uc\} \cup
\{\bigcup_{\fU \in
  \Uc} \fU +_r \fV \mid \fV \in \Vc\}  
\]
The following fact can be verified from the definitions.
\begin{fct} \label{fct:prosemi2} $(\bcC,+)$ is a semigroup.
\end{fct}

\medskip
\noindent {\bf Validity and Reachability for Configurations.} Let
$\Vc$ be a configuration and let $\fV \in 2^\fP$ be the union of all
sets in $\Vc$. The set of $\Vc$-valid forest types is the smallest $X
\subseteq H$ such that for every $(i,\Bs,f_H,f_V) \in \fV$, $f_H(H)
\subseteq X$ when $i = 0$ and $f_H(X) \subseteq X$
otherwise. $\Vc$-valid context types are defined as the smallest $Y
\subseteq V$ such that $Y \cdot Y \subseteq Y$ and for all
$(i,\Bs,f_H,f_V) \in \fV$, $f_V(X) \subseteq Y$ (with $X$ the set of
$\Vc$-valid forest types).

Finally, given $\Vc$-valid forest types $h$ and $h'$, we say that $h$ is
$\Vc$-reachable from $h'$ iff there exists $v \in V$ that is
$\Vc$-valid and such that $h = vh'$. The following fact can be
verified from the definitions.

\begin{fct} \label{fct:validpro} Let $Q$ be a set of \shals such that
$\fV = \{\beta(q,x) \mid (q,x) \in Q\}$. Then $h \in H$ (resp. $v \in
V$) is $\Vc$-valid iff it is $Q$-valid. Moreover, for all $\Vc$-valid
$h,h' \in H$, $h$ is $\Vc$-reachable from $h'$ iff $h$ is
$Q$-reachable from $h'$.
\end{fct}

We say that $\Vc$ is \emph{branching} iff $\fV$ contains a profile of arity
$2$. One can verify that this implies the existence of a maximal
$\Vc$-reachability class denoted $H_{\Vc}$. Finally, we say that a branching
$\Vc$ is \emph{reduced} when all $\Vc$-valid forest types are mutually
reachable, i.e. $H_{\Vc}$ is the whole set of $\Vc$-valid forest types.

\medskip
\noindent {\bf Profile Saturation.} We are now ready to rephrase
saturation as a property of the sets $\bcI[\alpha,X]$. Set $X \subseteq
H$, we say that a configuration $\Vc \in \bcC$ is
\emph{$X$-compatible} iff it is branching, it is reduced, $H_{\Vc} = X$
and $\Vc = \bigcup_i \Vc_i$ with $\Vc_i \in \bcI[\alpha,X]$ for all
$i$.

Let $\Vc$ be an $X$-compatible configuration. We say that $v \in V$ is 
$\Vc$-saturated iff there exist $v_1 \cdots v_n = v$ such that:
\begin{itemize}
\item for all $j$, $v_j$ is $\Vc$-valid.
\item for all $\fV \in \Vc$ there exists $(i,\Bs,f_H,f_V) \in \fV$
  such that either $i = 0$ (i.e. $\fV$ abstracts a set of non-port
  nodes) or $v_j \in f_V(H_{\Vc})$ for some $j$.
\end{itemize}
Let $\ell$ be as defined in Fact~\ref{fct:pro3}. The following fact is
a simple consequence of the definitions. 

\begin{fct} \label{fct:prosat}
Set $X \subseteq H$ and $v$ be an idempotent of $V$. For every $k \geq
\ell$ the following properties are equivalent: 
\begin{enumerate}
\item  There exists a branching and reduced $Q \subseteq \As^+$ such
  that $X = H_{Q}$ and $v$ is the image of some $(Q,k)$-saturated
  context.
\item There exists a $X$-compatible $\Vc \in \bcC$ such that $v$ is
  $\Vc$-saturated.
\end{enumerate}
\end{fct}
\begin{proof}[Proof sketch]
From top to bottom. Let $Q$ and $v$ be as in (1). Let $\Delta$ be the
$(Q,k)$-saturated context such that $\alpha(\Delta)=v$. We construct
$\Vc$ and $v_1\cdots v_n$ witnessing (2) as follows. To each $(q,x)
\in Q$, we associate the set $\fU =\{\beta(q',x') \mid (q',x')
\pequiv{k}^X (q,x)\}$. We then set $\Vc$ as the set of all such sets
\fU. It follows from the definition and Fact~\ref{fct:validpro} that
$H_Q = X = H_\Vc$. Moreover, $\Vc$ is by definition a union of
elements of $\bcI_k[\alpha,X]$ (and hence of $\bcI_k[\alpha,X]$ by
definition of $k$) and is therefore $X$-compatible. It is then
immediate to check that the $(Q,k)$-saturation of $\Delta$ implies the
existence of $v_1\cdots v_n$ with the desired properties. Note that we
did not use the hypothesis that $v$ is idempotent, it is only required
for the other direction.
 
From bottom to top. Let $\Vc$ and $v_1\cdots v_n$ be as required for
(2). By definition, we have $\Vc =\bigcup_i \Vc_i$ where each $\Vc_i$
is $H_\Vc$-relevant and therefore
$(H_\Vc,k)$-relevant. This means that for all $i$, there 
exists a set of \shal $Q_i$ for $\Vc_i$ as in the definition of
$\bcI_k[\alpha,H_{\fQ}]$. We set $Q = \bigcup_i Q_i$. It is immediate
from the definition of $Q$ and Fact~\ref{fct:validpro} that $H_Q = X =
H_\Vc$ and that $v_1,\dots,v_n$ are $Q$-valid. We construct the desired
$(Q,k)$-saturated context $\Delta$ as follows. For any node $(p,x) \in
Q$, we construct a $Q$-valid context of type $v$ having a node
$(p',x')$ on its backbone satisfying $(p',x') \pequiv{k}^{H_Q}
(p,x)$. It will then suffice to define $\Delta$ as the concatenation
of all these contexts. Since $v$ is idempotent $\Delta$ will have type
$v$ as well.

Let $(p,x)$ with $p$ in $Q$ and $x$ a port-node of $p$, by definition,
there exist some $i$, some $\fV \in \Vc_i$ and some $(i,\Bs,f_H,f_V)
\in \fV$ such that $(i,\Bs,f_H,f_V) = \beta(q,y)$ with $(p,x) 
\pequiv{k}^{H_Q} (q,y)$. As $x$ is a port-node, so is $y$ and we
have  $f_V(H_Q) \neq \emptyset$ and therefore by $\Vc$-saturation
of $v$, we get $(i',\Bs',f'_H,f'_V) \in \fV$ such that $f'_V(H_Q)$
contains $v_j$ for some $j$. By definition, we get $(p',x')
\pequiv{k}^{H_Q} (q,y) \pequiv{k}^{H_Q} (p,x)$ such that
$\beta(p',x') = (i',\Bs',f'_H,f'_V)$.  Hence we can create a $Q$-valid
context of type $v_j$, with a unique position $(p',x')$ on its
backbone. Since $v_1,\dots,v_n$ are all $Q$-valid, this
context can then be completed into a $Q$-valid context of type $v$
which terminates the proof.
\end{proof}

We say that $\alpha$ is \emph{closed under profile saturation} iff for
all $X \subseteq H$, for all $X$-compatible $\Vc \in \bcC$, for all $v
\in V$ that are $\Vc$-saturated and all $h_1,h_2 \in H_{\Vc}$:
\[
v^\omega h_1 =v^\omega h_2
\]
Observe that all quantifications in the definition range over finite
sets. Therefore, if one can compute the $X$-compatible configurations
for all $X$, one can decide closure under profile saturation by
testing all possible combinations. In the next proposition, we prove
that this is equivalent to testing closure under saturation.

\begin{proposition} \label{prop:satpro}
Let $\alpha: \A^\mydelta \rightarrow (H,V)$ be a morphism into a finite 
forest algebra. Then the following three properties are equivalent:
\begin{enumerate}
\item $\alpha$ is closed under saturation.
\item $\alpha$ is closed under $\ell$-saturation
\item $\alpha$ is closed under profile saturation.
\end{enumerate}
\end{proposition}

\begin{proof}
We prove that $1) \Rightarrow 3) \Rightarrow 2) \Rightarrow 1)$. That
$2) \Rightarrow 1)$ is immediate by definition of
saturation.

We now prove $1) \Rightarrow 3)$. Assume that $\alpha$ is closed
under saturation. By Lemma~\ref{lem-sat-k}, $\alpha$ is closed by
$k$-saturation for some $k \geq \ell$. We need to prove that $\alpha$
is closed under profile saturation. Let $X \subseteq H$, $\Vc \in
\bcC$ that is $X$-compatible $\Vc \in \bcC$, $v \in V$ that is
$\Vc$-saturated and $h_1,h_2 \in H_{\Vc}$. Using
Fact~\ref{fct:prosat}, we get $Q \subseteq \As^+$ such that $H_{\Vc} =
H_{Q}$ and $v^\omega$ is the image of some $(Q,k)$-saturated
context. It is now immediate from $k$-saturation that $v^\omega h_1 =
v^\omega h_2$.

It remains to prove that $3) \Rightarrow 2)$. Assume that $\alpha$ is
closed under profile saturation.  We need to prove that $\alpha$
is closed under $\ell$-saturation. Let $Q \subseteq \As^+$, $\Delta$
that is $(Q,\ell)$-saturated and $h_1,h_2 \in H_Q$. Using
Fact~\ref{fct:prosat},we get $\Vc \in \bcC$ such that $H_{\Vc} =
H_{Q}$ and $\alpha(\Delta^\omega)$ is $\Vc$-saturated. It is now
immediate from profile saturation that $\alpha(\Delta)^\omega h_1 =
\alpha(\Delta)^\omega h_2$.
\end{proof}

In view of Proposition~\ref{prop:satpro}, it is enough to show that closure
under profile saturation is decidable in order to prove Theorem~\ref{cor-decid}.
Because all the quantifications inside the definition of profile saturation
range over finite sets, it is enough to show that those finite sets, namely
$\bcI[\alpha,X]$ for all $X\subseteq H$, can be computed.

This is immediate in the case of ranked trees. Indeed for trees of rank $l$, the
set of legal \shals is a subset of $\As^l$. Therefore
$\bcI[\alpha,X] = \bcI_{l+1}[\alpha,X]$ can now be computed by
considering all the finitely many possible sets $Q\subseteq
\As^l$. Hence Theorem~\ref{cor-decid} is proved for regular languages
of ranked trees.

In the general case it is not obvious how to compute $\bcI[\alpha,X]$
and this is the goal of the remaining part of this section.


\subsection{Computing the Sets of $X$-indistinguishable Configurations}

We present an algorithm which, given as input $\alpha: \A^\mydelta \rightarrow
(H,V)$ and $X \subseteq H$, computes the set $\bcI[\alpha,X]$. This is a
fixpoint algorithm that starts from trivial configurations corresponding to
sets of \shals that are singletons composed of a single letter \shal
and saturate the set with two operations. 

Our first operation is the semigroup operation on $\bcC$ (recall
Fact~\ref{fct:prosemi2}) which corresponds to concatenating \shals. Our second
and most important operation is derived from a well-known property of \FOdw on
strings. Let $C$ be a finite string alphabet and let $u,u' \in C^+$ such that
$u,u'$ both contain all labels in $C$. Then for all $k \in \nat$ and any $u''
\in C^+$:
\[
(u)^ku''(u')^{k} \mequivk (u)^k(u')^{k} 
\]
In our case however, the situation will be slightly more complicated
as we work with the weaker equivalence $\pequiv{k}^X$ in which tests
on labels are relaxed. 

\begin{remark}
By definition for any $\Vc \in \bcI[\alpha,X]$ all profiles contained
in sets of \Vc have the same alphabet. Therefore, we will assume
implicitly that this is true of all sets of profiles we consider from
now and whenever we refer to ``the alphabet of \Vc'' we mean this
common alphabet.
\end{remark}

\bigskip
\noindent
{\bf Fixpoint Algorithm.} Recall from Fact~\ref{fct:prosemi2} that
$\bcC$ is equipped with a semigroup operation. We start with a few
definitions about alphabets that we will need in order to present
the algorithm. To each alphabet $\Bs \subseteq \As$, we associate a
configuration $\uplift{\Bs}$ as follows,
\[
\uplift{\Bs} = \{\{\beta(p,x) \mid \text{$p$ has alphabet $\Bs$ and
  $x$ has label $c$}\} \mid c \in \Bs\} \in \bcC
\]
Observe that for any \Bs, it is simple to compute $\uplift{\Bs}$ from
$\alpha$. Indeed, for any $c \in \Bs$, if $y$ denotes the unique
position in the \shal $c$, one can verify that,
\[
\left\{\beta(p,x) \mid \begin{array}{l}\text{$p$ has alphabet $\Bs$}\\
                         \text{and $x$ has label $c$}\end{array}\right\} = \fP +_r \{\beta(c,y)\} +_\ell
  \fP \bigcap \{\fv \in \fP \mid \fv \text{ has alphabet \Bs}\}
\]
An important remark is that while $\uplift{\Bs}$ is a configuration,
in general, it is not an $X$-relevant configuration (for any $X$). The
main idea behind the fixpoint algorithm is that $\uplift{\Bs}$ can
become $X$-relevant if one adds ''appropriate'' $X$-relevant
configurations to its left and to its right. The definition of
''appropriate'' is based on the notion of \emph{$X$-approximation} of
an alphabet that we define now.

In the \relaxedX game, there are three types of nodes, port-nodes,
\Xnodes and \bXnodes. Let $\Bs \subseteq \As$, and let $c,c' \in \Bs$
we say that $c,c'$ are \emph{$\Bs[X]$-equivalent} iff $c=c'$ or there
exists $b(\hole) \in \Bs$ such that $c,c'$ are port-nodes or \Xnodes
labels of inner label $b$. Finally, an \emph{$X$-approximation} of \Bs
is an alphabet $\Cs \subseteq \Bs$ such for any $c \in \Bs$, there
exists $c' \in \Cs$ that is \emph{$\Bs[X]$-equivalent} to $c$. 

We can now present the algorithm. Set $\bcT[\alpha] \subseteq
\bcI[\alpha,X]$ for all $X$ as the set of configurations associated to
sets of \shals of the form $\{c\}$ where $c$ is a single letter in
$\As$. More precisely, $\bcT[\alpha]$ is the set of configurations:
\[
\{\{\{\beta(c,x)\}\} \mid c \in \As \text{ and $x$ the unique position
  in $c$}\}
\]
We set $Sat[X,\alpha]$ as the smallest set $S \subseteq \bcC$
containing $\bcT[\alpha]$ and such that:

\begin{enumerate}
\item \label{op:one} For all $\Vc,\Vc' \in S$, $\Vc + \Vc' \in
  S$.  
\item \label{op:two} For all $\Bs \subseteq \As$, if $\Vc,\Vc' \in S$
  have (possibly different) alphabets that are both $X$-approximations
  of $\Bs$, then $\omega \Vc + \uplift{\Bs} + \omega \Vc' 
  \in S$.
\end{enumerate}
where $\omega = \omega(\bcC)$. Clearly $Sat[X,\alpha]$ can be computed
from $\alpha$. It is connected to $\bcI[\alpha,X]$ via the proposition
below. For $\Uc_1,\Uc_2 \in \bcC$, we write $\Uc_1 \lesspr \Uc_2$ iff 
\begin{enumerate}
\item\label{item-pos1} For every $\fV_1 \in \Uc_1$ there exists
  $\fV_2 \in \Uc_2$ such that $\fV_1 \subseteq \fV_2$.
\item\label{item-pos2} For every $\fV_2 \in \Uc_2$ there exists
  $\fV_1 \in \Uc_1$ such that $\fV_1 \subseteq \fV_2$.
\end{enumerate}
One can verify that $\lesspr$ is a preorder. If $\bcV \subseteq
\bcC$ , the \emph{downset} of $\bcV$ is set $\dclos \bcV = \{\Vc \mid
\exists \Uc \in \bcV \text{ such that } \Vc \lesspr \Uc\}$.

\begin{fct} \label{fct:prosemi3} Within $\bcC$, $+$
  is compatible with $\lesspr$ (i.e. $\Vc_1 \lesspr \Vc_2$ and $\Uc_1 \lesspr
  \Uc_2$ implies $\Vc_1+\Uc_1 \lesspr \Vc_2 + \Uc_2$).
\end{fct}

\begin{proposition} \label{prop:algo}
Let $\ell = 2|\As|^2(|\bcC| + 1)$ and $X \subseteq
H$, then for any $k \geq \ell$:  
\[ 
\bcI[\alpha,X] = \bcI_k[\alpha,X] = \dclos Sat[X,\alpha]
\]
\end{proposition}

It follows from Proposition~\ref{prop:algo} that $\bcI[\alpha,X]$ can
be computed for any $X \subseteq H$. By combining this with
Proposition~\ref{prop:satpro}, we obtain the desired corollary:

\begin{corollary} \label{cor:sat}
Let $\alpha: \A^\mydelta \rightarrow (H,V)$ be a morphism into a finite  
forest algebra. It is decidable whether $\alpha$ is closed under
saturation.
\end{corollary}

Observe that Proposition~\ref{prop:algo} also contains a bound
for $\ell$ in Fact~\ref{fct:pro3}. This bound is of particular
interest: as explained in Proposition~\ref{prop:satpro}, $\ell$ 
is also a bound for saturation, if $\alpha$ is closed under 
saturation, then it is closed under $\ell$-saturation.  

It now remains to prove Proposition~\ref{prop:algo}. We prove that for
any $k \geq \ell$, $\bcI[\alpha,X] \subseteq \bcI_k[\alpha,X]
\subseteq \dclos Sat[X,\alpha] \subseteq \bcI[\alpha,X]$. Observe
that  $\bcI[\alpha,X] \subseteq \bcI_k[\alpha,X]$ is immediate by
Fact~\ref{fct:pro3}. We give the two remaining inclusions their own
subsections.

\subsection{Proof of Correctness}

We prove that $Sat[X,\alpha] \subseteq \bcI[\alpha,X]$. One can then
verify that $\dclos \bcI[\alpha,X] = \bcI[\alpha,X]$ and therefore
that $\dclos Sat[X,\alpha] \subseteq \bcI[\alpha,X]$. Recall that
$\bcI[\alpha,X]$ is defined as $\bigcap_{k   \in \nat}
\bcI_k[\alpha,X]$. This means that it suffices to prove that for all
$k \in \nat$, $Sat[X,\alpha] \subseteq \bcI_k[\alpha,X]$. We fix such
a $k \in \nat$ for remainder of the proof.

By definition, $\bcT[\alpha] \subseteq \bcI_k[\alpha,X]$ for every $k \in
\nat$. We prove that $\bcI_k[\alpha,X]$ is closed under the two
operations in the definition of $Sat$. We begin with 
Operation~\eqref{op:one}.

\medskip
\noindent {\bf Operation~\eqref{op:one}.} Let $\Vc,\Vc' \in \bcI_k[\alpha,X]$
and let $Q,Q'$ be the sets of \shals witnessing the membership of $\Vc$ and
$\Vc'$ in $\bcI_k[\alpha,X]$. Set $R = \{q + q' \mid q \in Q \text{ and } q'
\in Q'\}$, we prove that $R$ witnesses the membership of $\Vc + \Vc'$ in
$\bcI_k[\alpha,X]$. This is a consequence of Fact~\ref{fct:gentypes} and the
following lemma:

\begin{lemma} \label{lem:comp}
Let $(p_1,x_1),(p_2,x_2) \in Q$ such that $(p_1,x_1) \pequiv{k}^X
(p_2,x_2)$ and $(p'_1,x'_1),(p'_2,x'_2) \in Q'$ such that $(p'_1,x'_1)
\pequiv{k}^X (p'_2,x'_2)$. Then if $r_1 = p_1 + p'_1$ and $r_2 = p_2 +
p'_2$,
\[
(r_1,x_1) \pequiv{k}^X (r_2,x_2) \quad\text{and}\quad (r_1,x'_1) \pequiv{k}^X (r_2,x'_2)
\]
\end{lemma}
\begin{proof}
This is a composition lemma whose proof is immediate using \efgame games.
\end{proof}

We have three conditions to check. That~(1) holds is immediate from
Lemma~\ref{lem:comp}. For~(2), if $(r,x)\in R$, we have $r = p + p'$
with $p,p' \in Q,Q'$. By symmetry, assume that $x \in p$. By
definition of $Q$, we get $(p_1,x_1),\dots,(p_n,x_n) \in Q$ such that
$(p,x) \pequiv{k}^X (p_1,x_1) \pequiv{k}^X \cdots \pequiv{k}^X
(p_n,x_n)$ and $\fV = \set{\beta(p_i,x_i) \mid i\leq n} \in \Vc$. Set
$\fV'$ as the union of all sets in $\Vc'$ and set $R'$ as the set of
pairs $(q,y)$ such that $q = p_j + p''$ with $j \leq n$ and $p'' \in
Q'$. By Lemma~\ref{lem:comp} for all $(q,y) \in R'$, $(r,x)
\pequiv{k}^X (q,y)$. Moreover, one can verify using
Fact~\ref{fct:gentypes} that 
\[
\beta(R') = \fV +_\ell \fV' \in \Vc
\]

It remains to verify~(3). Set $\fU \in \Vc + \Vc'$. By symmetry,
assume that $\fU = \fV +_\ell \fV'$ with $\fV \in 
\Vc$ and $\fV'$ the union of all sets of $\Vc'$. By definition of
$Q,Q'$ there exists $(p_1,x_1) \pequiv{k}^X \cdots \pequiv{k}^X
(p_n,x_n) \in Q$ such that $\fV = \set{\beta(p_i,x_i) \mid i \leq
  n}$. Using the same set $R'$ of pairs $(q,y)$ as above we get
that all pairs in $R'$ are $\pequiv{k}^X$-equivalent and
\[
\beta(R') = \fV +_\ell \fV' = \fU
\]

\medskip
\noindent {\bf Operation~\eqref{op:two}.} Let $\Bs \subseteq \As$ and
$\Vc,\Vc' \in \bcI_k[\alpha,X]$ having alphabets $\Cs,\Cs'$ that are 
$X$-approximations of $\Bs$ and let $\fV$ and $\fV'$ be the unions of
all sets in \Vc and $\Vc'$ respectively. Let $Q$ and $Q'$ be the sets
of \shals witnessing the membership of $\Vc$ and $\Vc'$ into
$\bcI_k[\alpha,X]$. Furthermore, set $R$ as the set of all \shals of
alphabet $\Bs$. We prove that $P = k\omega Q + R + k\omega Q'$
witnesses the fact that $k\omega \Vc + \uplift{\Bs} + k\omega \Vc' =
\omega \Vc + \uplift{\Bs} + \omega \Vc'$ belongs to
$\bcI_k[\alpha,X]$. Most of the proof is based on the following
property of the equivalence $\pequiv{k}^X$.

\begin{lemma}\label{lem:fo2game}
Let $\hat{q_1},\hat{q_2} \in k\omega Q$, $\hat{q_1}'\hat{q_2}' \in
k\omega Q'$ and $r_1,r_2 \in R$. Then for every nodes $x_1,x_2$ of
$r_1,r_2$ with the same label $c \in \Bs$
\[
(\hat{q_1} + r_1 +\hat{q_1}',x_1) \pequiv{k}^X (\hat{q_2} + r_2 +\hat{q_2}',x_2)
\]  
\end{lemma}

\begin{proof}
We give a winning strategy for Duplicator in the \relaxedX game. We
simplify the argument by assuming that $\hat{q_1} = \hat{q_2} =
k\omega q$ for some $q \in Q$ and that $\hat{q_1}' = \hat{q_2}' =
k\omega q'$  for some $q' \in Q'$. Since all \shals in $Q$
(resp. $Q'$) are $\pequivk^X$-equivalent (see Item~(1) in the
definition of $\bcI_k[\alpha,X]$), one can then obtain a strategy for
the general case by adapting this special case.

Recall that $q$ (resp. $q'$) has alphabet $\Cs$ (resp. $\Cs'$) and
that $\Cs$ and $\Cs'$ are $X$-approximations of $\Bs$. Therefore,
using a standard game argument,  one can verify that Duplicator can
win $k$ moves of the \relaxedX game between $(k\omega q + r_1 + k\omega
q',x_1)$ and $(k\omega q + r_2 + k\omega q',x_2)$ as long as no safety
move is played. In case a safety move is played the $X$-approximation
hypothesis guarantees that $r_1,r_2$ of alphabet $\Bs$ contains the
appropriate letter with a port-node label. 
\end{proof}

We now prove that the set $P$ satisfies the definition of
$\bcI_k[\alpha,X]$ for $k\omega \Vc + \uplift{\Bs} + k\omega \Vc'$. We
have three conditions to verify. That~(1) holds is immediate from
Lemma~\ref{lem:fo2game}.

For~(2), consider $(p,x) \in P$. By definition, $p = \hat{q} + r +
\hat{q}'$ with $r \in R$, $\hat{q} \in k\omega Q$ and $\hat{q}' \in
k\omega Q'$. We treat the case when $x \in r$ (the other cases are
treated with a similar argument). Let $c$ be the label of $x$ and let
$(r_1,x_1)$ \ldots $(r_n,x_n)$ be all nodes such that $r_i \in R$ and
$x_i$ is a node of $r_i$ of label $c$. By definition of \uplift{\Bs}
and $R$, we have 
\[
\fU = \{\beta(r_1,x_1),\dots,\beta(r_n,x_n)\} \in \uplift{\Bs}
\]
Let $(p_1,y_1)$ \ldots $(p_m,y_m)$ be all nodes such that $p_i \in P$ 
and $y_i$ is a node of label $c$ in the ``$R$-part'' of $p_i$. From
Lemma~\ref{lem:fo2game} we have:
\[
(p,x) \pequiv{k}^X (p_1,y_1) \pequiv{k}^X \dots \pequiv{k}^X (p_m,y_m)
\]
Observe that viewed as nodes of the ``$R$-part'' of $p_i$, the nodes 
$y_i$ are exactly the nodes $x_j$. Using Fact~\ref{fct:gentypes} one
can then verify that
\[
\set{\beta(p_i,y_i) \mid i \leq m} = \underbrace{\fV +_r \cdots +_r
  \fV}_{k\omega\ times} +_r \fU +_\ell \underbrace{\fV' +_\ell \cdots +_\ell
  \fV'}_{k\omega\ times} \in k\omega \Vc +
\uplift{\Bs} + k\omega \Vc'
\]
It remains to prove~(3). Let $\fU \in \omega \Vc +
\uplift{\Bs} + \omega \Vc'$. Again, we concentrate on the case when
\[
\fU = \underbrace{\fV +_r \cdots +_r
  \fV}_{k\omega\ times} +_r \fW +_\ell \underbrace{\fV' +_\ell \cdots +_\ell
  \fV'}_{k\omega\ times}
\]
\noindent
with $\fW \in \uplift{\fR}$ (other cases are treated in a similar
way). By definition of $\uplift{\Bs}$, we have
\[
\fW = \{\beta(r_1,x_1),\dots,\beta(r_n,x_n)\} \in \uplift{\Bs}
\]
\noindent
with $(r_1,x_1)$ \ldots $(r_n,x_n)$ as all nodes such that $r_i \in R$
and $x_i$ is a node of $r_i$ of label $c$ for some fixed $c$. Let
$(p_1,y_1),\dots,(p_m,y_m) \in P$ be all the \shals of $P$ such that
$y_i$ is a node of $p_i$ with label $c$ in the ``$R$-part'' of
$p_i$. By Lemma~\ref{lem:fo2game}, we have 
\[
(p_1,y_1) \pequiv{k}^X \dots \pequiv{k}^X (p_m,y_m)
\]
\noindent
Observe that viewed as nodes of the ``$R$-part'' of $p_i$, the nodes 
$y_i$ correspond exactly to the nodes $x_j$. Using
Fact~\ref{fct:gentypes} one can then verify that
\[
\set{\beta(p_i,y_i) \mid i \leq m} = \fU
\]

\subsection{Proof of Completeness}

Let $\ell$ be defined as in Proposition~\ref{prop:algo}. We
prove that for any $k \geq \ell$, $\bcI_k[\alpha,X] \subseteq
Sat[X,\alpha]$. We will need the following definition.

Let $k \in \nat$, $X \subseteq H$. To every \shal $q \in \As^+$, we
associate a configuration $\Gc_k[X](q) \in \bcI[\alpha,X]$. For any
$p,x$ set $\fV_{p,x} = \set{\beta(p',x') \mid (p,x)
\pequiv{k}^X (p',x')}$. We set
\[
\Gc_k[X](q) = \{\fV_{q,y} \mid y \in q\}
\]
The following two facts are immediate consequences of the definitions:
\begin{fct} \label{fct:gen1}
For all $k \leq k' \in \nat$, $X \subseteq H$ and $q \in \As^+$ we have
$\Gc_{k'}[X](q) \lesspr \Gc_k[X](q)$.
\end{fct}
\begin{fct} \label{fct:gen2}
For all $k \in \nat$ and $X \subseteq H$ we have $\bcI_k[\alpha,X] = \dclos
\{\Gc_k[X](q) \mid q \in \As^+\}$. 
\end{fct}

We can now finish the proof of Proposition~\ref{prop:algo}.
The proof is by  induction on the size of the alphabet as stated in
the proposition below.

\begin{proposition} \label{prop:comp}
Let $\Bs \subseteq \As$, $k \geq 2|\Bs|^2(|\bcC|+1)$ and
$p$ a \shal such that $p$ contains only labels in $\Bs$. Then
$\Gc_k[X](p) \in \dclos Sat[X,\alpha]$.
\end{proposition}

Using Proposition~\ref{prop:comp} with $\Bs = \As$, we obtain that for  
any $k \geq \ell$ and any $p \in \As^+$, we have $\Gc_k[X](p) \in
\dclos Sat[X,\alpha]$. It then follows
from Fact~\ref{fct:gen2} that $\bcI_k[\alpha,X] \subseteq \dclos
Sat[X,\alpha]$ which terminates the proof of
Proposition~\ref{prop:algo}. It now remains to prove
Proposition~\ref{prop:comp}. The remainder of the section is devoted
to this proof.

\smallskip

For the sake of simplifying the presentation, we assume that $p$ can 
be an empty \shal denoted '$\varepsilon$' and that
$Sat[X,\alpha]$ contains an artificial neutral element '$0$'
such that $\Gc_k[X](\varepsilon) = 0$ for any $k$. As $\varepsilon$ will be
the only \shal having that property this does not harm the generality
of the proof.

As explained above, the proof is by induction on the size of $\Bs$.
The base case happens when $\Bs = \emptyset$. In that case, $p =
\varepsilon$ and $\Gc_k[X](\varepsilon) \in Sat[X,\alpha]$ by  
definition. Assume now that $\Bs \neq \emptyset$, we set $k \geq
2|\Bs|^2(|\bcC|+1)$ and $p$ as a \shal containing only
labels in $\Bs$. We need to prove that $\Gc_k[X](p) \in
\dclos Sat[X,\alpha]$.

\medskip
First observe that when $p$ does not contain \emph{all} labels in
$\Bs$, the result is immediate by induction. Therefore, assume that
$p$ contains all labels in $\Bs$. We proceed as follows. First, we
define a new notion called a \patt{\Bs[X]}{n}. Intuitively, a \shal $q$
contains a \patt{\Bs[X]}{n} iff all labels in $\Bs$ (modulo
$\Bs[X]$-equivalence) are repeated at least $n$ times in $q$. Then,
we prove that if $p$ contains a \patt{\Bs[X]}{n} for a large enough
$n$, then $\Gc_k[X](p)$ can be decomposed in such a way that it can
be proved to be in $Sat[X,\alpha]$ by using induction on the
factors, and Operations~\eqref{op:one} and~\eqref{op:two} to compose
them. Otherwise, we prove that $\Gc_k[X](p)$ can be decomposed as a
sum of bounded length whose elements can be proved to be in
$Sat[X,\alpha]$ by induction. We then conclude using
Operation~\eqref{op:one}. We begin with the definition of 
\patts{\Bs[X]}{n}.

\medskip
\noindent {\bf \patts{\Bs[X]}{n}.} Consider the $\Bs[X]$-equivalence
of labels in $\Bs$ and let $m$ be the number of equivalence
classes. We fix an arbitrary order on these classes that we denote by
$C_0,\dots,C_{m-1} \subseteq \Bs$. Recall that \Cs is a an
$X$-approximation of \Bs iff \Cs contains at least one element of each
class. Let $n \in \nat$. We say that a \shal $q$ contains a
\patt{\Bs[X]}{n} iff $q$ can be decomposed as
\[
q = q_0 + c_0 + q_1 + c_1 + \dots + q_n + c_n + q_{n+1}
\]
\noindent
such that for all $i \leq n$, $c_i \in C_{j}$ (with $j= i \bmod m$) and
$q_i$ is a (possibly empty) \shal. In particular, the decomposition
above is called the \emph{leftmost decomposition} iff for all $i \leq
n$ no label in $C_{j}$ (with $j = i \bmod m$) occurs in
$q_i$. Symmetrically, in the \emph{rightmost decomposition}, for all $i
\geq 0$, no label in $C_{i}$ (with $j = i \bmod m$) occurs in
$q_{i+1}$. Observe that by definition the leftmost and rightmost
decompositions are unique. In the proof, we use the following decomposition
lemma.

\begin{lemma}[Decomposition Lemma] \label{lem:decomp}
Let $n \in \nat$. Let $q$ be a \shal that contains a \patt{\Bs[X]}{n}
and let $q = q_0 + c_0 + \dots + c_n + q_{n+1}$ be the associated
leftmost or rightmost decomposition. Then
\[
\Gc_k[X](q) \lesspr  \Gc_{k-n}[X](q_0) + \Gc_{k-n}[X](c_0) + \dots +
\Gc_{k-n}[X](c_n) + \Gc_{k-n}[X](q_{n+1})
\]
\end{lemma}

\begin{proof}
  This is a simple \efgame game argument. Because of the missing boundary labels
  within the $q_j$, using at most $n$ moves, Spoiler can make sure that
  the game stays within the appropriate segment $q_j$ and can use the remaining
  $k-n$ moves for describing that segment.
\end{proof}

This finishes the definition of patterns. Set $n = m(|\bcC| + 1)$. We now
consider two cases depending on whether our \shal $p$ contains a
\patt{\Bs[X]}{2n}.

\medskip
\noindent
{\bf Case 1: $p$ does not contain a \patt{\Bs[X]}{2n}.} In that case
we conclude using induction and Operation~\eqref{op:one}. Let $n'$ be
the largest number such that $p$ contains a \patt{\Bs[X]}{n'}. By
hypothesis $n' < 2n$. Let $p = p_0 + c_0 + \dots + c_{n'}+ p_{n'+1}$
be the associated leftmost decomposition. Observe that by definition,
for $i \leq n'$, $p_i$ uses a strictly smaller alphabet than
$\Bs$. Moreover, since $p$ does not contain a \patt{\Bs[X]}{n'+1} this
is also the case for $p_{n'+1}$. Set $\tilde{k} = k-n'$, by choice of
$k$, we have $\tilde{k} \geq 2(|\Bs|-1)^2(|\bcC|+1)$. Therefore, we
can use our induction hypothesis and for all $i$ we get,
\[
\Gc_{\tilde{k}}[X](p_i) \in \dclos Sat[X,\alpha]
\]
Moreover, for all $i$, $\Gc_{\tilde{k}}[X](c_i) \in \bcT[\alpha]
\subseteq Sat[X,\alpha]$. Finally, using
Lemma~\ref{lem:decomp} we obtain
\[
\Gc_{k}[X](p) \lesspr  \Gc_{\tilde{k}}[X](p_0) + \Gc_{\tilde{k}}[X](c_0) + \dots +
\Gc_{\tilde{k}}[X](c_{n'}) + \Gc_{\tilde{k}}[X](p_{n'+1})
\]
From Operation~\eqref{op:one} the right-hand sum
is in $\dclos Sat[X,\alpha]$. We then conclude that
$\Gc_{k}[X](p) \in \dclos Sat[X,\alpha]$ which terminates
this case.

\medskip
\noindent
{\bf Case 2: $p$ contains a \patt{\Bs[X]}{2n}.} In that case we
conclude using induction, Operation~\eqref{op:one} and
Operation~\eqref{op:two}. By hypothesis, we know that $p$ contains a
\patt{\Bs[X]}{n}, let $p = p_0 + c_0 + \dots + c_{n}+ p_{n+1}$ be the
associated leftmost decomposition. Since $p$ contains a
\patt{\Bs[X]}{2n}, $p_{n+1}$ must contain a \patt{\Bs[X]}{n}. We set
$p_{n+1} = p' + c'_0 + \dots + c'_{n}+ p'_{n+1}$ as the associated
rightmost decomposition. In the end we get
\[
p =  p_0 + c_0 + \dots + c_{n} + p' + c'_0 + p'_1 + \dots + c'_{n}+ p'_{n+1}
\]
Set $\tilde{k} = k - 2n$ and observe that by choice of $k$, $\tilde{k}
\geq 2(|\Bs|-1)^2(|\bcC|+1)$. Therefore, as in the previous
case, we get by induction that for all $i$, $\Gc_{\tilde{k}}[X](p_i)
\in \dclos Sat[X,\alpha]$, $\Gc_{\tilde{k}}[X](p'_i) \in
\dclos Sat[X,\alpha]$, $\Gc_{\tilde{k}}[X](c_i) \in \dclos
Sat[X,\alpha]$ and $\Gc_{\tilde{k}}[X](c'_i) \in \dclos
Sat[X,\alpha]$. Using the same inductive argument for $p'$
may not be possible as $p'$ might contain all labels in \Bs.

If $p'$ does not contain all labels in \Bs, then, by induction,
$\Gc_{\tilde{k}}[X](p') \in \dclos Sat[X,\alpha]$ and we can
then use Lemma~\ref{lem:decomp} as in Case~1 to conclude that
$\Gc_{k}[X](p) \in \dclos Sat[X,\alpha]$. Assume now that
$p'$ contains all labels in $\Bs$. Recall that $m$ is the number of
$\Bs[X]$-equivalence classes. For all $j \leq |\bcC|$, set
\[
\Vc_j = \sum^{m-1+jm}_{i = jm} (\Gc_{\tilde{k}}[X](p_{i}) +
\Gc_{\tilde{k}}[X](c_{i})) \qquad \Vc'_j = \sum^{m-1+jm}_{i = jm} 
(\Gc_{\tilde{k}}[X](c'_{i}) + \Gc_{\tilde{k}}[X](p'_{i+1}))
\]
Observe that for all $j$, by definition $\Vc_j,\Vc'_j$ have an
alphabet which is an $X$-approximation of $\Bs$ and by
Operation~\eqref{op:one}, $\Vc_j,\Vc'_j \in \dclos
Sat[X,\alpha]$. Moreover, it follows from a pigeon-hole
principle argument that the sequences $\Vc_0 + \cdots + \Vc_{|\bcC|}$
and $\Vc'_0 + \cdots + \Vc'_{|\bcC|}$ must contain ``loops'',
i.e. there exists $j_1 < j_2$ and $j'_1 < j'_2$
such that
\[
\begin{array}{rcl}
\Vc_0 + \cdots + \Vc_{j_1} & = & \Vc_0 + \cdots + \Vc_{j_2} \\
\Vc'_{j_2} + \cdots + \Vc'_{|\bcC|} & = & \Vc'_{j_1} + \cdots +
\Vc'_{|\bcC|}
\end{array}
\]
Set $\Uc_1 = \Vc_0 + \cdots + \Vc_{j_1}$, $\Uc_2 = \Vc_{j_1+1} +
\cdots + \Vc_{j_2}$, $\Uc'_1 = \Vc'_{j'_1} + \cdots \Vc'_{j'_2-1}$ and
$\Uc'_2 = \Vc'_{j'_2} + \cdots + \Vc'_{|\bcC|}$. Observe that by
Operation~\eqref{op:one}, we have $\Uc_1,\Uc_2,\Uc'_1,\Uc'_2 \in
\dclos Sat[X,\alpha]$ and that by construction the alphabets of
$\Uc_2,\Uc'_1$ are $X$-approximations of $\Bs$. Moreover, a little
algebra yields $\Uc_1 = \Uc_1 + \Uc_2 = \Uc_1 + \omega \Uc_2$ and
$\Uc'_2 = \Uc'_1 + \Uc'_2 = \omega \Uc'_1 + \Uc'_2$.

Set $p''= p_{j_2m} + \cdots + c_n + p' + c'_0 + \cdots +
p'_{j'_1m-1}$. Observe that by hypothesis on $p'$, $p''$ contains
all labels in $\Bs$. It follows from Fact~\ref{fct:gen1} and
Lemma~\ref{lem:decomp} that
\[
\Gc_{k}[X](p) \lesspr \Uc_1 + \Uc_2 + \Gc_{\tilde{k}}[X](p'')
+ \Uc'_1 + \Uc'_2 = \Uc_1 + \omega \Uc_2 + \Gc_{\tilde{k}}[X](p'')
+ \omega \Uc'_1 + \Uc'_2
\]
Moreover, since $p''$ has alphabet $\Bs$, it is immediate that
$\Gc_{\tilde{k}}[X](p'') \lesspr \uplift{\Bs}$. Therefore, 
\[
\Gc_{k}[X](p) \lesspr \Uc_1 + \omega \Uc_2 + \uplift{\Bs} + \omega \Uc'_1 + \Uc'_2
\]
It is now immediate from Operation~\eqref{op:two} $\omega \Uc_2 + \uplift{\Bs}
+ \omega \Uc'_1 \in \dclos Sat[X,\alpha]$. By combining this
with Operation~\eqref{op:one}, we obtain
\[
\Gc_{k}[X](p) \lesspr \Uc_1 + \omega \Uc_2 +  \uplift{\Bs} + \omega
\Uc'_1 + \Uc'_2 \in \dclos Sat[X,\alpha]
\]
We conclude that $\Gc_{k}[X](p) \in \dclos Sat[X,\alpha]$
which terminates the proof.

\section{Other logics}\label{others}

It turns out that the proof of Theorem~\ref{th-efmax} depends on the horizontal
modalities of the logic only via the notion of definability within \shals. It
can therefore be adapted to many other horizontal modalities assuming those can
at least express the fact that two nodes are siblings (ie. can talk about the
\shal of a given node). By tuning this notion
one can obtain several new characterizations.
We illustrate this feature in this section with the horizontal
predicates $\next$, $\previous$, $\mathbf{S}$ and $\mathbf{S}^{\neq}$,
adopting the  point of view of temporal logic.

The semantic of these predicates is defined as follows. The formula
$\mathbf{S}^{\neq}\varphi$ holds at a node $x$ if $\varphi$ holds at
some sibling of $x$ distinct from $x$. It is a shorthand for
$\forward\varphi \vee \past\varphi$. The formula $\mathbf{S}\varphi$
holds at $x$ if $\varphi$ holds at some sibling of $x$ 
including $x$. It is a shorthand for $\varphi \vee
\mathbf{S}^{\neq}\varphi$. The predicates $\next$ and $\previous$ are
the usual next sibling and previous sibling modalities.

\smallskip 

The vertical navigational modalities remain the same and the 
corresponding logics are denoted by \EFH, \EFHs, \EFmax and a
characterization can be obtained for each of them using the
same scheme as for \EFF.

\medskip

{\sloppy As before, \EFHs and \EFmax are equivalent to two-variable fragments of
first-order logic. \EFHs has the same expressive power as \FOds while
\EFmax corresponds to \FOdxs. Here $s(x,y)$ is a 
binary predicate that holds when $x$ and $y$ are siblings and \succh
is a binary predicate that holds when $y$ is the next sibling of
$x$. These facts can be proved along the same way as the equivalence
between \FOd and \EFF, see Theorem~\ref{thm-fod-eff}.}

Note that this is no longer the case for \EFH as languages defined in 
this formalism are closed under bisimulation while in the two variable
fragment of first-order logic it is possible to have quantifications
over incomparable nodes by using the equality and negation which rules
out closure under bisimulation. 

\medskip

The proof techniques presented in the previous sections require at
least the power of testing whether two nodes are sibling in order to
extract a \shal within a forest. Hence it cannot be applied to \FOdv
and finding a decidable characterization for this logic remains an
open question. Similarly, we rely on the fact that the child relation
cannot be expressed and finding a decidable characterization in the
presence of this predicate remains also an open question. 

\medskip

As we don't have a vertical successor modality, the characterizations we obtain
for \EFH, \EFHs and \EFmax still require Identity~\eqref{eqv} on the vertical
monoid $V$ of the forest algebra. Identity~\eqref{eqh} is now replaced by the
appropriate identity corresponding to the new horizontal expressive
power. Finally the notion of saturation is adapted by replacing $\pequiv{k}^X$
with a notion reflecting the horizontal expressive power of the logic. It is
defined as in Section~\ref{sec-reach} by only modifying the allowed moves in
the game in order to reflect the horizontal expressive power of the associated
logic (the constraints on the labels remaining untouched within each game).  In
a similar way, $\mequiv{k}$ is replaced in the proof by the suitable
game. Besides these changes at the level of definitions, the characterization is stated
and proved as for Theorem~\ref{th-efmax}.

\subsection{\EFH}
\newcommand\spequiv{\ensuremath{\mathbf{S}\text{-}\!\!\pequiv{}^X}\xspace}
\newcommand\smequiv{\ensuremath{\mathbf{S}\text{-}\!\!\mequiv{}}\xspace}

In this case the games on \shals are defined with no navigational
constraints on Duplicators moves: Duplicator can respond by choosing
an arbitrary node, the restriction being only on its label.

Note that the games no longer depend on $k$, as only the presence or absence of
a given symbol of $\As$ inside the \shal matters. We write
\smequiv and \spequiv the
equivalence relations resulting from this game and its $X$-relaxed variant.

The following  analog of Claim~\ref{lemma-types-fod} is an immediate
consequences of the definitions.

\begin{claim}\label{lemma-S-types}
Let $X \subseteq H$ and $(p,x)$ be a node. There is a
\EFH formula $\psi_{p,x}$ having one free variable and such that for
any forest $s$, $\psi_{p,x}$ holds exactly at all nodes $(p',x')$ such
that $(p,x) ~\spequiv (p',x')$.
\end{claim}

Recall the definition of saturation given in Section~\ref{carac}.  The notion
of $\mathbf{S}$-saturation is obtained identically after replacing
$\pequiv{k}^X$ with $\spequiv$. With these new
definitions we get:

\begin{theorem}\label{th-efs}
A regular forest language $L$ is definable in \EFH iff its syntactic
morphism $\alpha: \A^{\Delta} \rightarrow (H,V)$ satisfies: 

\begin{enumerate}[label=\alph*)]
\item $H$ satisfies the identities
{\small\begin{equation}\label{eqhs} 2h=h \text{ and } f+g=g+f
  \end{equation}
}
\item $V$ satisfies Identity~\eqref{eqv}
{\small\begin{equation*}
(uv)^{\omega}v(uv)^{\omega}=(uv)^\omega
\end{equation*}
}
\item the leaf completion of $\alpha$ is closed under
  $\mathbf{S}$-saturation.
\end{enumerate}
\end{theorem}

\noindent Note that \eqref{eqhs} simply states that the logic is closed under 
bisimulation, hence reflecting exactly the horizontal expressive power
of \EFH.

Concerning the proof of Theorem~\ref{th-efs}, aside from the initial
choice of the integer $k$ which is no longer necessary here, it is
identical to the one we gave for Theorem~\ref{th-efmax} after
replacing Lemma~\ref{lem-choicek} by the following result: 

\begin{lemma} \label{lem-choicek-efh}
Let $L$ be a language whose syntactic forest algebra satisfies the
identities stated in Theorem~\ref{th-efs}. For all \shals $p ~\smequiv p'$
and for all forests $s$, $p[\bar s]$ and
$p'[\bar s]$ have the same forest type.
\end{lemma}

\begin{proof} Since $p ~\smequiv p'$ the forests
$p[\bar s]$ and $p'[\bar s]$ contain the same symbols but possibly with a
different number of occurrences. It follows from~\eqref{eqhs} that
$\alpha(p[\bar s]) = \alpha(p'[\bar s])$.
\end{proof}

\subsection{\EFHs and \FOds}
\newcommand\sspequiv{\ensuremath{\mathbf{S}^{\neq}\text{-}\!\!\pequiv{}^X}\xspace}
\newcommand\ssmequiv{\ensuremath{\mathbf{S}^{\neq}\text{-}\!\!\mequiv{}}\xspace}

As before the key point is the allowed moves in the
Ehrenfeucht-Fraïssé games. In this case we only require that 
Duplicator moves in a different position as soon as Spoiler does. 

As in the previous section, the games no longer depend on $k$ as it only
matters whether or not a label occurs and whether or not it occurs twice. We
write \ssmequiv and \sspequiv the equivalence relations resulting from this
game and its $X$-relaxed variant.

The following analog of Claim~\ref{lemma-types-fod} is an immediate
consequences of the definitions.

\begin{claim}\label{lemma-Sn-types}
Let $X \subseteq H$ and $(p,x)$ be a node. There is a
\EFHs formula $\psi_{p,x}$ having one free variable and such that for
any forest $s$, $\psi_{p,x}$ holds exactly at all nodes $(p',x')$ such
that $(p,x) ~\sspequiv (p',x')$.
\end{claim}

As in the previous case, replacing $\pequiv{k}^X$ with $\sspequiv$ in the
definition of saturation yields a new notion of saturation that we call
$\mathbf{S}^{\neq}$-saturation. We can show:

\begin{theorem}\label{th-efss}
A regular forest language $L$ is definable in \EFHs iff its syntactic
morphism $\alpha: \A^{\Delta} \rightarrow (H,V)$ satisfies: 

\begin{enumerate}[label=\alph*)]
\item $H$ satisfies the identities
{\small\begin{equation}\label{eqhss} 3h=2h \text{ and } f+g=g+f
  \end{equation}
}
\item $V$ satisfies Identity~\eqref{eqv}
{\small\begin{equation*}
(uv)^{\omega}v(uv)^{\omega}=(uv)^\omega
\end{equation*}
}
\item the leaf completion of $\alpha$ is closed under
  $\mathbf{S}^{\neq}$-saturation.
\end{enumerate}
\end{theorem}

\noindent Notice that~\eqref{eqhss} reflects exactly the horizontal expressive
power of \EFHs: no horizontal order and counting up to threshold $2$.

Concerning the proof of Theorem~\ref{th-efss}, aside from the initial
choice of the integer $k$ which is no longer necessary here, it is
identical to the one we gave for Theorem~\ref{th-efmax} after
replacing Lemma~\ref{lem-choicek} by the following result:

\begin{lemma} \label{lem-choicek-efhs}
Let $L$ be a language whose syntactic forest algebra satisfies the
identities stated in Theorem~\ref{th-efss}. For all \shals $p ~\ssmequiv p'$ and all forests $s$, $p[\bar s]$ and 
$p'[\bar s]$ have the same forest type.
\end{lemma}
\begin{proof}
Since $p ~\ssmequiv p'$ the forests $p[\bar s]$
and $p'[\bar s]$ contain the same symbols with the same number of
occurrences up to threshold $2$. It follows from~\eqref{eqhss} that
$\alpha(p[\bar s]) = \alpha(p'[\bar s])$.
\end{proof} 

\subsection{\EFmax and \FOdxs}
\newcommand\sucpequiv{\ensuremath{\mathbf{Suc}\text{-}\!\!\pequiv{k}^X}\xspace}
\newcommand\sucmequiv{\ensuremath{\mathbf{Suc}\text{-}\!\!\mequiv{k}}\xspace}

In this case Duplicator not only must respect the direction in which
Spoiler has moved his pebble, but she also must place her pebble on
the successor (predecessor) of the current position if this was also
the situation for Spoiler.

The games now depends on $k$ and we write \sucmequiv and \sucpequiv the
equivalence relations resulting from this game and its $X$-relaxed variant.

The following analog of Claim~\ref{lemma-types-fod} is an immediate
consequences of the definitions.

\begin{claim}\label{lemma-Succ-types}
Let $X \subseteq H$ and $(p,x)$ be a node. There is a
\EFmax formula $\psi_{p,x}$ having one free variable and such that for
any forest $s$, $\psi_{p,x}$ holds exactly at all nodes $(p',x')$ such
that $(p,x) ~\sucpequiv (p',x')$.
\end{claim}

As in the previous cases, we obtain from \sucpequiv a new notion of saturation
that we call $\next$-saturation. We can show:

\begin{theorem}\label{th-fodx}
A regular forest language $L$ is definable in \FOdxs iff its syntactic
morphism $\alpha: \A^{\Delta} \rightarrow (H,V)$ satisfies: 

\begin{enumerate}[label=\alph*)]
\item $H$ satisfies for all $h,g \in H$, for all $e \in H$ such that
  $2e=e$:
\begin{equation}\label{eqhx}
\omega(e+h+e+g+e)+ g+ \omega(e+h+e+g+e)=\omega(e+h+e+g+e)
\end{equation}

\item $V$ satisfies Identity~\eqref{eqv}
\begin{equation*}
(uv)^{\omega}v(uv)^{\omega}=(uv)^\omega
\end{equation*}

\item the leaf completion of $\alpha$ is closed under
  $\next$-saturation.
\end{enumerate}
\end{theorem}

\noindent Equation~(\ref{eqhx}) is extracted from the following result which is
essentially proved in~\cite{fodeux} based on a result
of~\cite{almeida} (see Footnote on page~\pageref{foot}).

\begin{theorem}[\cite{fodeux},\cite{almeida}]\label{thm-word-fodx}
A regular string language $L$ is definable in \FOdxw iff its syntactic
semigroup $S$ satisfies for all $u,v \in S$, for all $e \in S$ such
that $e^2=e$: 
\begin{equation*}
(eueve)^\omega v (eueve)^\omega =(eueve)^\omega
\end{equation*} 
\end{theorem}

Again, the proof of Theorem~\ref{th-fodx} follows the lines of the
proof of Theorem~\ref{th-efmax} after replacing
Lemma~\ref{lem-choicek} by the following simple result:

\begin{lemma} \label{lem-choicek-x} Let $L$ be a language whose syntactic
  forest algebra satisfies the identities stated in
  Theorem~\ref{th-fodx}. There exists a number $k'$ such that for all $k \geq
  k'$, all \shals $p ~\sucmequiv p'$ and all forests $s$, $p[\bar s]$ and
  $p'[\bar s]$ have the same forest type.
\end{lemma}

\begin{proof}
This is a consequence of the fact that $H$ satisfies
Identity~\eqref{eqhx}. The proof is identical to the one we provided
for Lemma~\ref{lem-choicek} replacing Theorem~\ref{thm-word-fod} by
Theorem~\ref{thm-word-fodx} and \mequiv{k} with \sucmequiv.
\end{proof} 

\subsection{Decidability}

Deciding whether a regular forest language is definable in \EFH and \EFHs is
simple from Theorem~\ref{th-efs} and Theorem~\ref{th-efss}. As in
Section~\ref{more} we prove that the corresponding notions of saturation are
equivalent to their abstract variant. The latter are decidable because they
don't depend on $k$ and, up to equivalence, only finitely many $Q\subseteq
\As^+$ needs to be considered.

However, for \FOdxw, it is not clear how to generalize the
construction of the indistinguishable sets. We leave this and the
status of deciding definability in \mbox{\FOdxs} as an open problem.

Let $\ell$ be defined as in Proposition~\ref{prop:algo}. We
prove that for any $k \geq \ell$, $\bcI_k[\alpha,X] \subseteq
Sat[X,\alpha]$. We will need the following definition.

Let $k \in \nat$, $X \subseteq H$. To every \shal $q \in \As^+$, we
associate a configuration $\Gc_k[X](q) \in \bcI[\alpha,X]$. For any
$p,x$ set $\fV_{p,x} = \set{\beta(p',x') \mid (p,x)
\pequiv{k}^X (p',x')}$. We set
\[
\Gc_k[X](q) = \{\fV_{q,y} \mid y \in q\}
\]
The following two facts are immediate consequences of the definitions:
\begin{fct} \label{fct:gen1}
For all $k \leq k' \in \nat$, $X \subseteq H$ and $q \in \As^+$ we have
$\Gc_{k'}[X](q) \lesspr \Gc_k[X](q)$.
\end{fct}
\begin{fct} \label{fct:gen2}
For all $k \in \nat$ and $X \subseteq H$ we have $\bcI_k[\alpha,X] = \dclos
\{\Gc_k[X](q) \mid q \in \As^+\}$. 
\end{fct}

We can now finish the proof of Proposition~\ref{prop:algo}.
The proof is by  induction on the size of the alphabet as stated in
the proposition below.

\begin{proposition} \label{prop:comp}
Let $\Bs \subseteq \As$, $k \geq 2|\Bs|^2(|\bcC|+1)$ and
$p$ a \shal such that $p$ contains only labels in $\Bs$. Then
$\Gc_k[X](p) \in \dclos Sat[X,\alpha]$.
\end{proposition}

Using Proposition~\ref{prop:comp} with $\Bs = \As$, we obtain that for  
any $k \geq \ell$ and any $p \in \As^+$, we have $\Gc_k[X](p) \in
\dclos Sat[X,\alpha]$. It then follows
from Fact~\ref{fct:gen2} that $\bcI_k[\alpha,X] \subseteq \dclos
Sat[X,\alpha]$ which terminates the proof of
Proposition~\ref{prop:algo}. It now remains to prove
Proposition~\ref{prop:comp}. The remainder of the section is devoted
to this proof.

\smallskip

For the sake of simplifying the presentation, we assume that $p$ can 
be an empty \shal denoted '$\varepsilon$' and that
$Sat[X,\alpha]$ contains an artificial neutral element '$0$'
such that $\Gc_k[X](\varepsilon) = 0$ for any $k$. As $\varepsilon$ will be
the only \shal having that property this does not harm the generality
of the proof.

As explained above, the proof is by induction on the size of $\Bs$.
The base case happens when $\Bs = \emptyset$. In that case, $p =
\varepsilon$ and $\Gc_k[X](\varepsilon) \in Sat[X,\alpha]$ by  
definition. Assume now that $\Bs \neq \emptyset$, we set $k \geq
2|\Bs|^2(|\bcC|+1)$ and $p$ as a \shal containing only
labels in $\Bs$. We need to prove that $\Gc_k[X](p) \in
\dclos Sat[X,\alpha]$.

\medskip
First observe that when $p$ does not contain \emph{all} labels in
$\Bs$, the result is immediate by induction. Therefore, assume that
$p$ contains all labels in $\Bs$. We proceed as follows. First, we
define a new notion called a \patt{\Bs[X]}{n}. Intuitively, a \shal $q$
contains a \patt{\Bs[X]}{n} iff all labels in $\Bs$ (modulo
$\Bs[X]$-equivalence) are repeated at least $n$ times in $q$. Then,
we prove that if $p$ contains a \patt{\Bs[X]}{n} for a large enough
$n$, then $\Gc_k[X](p)$ can be decomposed in such a way that it can
be proved to be in $Sat[X,\alpha]$ by using induction on the
factors, and Operations~\eqref{op:one} and~\eqref{op:two} to compose
them. Otherwise, we prove that $\Gc_k[X](p)$ can be decomposed as a
sum of bounded length whose elements can be proved to be in
$Sat[X,\alpha]$ by induction. We then conclude using
Operation~\eqref{op:one}. We begin with the definition of 
\patts{\Bs[X]}{n}.

\medskip
\noindent {\bf \patts{\Bs[X]}{n}.} Consider the $\Bs[X]$-equivalence
of labels in $\Bs$ and let $m$ be the number of equivalence
classes. We fix an arbitrary order on these classes that we denote by
$C_0,\dots,C_{m-1} \subseteq \Bs$. Recall that \Cs is a an
$X$-approximation of \Bs iff \Cs contains at least one element of each
class. Let $n \in \nat$. We say that a \shal $q$ contains a
\patt{\Bs[X]}{n} iff $q$ can be decomposed as
\[
q = q_0 + c_0 + q_1 + c_1 + \dots + q_n + c_n + q_{n+1}
\]
\noindent
such that for all $i \leq n$, $c_i \in C_{j}$ (with $j= i \bmod m$) and
$q_i$ is a (possibly empty) \shal. In particular, the decomposition
above is called the \emph{leftmost decomposition} iff for all $i \leq
n$ no label in $C_{j}$ (with $j = i \bmod m$) occurs in
$q_i$. Symmetrically, in the \emph{rightmost decomposition}, for all $i
\geq 0$, no label in $C_{i}$ (with $j = i \bmod m$) occurs in
$q_{i+1}$. Observe that by definition the leftmost and rightmost
decompositions are unique. In the proof, we use the following decomposition
lemma.

\begin{lemma}[Decomposition Lemma] \label{lem:decomp}
Let $n \in \nat$. Let $q$ be a \shal that contains a \patt{\Bs[X]}{n}
and let $q = q_0 + c_0 + \dots + c_n + q_{n+1}$ be the associated
leftmost or rightmost decomposition. Then
\[
\Gc_k[X](q) \lesspr  \Gc_{k-n}[X](q_0) + \Gc_{k-n}[X](c_0) + \dots +
\Gc_{k-n}[X](c_n) + \Gc_{k-n}[X](q_{n+1})
\]
\end{lemma}

\begin{proof}
  This is a simple \efgame game argument. Because of the missing boundary labels
  within the $q_j$, using at most $n$ moves, Spoiler can make sure that
  the game stays within the appropriate segment $q_j$ and can use the remaining
  $k-n$ moves for describing that segment.
\end{proof}

This finishes the definition of patterns. Set $n = m(|\bcC| + 1)$. We now
consider two cases depending on whether our \shal $p$ contains a
\patt{\Bs[X]}{2n}.

\medskip
\noindent
{\bf Case 1: $p$ does not contain a \patt{\Bs[X]}{2n}.} In that case
we conclude using induction and Operation~\eqref{op:one}. Let $n'$ be
the largest number such that $p$ contains a \patt{\Bs[X]}{n'}. By
hypothesis $n' < 2n$. Let $p = p_0 + c_0 + \dots + c_{n'}+ p_{n'+1}$
be the associated leftmost decomposition. Observe that by definition,
for $i \leq n'$, $p_i$ uses a strictly smaller alphabet than
$\Bs$. Moreover, since $p$ does not contain a \patt{\Bs[X]}{n'+1} this
is also the case for $p_{n'+1}$. Set $\tilde{k} = k-n'$, by choice of
$k$, we have $\tilde{k} \geq 2(|\Bs|-1)^2(|\bcC|+1)$. Therefore, we
can use our induction hypothesis and for all $i$ we get,
\[
\Gc_{\tilde{k}}[X](p_i) \in \dclos Sat[X,\alpha]
\]
Moreover, for all $i$, $\Gc_{\tilde{k}}[X](c_i) \in \bcT[\alpha]
\subseteq Sat[X,\alpha]$. Finally, using
Lemma~\ref{lem:decomp} we obtain
\[
\Gc_{k}[X](p) \lesspr  \Gc_{\tilde{k}}[X](p_0) + \Gc_{\tilde{k}}[X](c_0) + \dots +
\Gc_{\tilde{k}}[X](c_{n'}) + \Gc_{\tilde{k}}[X](p_{n'+1})
\]
From Operation~\eqref{op:one} the right-hand sum
is in $\dclos Sat[X,\alpha]$. We then conclude that
$\Gc_{k}[X](p) \in \dclos Sat[X,\alpha]$ which terminates
this case.

\medskip
\noindent
{\bf Case 2: $p$ contains a \patt{\Bs[X]}{2n}.} In that case we
conclude using induction, Operation~\eqref{op:one} and
Operation~\eqref{op:two}. By hypothesis, we know that $p$ contains a
\patt{\Bs[X]}{n}, let $p = p_0 + c_0 + \dots + c_{n}+ p_{n+1}$ be the
associated leftmost decomposition. Since $p$ contains a
\patt{\Bs[X]}{2n}, $p_{n+1}$ must contain a \patt{\Bs[X]}{n}. We set
$p_{n+1} = p' + c'_0 + \dots + c'_{n}+ p'_{n+1}$ as the associated
rightmost decomposition. In the end we get
\[
p =  p_0 + c_0 + \dots + c_{n} + p' + c'_0 + p'_1 + \dots + c'_{n}+ p'_{n+1}
\]
Set $\tilde{k} = k - 2n$ and observe that by choice of $k$, $\tilde{k}
\geq 2(|\Bs|-1)^2(|\bcC|+1)$. Therefore, as in the previous
case, we get by induction that for all $i$, $\Gc_{\tilde{k}}[X](p_i)
\in \dclos Sat[X,\alpha]$, $\Gc_{\tilde{k}}[X](p'_i) \in
\dclos Sat[X,\alpha]$, $\Gc_{\tilde{k}}[X](c_i) \in \dclos
Sat[X,\alpha]$ and $\Gc_{\tilde{k}}[X](c'_i) \in \dclos
Sat[X,\alpha]$. Using the same inductive argument for $p'$
may not be possible as $p'$ might contain all labels in \Bs.

If $p'$ does not contain all labels in \Bs, then, by induction,
$\Gc_{\tilde{k}}[X](p') \in \dclos Sat[X,\alpha]$ and we can
then use Lemma~\ref{lem:decomp} as in Case~1 to conclude that
$\Gc_{k}[X](p) \in \dclos Sat[X,\alpha]$. Assume now that
$p'$ contains all labels in $\Bs$. Recall that $m$ is the number of
$\Bs[X]$-equivalence classes. For all $j \leq |\bcC|$, set
\[
\Vc_j = \sum^{m-1+jm}_{i = jm} (\Gc_{\tilde{k}}[X](p_{i}) +
\Gc_{\tilde{k}}[X](c_{i})) \qquad \Vc'_j = \sum^{m-1+jm}_{i = jm} 
(\Gc_{\tilde{k}}[X](c'_{i}) + \Gc_{\tilde{k}}[X](p'_{i+1}))
\]
Observe that for all $j$, by definition $\Vc_j,\Vc'_j$ have an
alphabet which is an $X$-approximation of $\Bs$ and by
Operation~\eqref{op:one}, $\Vc_j,\Vc'_j \in \dclos
Sat[X,\alpha]$. Moreover, it follows from a pigeon-hole
principle argument that the sequences $\Vc_0 + \cdots + \Vc_{|\bcC|}$
and $\Vc'_0 + \cdots + \Vc'_{|\bcC|}$ must contain ``loops'',
i.e. there exists $j_1 < j_2$ and $j'_1 < j'_2$
such that
\[
\begin{array}{rcl}
\Vc_0 + \cdots + \Vc_{j_1} & = & \Vc_0 + \cdots + \Vc_{j_2} \\
\Vc'_{j_2} + \cdots + \Vc'_{|\bcC|} & = & \Vc'_{j_1} + \cdots +
\Vc'_{|\bcC|}
\end{array}
\]
Set $\Uc_1 = \Vc_0 + \cdots + \Vc_{j_1}$, $\Uc_2 = \Vc_{j_1+1} +
\cdots + \Vc_{j_2}$, $\Uc'_1 = \Vc'_{j'_1} + \cdots \Vc'_{j'_2-1}$ and
$\Uc'_2 = \Vc'_{j'_2} + \cdots + \Vc'_{|\bcC|}$. Observe that by
Operation~\eqref{op:one}, we have $\Uc_1,\Uc_2,\Uc'_1,\Uc'_2 \in
\dclos Sat[X,\alpha]$ and that by construction the alphabets of
$\Uc_2,\Uc'_1$ are $X$-approximations of $\Bs$. Moreover, a little
algebra yields $\Uc_1 = \Uc_1 + \Uc_2 = \Uc_1 + \omega \Uc_2$ and
$\Uc'_2 = \Uc'_1 + \Uc'_2 = \omega \Uc'_1 + \Uc'_2$.

Set $p''= p_{j_2m} + \cdots + c_n + p' + c'_0 + \cdots +
p'_{j'_1m-1}$. Observe that by hypothesis on $p'$, $p''$ contains
all labels in $\Bs$. It follows from Fact~\ref{fct:gen1} and
Lemma~\ref{lem:decomp} that
\[
\Gc_{k}[X](p) \lesspr \Uc_1 + \Uc_2 + \Gc_{\tilde{k}}[X](p'')
+ \Uc'_1 + \Uc'_2 = \Uc_1 + \omega \Uc_2 + \Gc_{\tilde{k}}[X](p'')
+ \omega \Uc'_1 + \Uc'_2
\]
Moreover, since $p''$ has alphabet $\Bs$, it is immediate that
$\Gc_{\tilde{k}}[X](p'') \lesspr \uplift{\Bs}$. Therefore, 
\[
\Gc_{k}[X](p) \lesspr \Uc_1 + \omega \Uc_2 + \uplift{\Bs} + \omega \Uc'_1 + \Uc'_2
\]
It is now immediate from Operation~\eqref{op:two} $\omega \Uc_2 + \uplift{\Bs}
+ \omega \Uc'_1 \in \dclos Sat[X,\alpha]$. By combining this
with Operation~\eqref{op:one}, we obtain
\[
\Gc_{k}[X](p) \lesspr \Uc_1 + \omega \Uc_2 +  \uplift{\Bs} + \omega
\Uc'_1 + \Uc'_2 \in \dclos Sat[X,\alpha]
\]
We conclude that $\Gc_{k}[X](p) \in \dclos Sat[X,\alpha]$
which terminates the proof.

\section{Discussion}

We have obtained a characterization for \FOd, using identities on the
syntactic forest algebra and the new notion of saturation.  Our proof
technique applies to many other logical formalisms assuming these only
differ from \FOd by their horizontal expressive power and that they
can at least express the fact that two nodes are siblings.

We have shown all these characterizations to be decidable except for
\mbox{\FOdxs}. We leave this case as an open problem. As explained in
Section~\ref{others}, it would be enough to generalize our algorithm
for computing profiles (i.e. Proposition~\ref{prop:algo}) to the
appropriate notion of profile for \FOdxs.

Since \FOdv is unable to express the sibling relation, it
cannot be covered by our techniques and we leave open the problem of
finding a decidable characterization for this logic.

It would also be interesting to incorporate the vertical successor in
our proofs to obtain a decidable characterization for \FOdx. This
would yield a decidable characterization of the navigational core of
XPath. We believe this requires new ideas.

It terms of complexity, a rough analysis of the proof of
Theorem~\ref{cor-decid} yields a {\sc 4-Exptime} upper bound on the
complexity of the problem. It is likely that this can be
improved. Recall that the complexity of the same problem for the
corresponding logics over words, which amounts to checking
(\ref{eqv}), is polynomial is the size of the syntactic monoid. 

It would also be interesting to obtain an equivalent characterization
of \FOd which remains decidable while avoiding the cumbersome notion
of saturation. For instance it is not clear whether the notion of
confusion introduced in~\cite{BSW09} can be used as a replacement. We
leave this as an open problem.

\paragraph{\bf Acknowledgment} We thanks the reviewers for their comments on
earlier versions of this article. Their comments led to significant
improvements of the paper.

\bibliographystyle{alpha}
\bibliography{main}

\end{document}